\documentclass[12pt]{article}
\usepackage{amsmath}
\usepackage{amsfonts}
\usepackage{amssymb}
\usepackage{amsthm}
\usepackage{graphicx}
\usepackage{fullpage}
\usepackage{setspace}
\usepackage{verbatim}
\usepackage{natbib}
\usepackage{appendix}

\newtheorem{theorem}{Theorem}[section]
\newtheorem{lemma}{Lemma}[section]

\newtheorem{assumption}{Assumption}[section]

\title{Weighted KS Statistics for Inference on Conditional Moment
  Inequalities\footnotetext{First version: March
    2011. This version: October 2011.}}

\author{Timothy B. Armstrong\thanks{email: timothya@stanford.edu.  Thanks
    to Han Hong and Joe Romano for guidance and many useful discussions,
    and to Liran Einav, Azeem Shaikh, Tim Bresnahan, Guido Imbens, Raj
    Chetty, Whitney Newey, Victor Chernozhukov, Jerry Hausman, Andres
    Santos, Elie Tamer, Vicky Zinde-Walsh, Alberto Abadie, Karim Chalak,
    Xu Cheng,
    Konrad Menzel,
    Stefan Hoderlein,
    Don Andrews, Peter Phillips, Taisuke Otsu, Ed Vytlacil, Xiaohong Chen,
    Yuichi Kitamura
    and participants at seminars at Stanford and MIT
    for helpful comments and criticism.  All remaining errors are
    my own.  This paper was written with generous support from a
    fellowship from the endowment in memory of
    B.F. Haley and E.S. Shaw through the Stanford Institute for Economic
    Policy Research.}
\\
Stanford University}

\date{}

\begin{document}

\maketitle

\begin{abstract}

This paper proposes confidence regions for the identified set in
conditional moment inequality models using Kolmogorov-Smirnov
statistics with a truncated inverse variance weighting with increasing
truncation points.  The new weighting differs from those proposed in the
literature in two important ways.
First,
confidence regions based on KS tests with the weighting function I
propose %
converge to the identified set at a faster rate than existing procedures
based on bounded weight functions in a broad class of models.
This
provides a theoretical justification for inverse variance
weighting in this context, and contrasts with analogous results for
conditional moment equalities in which optimal weighting only affects the
asymptotic variance.
Second, the new weighting changes the asymptotic
behavior, including the rate of convergence, of the KS statistic itself,
requiring a new asymptotic theory in choosing the critical value, which I
provide.
To make these comparisons,
I derive rates of convergence for the confidence regions I
propose along with new results for rates of convergence of existing
estimators under a general set of conditions.  A series of examples
illustrates the broad applicability of the conditions.  A monte carlo
study examines the finite sample behavior of the confidence regions.
\end{abstract}

\section{Introduction}

This paper proposes methods for inference in conditional moment inequality
models and derives new relative efficiency results for these models
to show that these methods are more efficient than available methods in a
certain precise sense.
Formally, %
these models are defined by a
restriction of the form $E_P(m(W_i,\theta)|X_i)\ge 0$ almost surely.
Here, $m$ is a known parametric function, which may be vector valued (in
which case the inequality is interpreted as elementwise).  This setup
includes many models commonly used in econometrics, including regression
models with endogenously censored or missing data, selection models, and
certain models of firm and consumer behavior.

The problem is
to perform inference on the identified set
\begin{align*}
\Theta_0(P)\equiv \{\theta|E_P(m(W_i,\theta)|X_i)\ge 0\text{ a.s.}\}
\end{align*}
given a
sample $(X_1,W_1),\ldots,(X_n,W_n)$ from $P$.  This paper proposes
confidence regions $\mathcal{C}_n$ that satisfy
\begin{align}\label{cht_coverage}
\liminf_{n\to\infty}\inf_{P\in\mathcal{P}}
  P(\Theta_0(P)\subseteq\mathcal{C}_n)\ge 1-\alpha
\end{align}
for classes of probability distributions $\mathcal{P}$ restricted only by
mild regularity conditions.  For these confidence regions and several
confidence regions available in the literature satisfying this
requirement, I derive rates of convergence of $\mathcal{C}_n$ to
$\Theta_0(P)$.  The results give sequences $a_n$, which depend on the
smoothness of $\mathcal{P}$ and the method used to construct
$\mathcal{C}_n$, such that
\begin{align}\label{rate_def_eq}
\sup_{P\in\mathcal{P}}P(d_H(\theta_0(P),\mathcal{C}_n)\ge a_n) \to 0.
\end{align}
These results show that, in a general class of models, the confidence
regions proposed here are the only ones to obtain the best rate $a_n$
in (\ref{rate_def_eq}) for a variety of classes $\mathcal{P}$ defined by
different smoothness conditions without prior knowledge of $\mathcal{P}$.
In this sense, the confidence regions proposed here are adaptive.

The confidence regions proposed in this paper are
based on a Kolmogorov-Smirnov (KS) statistic weighted by a truncation of
the inverse of the sample variance
with an increasing sequence of
truncation points.
Following the approach of \citet{chernozhukov_estimation_2007} and
\citet{romano_inference_2010},
the confidence regions invert these tests using critical values that
control the familywise error rate over parameter values in the identified
set, resulting in a set that satisfies (\ref{cht_coverage}).
The increasing sequence of truncation points I propose changes the
asymptotic behavior and, in particular, the rate of convergence of the KS
statistic relative to the bounded weightings proposed in the literature.
This requires a new asymptotic theory in choosing the critical value,
which I develop.
I derive the rate of convergence to the identified set for these
confidence regions
under conditions that apply to a
broad
class of models while still being interpretable.
Since general results for rates of convergence to the identified set have
not been derived for confidence regions based on kernel methods or KS
statistics with bounded weights, %
I derive
rates of convergence for confidence
regions based on these existing approaches as well.
For the class of models
I consider, I find that using the inverse variance with increasing
trunaction points as the weight function in the KS statistic results in a
confidence region for the identified set that has a faster rate of
convergence to the identified set than
the KS statistic based confidence regions with bounded weights proposed in
the literature, and achieves the same rate of convergence as a kernel
estimate with the optimal bandwidth.
For classes of underlying distributions in which smoothness of two
derivatives or less is imposed, these rates correspond with the upper
bounds derived by \citet{stone_optimal_1982} for estimating conditional
means.

To my knowledge, these results provide the first theoretical justification
for weighting moments by their variance in conditional moment inequality
problems.  If
the truncation parameter is allowed to increase fast enough, weighting by
the variance in the KS objective function increases the rate of
convergence of confidence regions to the identified set under the
conditions I consider.
Given that numerous negative results exist for similar problems, it might
be surprising that such general results on relative efficiency
could be obtained.
For one, the
tests procedures I compare are adapted from nonparametric goodness of
fit tests.  The general concensus in this literature is that the relative
efficiency of these tests will depend on the particular situation, and
that, while power results can be obtained for certain types of
alternatives, one cannot make any broad conclusions about which tests are
more powerful.
An important insight of this paper is that, although one
cannot make a general statement about one procedure being optimal against
all possible alternatives in every setting, most conditional moment
inequality models used in practice place restrictions on how parameter
values not in the identified set translate to the conditional moment
restriction being violated.  One of the contributions of this paper is to
propose a set of interpretable conditions under which the truncated
variance weighting proposed in this paper is most efficient, and to show
that several models used in practice satisfy them.

A second reason that relative efficiency in this setting might seem like
an intractable problem is that,
even for the seemingly simpler problem of
inference based on finitely many moment inequalities, no
relative efficiency results in terms of local power comparisons have been
developed.
Indeed, the lack of such results has motivated interest in large
deviations optimality \citep{canay_inference_2010}, which are of
particular interest when local power comparisons do not give a clear
recommendation.
This paper makes progress in a
seemingly more difficult problem by showing that, while power comparisons
in models with unconditional moment inequalities involve subtle issues of
how relative efficiency should be defined for inference on sets, power
comparisons for conditional moment inequalities can be made with the
coarser comparison of rates of convergence to the identified set.  Since
different approaches to inference on the identified set lead to different
rates of convergence to the identified set, %
comparing rates of convergence leads to clear
recommendations of which estimator to use.

Part of the intuition for the efficiency of the inverse variance weighting
proposed in this paper relative to other methods is similar to the
intuition for why weighting by the inverse of the variance matrix in the
GMM objective function improves the asymptotic variance of GMM estimators.
Moments that can be estimated more accurately should be given more
weight.
However, as I describe in more detail in the body of the paper,
the result is also related to the choice of bandwidth in kernel
estimation.  The KS statistics for moment inequality models I consider
take the supremum of an infinite number of unconditional moment
inequalities that together are equivalent to the conditional moment
inequality.  Under the conditions in this paper, local alternatives
violate a sequence of unconditional moments that behave like means of
kernel functions under a decreasing sequence of bandwidths.  Weighting by
the inverse of the variance allows the KS statistic to automatically
choose the unconditional moments that correspond to the optimal bandwidth,
while controlling the probability of type I error even when smoothness
conditions needed for kernel estimation do not hold.

One interpretation
of this result is that inverse variance weighting results in a test that
is adaptive to smoothness conditions on the conditional mean.  Indeed, the
rates of convergence to the identified set derived in this paper coincide
with the optimal rates of convergence for estimates of conditional means
under Lipschitz condition or a bounded second derivative derived in
\citet{stone_optimal_1982}.  The confidence regions proposed in this paper
are also adaptive to Holder conditions and intermediate levels of
smoothness.  Thus, this paper draws a connection between optimal weighting
functions and adaptive estimation.

Another way of describing the intuition for the better rate of convergence
with variance weighting is that it helps alleviate a nonsimilarity problem
with KS statistics applied to conditional moment inequality problems.
As shown by \citet{armstrong_asymptotically_2011}, KS statistics with
bounded weights will converge at different rates on the boundary of the
identified set depending on the shape of the conditional mean.
The results from that paper can be used to improve the power of tests
based on bounded weights, but require pre tests to determing the rate of
convergence of the test statistic.
The weight functions I propose in this paper scale up low variance moments
so that the KS statistic will be of the same order of magnitude whether
the supremum is achieved at a low or high variance moment.  This makes
the procedures proposed in this paper more powerful against sequences of
alternative parameter values that determine rates of convergence in the
Hausdorff metric, leading to a faster rate of convergence
for the confidence region even when a worst case critical value is used.

The results in this paper show that, in certain smoothness classes,
confidence regions based on the methods in this paper achieve the best
rate of convergence to the identified set in the Hausdorff metric.  While
other methods %
achieve the same rate of convergence
if prior information is known about
the shape of the conditional mean,
these methods will do much worse if incorrect prior information is used to
choose a different approach.  A succinct way of putting this is that,
among the approaches considered here, the approach based on inverse
variance weighted KS statistics has the optimal minimax rate for a broad
set of smoothness classes.  %
While minimax definitions of relative efficiency are useful, they ignore
the possibility that, while the inverse variance weighting approach is
better in the worst
case in a particular class of distributions, other approaches might do much
better under more favorable data
generating processes.  However, the results in Section \ref{other_est_sec}
show that, even in a very restrictive set of cases that are more favorable
for the approach based on bounded weights, the inverse variance weighting
proposed in this paper will only lose a $\log n$ term in the rate of
convergence to the identified set relative to the rate of convergence
using bounded weights.  This contrasts with the polynomial differences in
rates of convergence in cases where bounded weights or kernel based
methods do worse.

The sets considered in this paper
are confidence sets in the sense of \citet{chernozhukov_estimation_2007},
since they contain the identified set with a prespecified probability
asymptotically.  One can also interpret these sets as outwardly biased
estimates of the identified set, similar to those proposed by
\citet{chernozhukov_estimation_2007}.  Throughout the paper, I refer to
these sets interchangeably as confidence regions and as estimates of the
identified set.  Interpreting these sets as confidence regions, the rates
of convergence in the Hausdorff metric derived in this paper are a measure
of the power of these tests against local alternatives.  The rates of
convergence derived here imply local power results for sequences of
parameter values that approach the boundary of the identified set.  In
addition to the confidence regions considered here that contain the entire
identified set, methods similar to those used in this paper could be used
to construct confidence regions for points in the identified set, as
proposed by \citet{imbens_confidence_2004}.  Local power results for tests
satisfying this less stringent requirement would follow from similar
arguments.

The new class of weightings proposed in this paper leads to a nontrivial
change in the behavior of the statistic.  Whereas the KS type statistics
considered by \citet{andrews_inference_2009} and
\citet{kim_kyoo_il_set_2008} are defined as the supremum of a random
process that converges to a tight random process,
this does not hold with the increasing truncation
points for the inverse variance weighting used here.  Thus, while the
statistics using bounded weights can be handled using functional central
limit theorems in the supremum norm, such as those in
\citet{vaart_weak_1996}, such results do not apply for the weighting
functions in this paper.  To overcome this, I use maximal inequalities
that bound the supremum of a random process by a function of the maximal
variance of the process.
The asymptotic
bounds on the
sampling distribution of the statistic with the new weighting
follow arguments in \citet{pollard_convergence_1984}, with some slight
modifications to obtain uniformity in the underlying distribution.  A
disadvantage of this approach is that it only leads to an upper bound on
the critical value for the test statistic, leading to conservative
inference.  While this is also the case for many procedures in the moment
inequalities setting, it would be useful to extend these results to derive
less conservative critical values.  On the other hand, the local power
results in this paper show that, even with these conservative critical
values, confidence sets based on the weighting proposed in this paper
converge to the identified set at a faster rate than confidence regions
based on bounded weightings.

This paper relates to the recent literature on econometric models defined
by moment inequalities
and, in particular, conditional moment inequalities where the conditioning
variable is continuously distributed.  \citet{andrews_inference_2009},
\citet{kim_kyoo_il_set_2008}, %
\citet{menzel_estimation_2008,menzel_consistent_2010} and
\citet{chernozhukov_intersection_2009} treat this problem in different
ways.
The estimators of the identified set considered in the present
paper are most similar to those considered by \citet{andrews_inference_2009}
and \citet{kim_kyoo_il_set_2008}, %
the only major difference being the magnitude of a truncation parameter
relative to the sample size.  One of the contributions of this paper is to
show how allowing the truncation parameter to change with the sample size
changes the behavior of the KS statistic in nontrivial ways, and how to
use this to form set estimates that, in a broad class of models, converge
to the identified set at a faster rate.  In addition, the rates of
convergence to the identified set for some of these approaches derived in
the present paper are the first local power results for these
methods that apply generically to conditional moment inequality models in
the set identified case.
These estimators and inference procedures build on the idea of
transforming conditional moment inequalities to unconditional moment
inequalities, which was used by \citet{khan_inference_2009} to propose
estimates for a point identified model.
Their setting differs from most of those considered
here in that their model is point identified with a root-$n$ rate of
convergence for the point estimate.
\citet{galichon_test_2009} propose a similar statistic for a class of
models under a different setup with possible lack of point identification.

More broadly, this paper relates to the literature on set identified
models.  Much of this research has been on models defined by finitely many
unconditional moment inequalities.  Papers that treat this problem include
\citet{andrews_confidence_2004}, %
\citet{andrews_inference_2008}, %
\citet{andrews_validity_2009}, %
\citet{andrews_inference_2010}, %
\citet{chernozhukov_estimation_2007}, %
\citet{romano_inference_2010}, %
\citet{romano_inference_2008}, %
\citet{bugni_bootstrap_2010},
\citet{beresteanu_asymptotic_2008},
\citet{moon_bayesian_2009},
\citet{imbens_confidence_2004}
and \citet{stoye_more_2009}.

The paper is organized as follows.
In Section \ref{setup_sec}, I describe the estimation problem and
estimators of the identified set, and give an
informal description of some of the results in the paper and the intuition
behind them.
In Section \ref{cov_sec}, I state conditions under which the
estimate contains the identified set with probability approaching one.
In Section \ref{rate_sec}, I state conditions for consistency and rates of
convergence.
In Section \ref{applications_sec}, I verify the conditions of Section
\ref{rate_sec} in some examples.
In Section \ref{other_est_sec}, I derive rates of convergence of other
estimators of the identified set and compare them to rates of convergence
for the estimators proposed in this paper.
Section \ref{monte_carlo_sec} reports the results of a monte carlo study
of the finite sample properties of the estimators.
Section \ref{conclusion} concludes, and an appendix contain proofs and
additional results referred to in the body of the paper.

I use the following notation throughout the paper.
For observations
$(X_1,W_1),\ldots,(X_n,W_n)$ and a
measurable function $h$ on the sample space, $E_nh(X_i,W_i)\equiv
\frac{1}{n}\sum_{i=1}^n h(X_i,W_i)$ denotes the
sample mean and $E_P h(X_i,W_i)$ denotes the mean of $h(X_i,W_i)$ under
the probability measure $P$.  The support of a random variable $X_i$ under
a probability measure $P$ is denoted $\text{supp}_P(X_i)$.
I use double
subscripts to denote elements of vector observations so that $X_{i,j}$
denotes the $j$th component of the $i$th observation $X_i$.
For a vector $x\in\mathbb{R}^k$, use the notation $x_{-i}$ to denote the
vector $(x_1,\ldots,x_{i-1},x_{i+1},\ldots,x_k)'$.
Inequalities on Euclidean space refer to the partial ordering of
elementwise inequality.  I use $a\wedge b$ to denote the elementwise
minimum and $a\vee b$ to denote the elementwise maximum of $a$ and $b$.
For a norm $\|\cdot\|$ on $\mathbb{R}^k$, $\|t\|_{-}\equiv \|t\wedge 0\|$.
Unless otherwise noted, $\|\cdot\|$ denotes the Euclidean norm.

\section{Setup and Informal Description of Results}
  \label{setup_sec}

We observe iid observations $(X_1,W_1),\ldots,(X_n,W_n)$ distributed
according to some probability distribution $P\in\mathcal{P}$, and wish to
perform inference on the identified set $\Theta_0(P)$ of parameters
$\theta\in\Theta\subseteq\mathbb{R}^d$ that satisfy the conditional moment
inequalities
\begin{align*}
E_P[m(W_i,\theta)|X_i]\ge 0 \,\,\, \text{$P$-a.s.}
\end{align*}
Here, $X_i$ and $W_i$ are random variables on $\mathbb{R}^{d_X}$ and
$\mathbb{R}^{d_W}$ respectively, and $m:\mathbb{R}^{d_W}\times \Theta\to
\mathbb{R}^{d_Y}$ is a measurable function.
See Section \ref{applications_sec} for examples of econometric models
that fit into this framework.
In what follows, $\bar m(\theta,x,P)$ will denote a version of
$E_P[m(W_i,\theta)|X_i=x]$.

I consider inference on $\Theta_0(P)$ using
a standard deviation weighted KS statistic defined as follows.  Let
$\mathcal{G}$ be a class of functions from $\mathbb{R}^{d_X}$ to
$\mathbb{R}_+$.
Let $\mu_{P,j}(\theta,g)=E_Pm_j(W_i,\theta)g_j(X_i)$ and
$\sigma_{P,j}(\theta,g)
  =\{E_P[m_j(W_i,\theta)g_j(X_i)]^2-[E_Pm_j(W_i,\theta)g_j(X_i)]^2\}^{1/2}$
and define the sample analogues
$\hat\mu_{n,j}(\theta,g)=E_nm_j(W_i,\theta)g_j(X_i)$ and
$\hat\sigma_{n,j}(\theta,g)
  =\{E_n[m_j(W_i,\theta)g_j(X_i)]^2-[E_nm_j(W_i,\theta)g_j(X_i)]^2\}^{1/2}$.
Since the functions in $\mathcal{G}$ are nonnegative,
$E_P[m(W_i,\theta)|X_i=x]\ge 0$ for all $x$ implies that
$\mu_{P,j}(\theta,g)=E_Pm_j(W_i,\theta)g_j(X_i)$  is nonnegative for all
$g$ and $j$.  The
KS statistics in this paper are designed to be positive and large in
magnitude when one of these moments is small (negative and large in
magnitude).
For a fixed function $S:\mathbb{R}^{d_Y}\to\mathbb{R}_+$ chosen by the
researcher, the KS statistic is defined as
\begin{align*}
T_n(\theta)=\sup_{g\in\mathcal{G}}
S\left(\frac{\hat\mu_{n,1}(\theta,g)}{\hat\sigma_{n,1}(\theta,g)\vee \sigma_n}
  ,\ldots,\frac{\hat\mu_{n,d_Y}(\theta,g)}
    {\hat\sigma_{n,d_Y}(\theta,g)\vee \sigma_n}\right)
\end{align*}
where $\sigma_n$ is a sequence of truncation points.
Here, $S$ is a function that is positive and large in magnitude when one
of its arguments is negative and large in magnitude.  Possible choices
include $t\mapsto \|t\|_{-}$ or, more generally, any function that
satisfies Assumption \ref{S_assump}, given in Section \ref{cov_sec}.
If $T_n(\theta)$ is positive and large in magnitude, this is evidence that
$\mu_{P,j}(\theta,g)$ is negative for some $j$ and $g$, so that $\theta$
is not in the identified set.

The set estimates in this paper invert this test statistic using
critical values that control the probability of false rejection
uniformly over $\Theta$, as proposed by
\citet{chernozhukov_estimation_2007}.  For some data dependent
value $\hat c_n$, the confidence region $\mathcal{C}_n(\hat c_n)$ for the
identified set is defined as
\begin{align*}
\mathcal{C}_n(\hat c_n)
\equiv \left\{\theta\in\Theta
  \bigg |\frac{\sqrt{n}}{\sqrt{\log n}}T_n(\theta)\le \hat c_n\right\}.
\end{align*}
Defining the critical relative to the scaling
$\frac{\sqrt{n}}{\sqrt{\log n}}$ anticipates results on the rate of
convergence of $T_n(\theta)$ stated in what follows.

\subsection{Intuition for the Results}\label{intuition_sec}

To describe the intuition behind the results in this paper, consider a
special case of the KS statistic based confidence regions I treat in this
paper applied to a particular model.  Consider an interval regression
model, in which we
posit a linear conditional mean for a latent variable $W_i^*$ given an
observed variable $X_i$,
$E_P(W_i^*|X_i)=\theta_1+X_i'\theta_{-1}$, but
only observe intervals known to contain $W_i^*$.
Here, $X_i$ is a continuously distributed random variable on
$\mathbb{R}^{d_X}$.
While surveys that
elicit interval responses are an obvious application, this encompasses
other forms of incomplete data including selection models and missing data
(see Section \ref{selection_sec} for an example).
I give a more thorough treatment of this model in Sections
\ref{oneside_sec} and \ref{int_reg_sec}.  To keep things simple, suppose
that we only observe a one sided interval containing $W_i^*$.  That is, we
observe a variable $W_i^H$ known to be greater than or equal to $W_i^*$.
Then the problem can be defined formally as estimating or performing
inference on the identified set $\Theta_0(P)$ of values of
$\theta=(\theta_1,\theta_{-1})$ that satisfy
$E_P(W_i^H|X_i)\ge\theta_1+X_i'\theta_{-1}$.

To fix ideas, consider using the KS statistic defined above with the class
of functions $\mathcal{G}$ given by the set of indicator functions
$I(\|X_i-s\|\le h)$ with $s$ ranging over real numbers and $h$ ranging
over nonnegative reals.  The results in this paper allow other classes of
functions for $\mathcal{G}$, including other kernel functions, but this
example captures the main ideas.  For some positive weighting function
$\omega(\theta,s,h)$, define the KS statistic
$T_{n,\omega}(\theta)
  =\sup_{s,h}|\omega(\theta,s,h)
    E_n(W_i^H-\theta_1-X_i'\theta_{-1})I(\|X_i-s\|\le h)|_{-}$
where $|r|_{-}\equiv |r\wedge 0|$.  This corresponds to the KS statistic
defined above with $S(r)=|r|_{-}$ and with the weight function
$\frac{1}{\hat \sigma (\theta,s,h)\vee \sigma_n}$
(here $\hat \sigma(\theta,s,h)
  \equiv \{E_n[(W_i^H-\theta_1-X_i'\theta_{-1})I(\|X_i-s\|\le h)]^2
    -[E_n(W_i^H-\theta_1-X_i'\theta_{-1})I(\|X_i-s\|\le h)]^2\}^{1/2}$)
replaced by an arbitrary
weight function $\omega(\theta,s,h)$.  I derive rates of convergence for
set estimates based on the truncated variance weight function
$\frac{1}{\hat \sigma (\theta,s,h)\vee \sigma_n}$ in Section
\ref{rate_sec}.  In Section \ref{other_est_sec}, I derive rates of
convergence to the identified set for estimators based on KS statistics
with $\omega$ given by a function that is bounded uniformly in the sample
size $n$.  In the remainder of this section, I state these results
informally and describe some of the intuition behind them.

Following \citet{andrews_inference_2009} and \citet{kim_kyoo_il_set_2008},
one can show that $T_{n,\omega}(\theta)$ will converge at a $\sqrt{n}$
rate under regularity conditions if $\omega(\theta,s,h)$ is bounded
uniformly in $n$.
However,
since the variance of the moment indexed
by $(\theta,s,h)$ will be arbitrarily small when $h$ is small ($X_i$ has
a continuous distribution),
setting $\omega(\theta,s,h)$ equal to
$\frac{1}{\hat \sigma (\theta,s,h)\vee \sigma_n}$
gives a weight function that increases without
bound as $\sigma_n$ decreases with the sample size.
This decreases the rate of convergence from $\sqrt{n}$ to
$\sqrt{n/\log n}$ in general.
The estimators of the identified set I
propose in this paper are based on inverting KS tests with this weighting
function, where $\sqrt{n/\log n}T_{n,\omega}(\theta)$ is compared to a
critical value $\hat c_n$ that is bounded or increases slowly.  With a
bounded weight function that does not increase with $n$,
$\sqrt{n}T_{n,\omega}(\theta)$ is compared to a bounded or slowly
increasing critical value.

In this paper, I consider rates of convergence of these confidence regions
to the
identified set.  While power against a fixed sequence of local
alternatives is a bit different than rates of convergence to the
identified set (see the discussion at the end of Section
\ref{oneside_sec}, the conditions in Section \ref{int_reg_sec}, and the
example in Section \ref{slope_examp_sec} of the appendix for some of the
issues that arise in going from sequences of local alternatives to rates
of convergence to the identified set), much of the intuition for the
results in this paper can be exposited in the context of a single sequence
of local alternatives.  Consider a value of $\theta$ such that the
regression line $\theta_1+X_i'\theta_{-1}$ is tangent to the conditional
mean $E_P(W_i^H|X_i)$ at a single point $x_0$, and $X_i$ has a density
bounded away from zero and infinity near $x_0$.  This will typically be the
case at least for some, if not all, elements on the boundary of the
identified set.  The results are the same if $x_0$ is replaced by a
finite set, and can be extended to cases of set identification at infinity
or at a finite boundary in which $x_0$ may be infinite and
the density of
$X_i$ may
go to zero or infinity near $x_0$ by transforming the model (see Section
\ref{selection_sec}).
Suppose that, for some $\alpha>0$,
\begin{align}\label{alpha_cond_intuition}
\text{$E_P(W_i^H-\theta_1-X_i'\theta_{-1}|X_i=x)$ increases like
$\|x-x_0\|^\alpha$}
\end{align}
as $\|x-x_0\|$ increases for $x$ close to $x_0$.  If
$E_P(W_i^H|X_i=x)$ is twice differentiable and $x_0$ is on the interior of
the support of $X_i$, this will hold with $\alpha=2$, and a Lipschitz
condition on $E_P(W_i^H|X_i=x)$ leads to $\alpha=1$.  While other values
of $\alpha$ appear less natural in this context, they are common in
irregularly identified cases such as the selection model considered in
Section \ref{selection_sec}.

Consider the power of KS tests against local alternatives of the form
$\theta_n=(\theta_{1,0}+a_n,\theta_{-1,0})$, where
$\theta_0=(\theta_{1,0},\theta_{-1,0})$ is on the boundary of the
identified set and satisfies the above conditions for some $\alpha$.
Since moments centered at $x_0$ will have more negative expected values
under this sequence of alternatives, the moments with the most power for
detecting this sequence of local alternatives will be those indexed by
$s=x_0$ and some sequence of values of $h$.
For both classes of weight functions, the order of magnitude of the value
of $h$ that indexes the moment with the most power will be determined by a
tradeoff between variance and the magnitude of the expectation.
The 
KS objective function evaluated at some $(\theta,s,h)$ is the sum of a
mean
zero term $(E_n-E_P)(W_i^H-\theta_1-X_i'\theta_{-1})I(\|X_i-s\|\le h)$ and a
drift term $E_P(W_i^H-\theta_1-X_i'\theta_{-1})I(\|X_i-s\|\le h)$.  Under
$(\theta_n,s,h)$ with $s=x_0$, the drift term is
\begin{align}\label{bias_var_intuition}
&E_P(W_i^H-\theta_{1,n}-X_i'\theta_{-1,n})I(\|X_i-x_0\|\le h)
=E_P(W_i^H-\theta_{1,0}-a_n-X_i'\theta_{-1,0})I(\|X_i-x_0\|\le h) \notag  \\
&=E_P(W_i^H-\theta_{1,0}-X_i'\theta_{-1,0})I(\|X_i-x_0\|\le h)
-a_nE_P I(\|X_i-x_0\|\le h).
\end{align}
Some calculation shows that the first term in the above display is of
order $h^{\alpha+d_X}$, while the second term in the above display is
of order $-a_nh^{d_X}$.

Which values of $h$ result in the corresponding
moment having power depends on the mean zero term and the scaling, which
depends on the weight function.
First, consider the increasing sequence of weight functions given by
$\omega(\theta_n,x_0,h)=\frac{1}{\hat \sigma (\theta_n,x_0,h)\vee
  \sigma_n}$.  In
this case, the $\mathcal{O}(h^{\alpha+d_X}-a_nh^{d_X})$ term in the above
display will be divided by $\hat \sigma (\theta_n,x_0,h)\vee \sigma_n$,
which, for $\sigma_n$ small enough, will be approximately equal to the
standard deviation of the moment indexed by $(\theta_n,x_0,h)$, which is
of order $h^{d_X/2}$, and compared to a critical value that is of order
$(n/\log n)^{-1/2}$ (the mean zero term will be of the same order of
magnitude as the normalized critical value, so it will not affect the
power calculation).  Thus, the local alternative indexed by $a_n$ will be
detected if
$\mathcal{O}\left(\frac{h^{\alpha+d_X}-a_nh^{d_X}}{h^{d_X/2}}\right)
  \le -\mathcal{O}(n/\log n)^{-1/2}$ for some $h$.  The left hand side is
minimized when $h$ is equal to a small constant times $a_n^{1/\alpha}$,
which leads to the left hand side being of order
$-a_n^{(d_X+2\alpha)/(2\alpha)}$.  This will be less than the 
$-\mathcal{O}(n/\log n)^{-1/2}$ critical value if $a_n$ is greater than or
equal to a large enough constant times
$(n/\log n)^{-\alpha/(d_X+2\alpha)}$.  An argument that formalizes these
ideas and adapts them to derive rates of convergence to the identified set
rather than power against fixed sequences shows that this is the rate of
convergence of set estimates based on KS statistics with the truncated
inverse variance weight function I propose in this paper under more
general conditions that include this model as a special case.

Now consider using a KS statistic with a bounded weight function.  The
drift term will still be of order $h^{\alpha+d_X}-a_nh^{d_X}$ before being
multiplied by the weight function, but, since the weight function is
bounded uniformly in $n$, weighting will not increase the order of
magnitude of the drift term.  In this case, the KS statistics will be
compared to a critical value of order $n^{-1/2}$, and the mean zero term
will be of a smaller order of magnitude, so that the local alternative
indexed by $a_n$ will be detected if
$\mathcal{O}(h^{\alpha+d_X}-a_nh^{d_X})\le -\mathcal{O}(n^{-1/2})$.  As
before, the left hand side is minimized when $h$ is equal to some small
constant
times $a_n^{1/\alpha}$.  In this case, this leads to the left hand side
being of order $a_n^{(d_X+\alpha)/\alpha}$.  This will be less than the
$-\mathcal{O}(n^{-1/2})$ critical value of $a_n$ is greater than some
large constant times $n^{-\alpha/(2d_X+2\alpha)}$.  This is a slower rate
of convergence than the $(n/\log n)^{-\alpha/(d_X+2\alpha)}$ rate for
estimaters that use the inverse variance weighting with increasing
truncation points.

The increase in power from weighting low variance moments by the inverse
of their standard deviations comes from the fact that local alternatives
violate the conditional moment inequality on a shrinking subset of the
support of the conditioning variable.  If we require that the weight be
bounded uniformly in $n$, low variance moments cannot be weighted properly
because the inverse of the standard deviation will be greater than the
truncation point.  One way of putting this is that the KS statistic
chooses the optimal order of magnitude for the kernel bandwidth by
performing a bias-variance tradeoff automatically, and the variance
scaling makes sure that the correct variance is used in making this
calculation.

\section{Coverage of the Identified Set}%
\label{cov_sec}

In this section, I state conditions under which the
confidence region $\mathcal{C}_n(\hat c_n)$ contains the identified set
$\Theta_0(P)$ with probability approaching one.  Under these
conditions, these estimates control the probability of falsely
concluding that the data are not consistent with some parameter
value.  I show that the probability that the estimate contains the
identified set converges to one uniformly in any class of probability
distributions $\mathcal{P}$ that satisfy a set of assumptions stated
below.  Since these conditions do not restrict the smoothness of the
conditional mean $\bar m(\theta,x,P)$ or the distribution of the
conditioning variable, this shows that the estimator is
robust to many types of data generating processes, at least in the
sense of controlling the probability of type I error.  In contrast,
rates of convergence derived later in the paper depend on additional
smoothness conditions on the data generating process.  Thus, while we
can be reasonably confident rejecting potential parameter values with
this method, the power of the KS statistic based estimates (and the
other set estimators considered in Section \ref{other_est_sec}) will
depend on the shape of the data generating process.

I make the following assumptions.

\begin{assumption}\label{g_pos_assump}
$g_j(X_i)\ge 0$ $P$-a.s. for $j$ from $1$ to $d_Y$ for $g\in\mathcal{G}$
and $P\in\mathcal{P}$.
\end{assumption}

Assumption \ref{g_pos_assump} states that the conditional moment
inequalities are integrated against nonnegative functions, so that going
from conditional moment inequalities to unconditional moment inequalities
does not change the sign of the moment inequalities.

\begin{assumption}\label{covering_assump}
For $j$ from $1$ to $d_Y$, define the classes of functions
$\mathcal{F}_{j,1}=\{s m_j(W_i,\theta)g_j(X_i)+t
  |\theta\in\Theta,g\in\mathcal{G},s,t\in [-(\overline Y\vee 1),
  \overline Y\vee 1]\}$
and
$\mathcal{F}_{j,2}=\{(s m_j(W_i,\theta)g_j(X_i)+t)^2
  |\theta\in\Theta,g\in\mathcal{G},s,t\in [-(\overline Y\vee 1),
  \overline Y\vee 1]\}$
where $\overline Y$ is defined in Assumption \ref{bdd_assump} below.
Suppose that, for $j$ from $1$ to $d_Y$ and $i=1,2$,
$\sup_Q N(\varepsilon,\mathcal{F}_{j,i},L_1(Q))\le A\varepsilon^{-V}$
for $0<\varepsilon<1$ for some $A,V>0$, where the
supremum over $Q$ is
over all probability measures and
$N(\varepsilon,\mathcal{F}_{j,i},L_1(Q))$ is the $L_1$ covering number
defined in \citet{pollard_convergence_1984}.
\end{assumption}

\begin{assumption}\label{bdd_assump}
For some fixed $\overline Y\ge 0$, $|m_j(W_i,\theta)g_j(X_i)|\le \overline
Y$ $P$-a.s. for $j$ from $1$ to $d_Y$ for all $P\in\mathcal{P}$.
\end{assumption}

Assumption \ref{covering_assump} bounds the complexity of the classes of
functions involved so that empirical process methods can be used.  This
condition will hold if the corresponding bounds hold for $\mathcal{G}$ and
$\{w\mapsto m(w,\theta)|\theta\in\Theta\}$ individually.  In Section
\ref{covering_sec} of the appendix, I state sufficient conditions
for Assumption \ref{covering_assump}, and verify them for some classes of
functions $\mathcal{G}$ and the moment functions $m$ from the examples in
Section \ref{applications_sec}.
See \citet{pollard_convergence_1984} or \citet{vaart_weak_1996} for
definitions and additional sufficient conditions for these covering number
bounds.

Assumption \ref{bdd_assump} is natural in many cases, such as models
defined
by quantile restrictions.  In other cases, it restricts some variables to
a finite interval.  While this is clearly stronger than just bounding some
of the moments of $m_j(W_i,\theta)g_j(X_i)$, when combinded with
Assumption \ref{covering_assump}, it leads to rates of convergence that
are uniform in $\theta$ and $g$ and in the underlying distribution with no
additional assumptions on the shape of the conditional mean or variance or
the smoothness of the cdfs of the random variables.

I make the following assumption on the function $S$.  These assumptions
are satisfied by the function $t\to \|t\|_{-}\equiv \|t\wedge 0\|$ for any
norm $\|\cdot\|$ on Euclidean space.

\begin{assumption}\label{S_assump}
$S:\mathbb{R}^{d_Y}\to\mathbb{R}_+$ satisfies
(i) $S(t)>0$ iff. $t_j<0$ for some $j$ and
(ii) for some positive constants $K_{S,1}$ and $K_{S,2}$, we have, for
any $c>0$, $S(t)\ge c\Longrightarrow t_j\le -cK_{S,1} \text{ some $j$}$
and $S(t)\le c\Longrightarrow t_j\ge -cK_{S,2} \text{ all $j$}$.
\end{assumption}

Finally, I make the following assumption on the sequence of cutoff values
for the weighting functions.

\begin{assumption}\label{cutoff_assump}
$\sigma_n$ is bounded from above and for some possibly data dependent
value $\hat a_n$, $\sigma_n\sqrt{n/\log n}\ge \hat a_n$.
\end{assumption}

This assumption will be invoked with additional assumptions on how $\hat
a_n$ is chosen.  In all cases, I will require $\hat a_n$ to be bounded
away from zero, but some of the results will require stronger conditions
on $\hat a_n$.

Under these conditions with $\hat a_n$ and $\hat c_n$ chosen large enough,
the probability of type I error (in the sense of the estimate not
containing the identified set) converges so zero uniformly in
$P\in\mathcal{P}$.  In the following theorem, the constant $K$ that
determines how large $\hat a_n$ and $\hat c_n$ must be could in principle
be calculated as a function of $P$ using the maximal inequalities in the
proof and then estimated.   However, the resulting bounds would be
conservative in most cases.  In practice, it may be more sensible to take
some data dependent value such as
$\sup_{\theta\in\Theta,1\le j\le d_Y}
  \{E_n[m_j(W_i,\theta)-E_nm_j(W_i,\theta)]^2\}^{1/2}$
and multiply it by a sequence going slowly to infinity such as $\log n$ or
$\log \log n$.

\begin{theorem}\label{coverage_thm}
Suppose that Assumptions \ref{g_pos_assump}, \ref{covering_assump},
\ref{bdd_assump}, \ref{S_assump} and \ref{cutoff_assump} hold with $\hat
a_n\ge K$ and $\hat c_n\ge K$ with probability approaching one uniformly
in $P\in\mathcal{P}$.  If $K$ is larger than some constant that depends
only on $V$ and $\overline Y$ in Assumptions \ref{covering_assump} and
\ref{bdd_assump}, then
\begin{align*}
\inf_{P\in\mathcal{P}} P(\Theta_0(P)\subseteq \mathcal{C}_n(\hat c_n))
  \stackrel{n\to\infty}{\to} 1.
\end{align*}
\end{theorem}

The $\sqrt{n/\log n}$ rate of convergence of the KS statistic is slower
than the $\sqrt{n}$ rate of convergence with fixed $\sigma_n$ derived by
\citet{andrews_inference_2009}.  In Section \ref{exact_rate_sec}, I show
that the rate of convergence is strictly slower than $\sqrt{n}$ under
conditions that include many cases of interest.
One might try to conclude from this that the procedures proposed by
\citet{andrews_inference_2009} and \citet{kim_kyoo_il_set_2008} will
suffer from type I error with probability approaching one if the cutoff
for the weight function ($1/\sigma_n$ in the notation of this paper)
increases with the sample size.  While this would be true if the critical
value for these tests were held fixed, the tests proposed in these papers
use estimated critical values that could increase with the sample size if
$\sigma_n$ goes to zero.  If the critical values increase fast enough,
these tests will still be valid, but it is not clear from existing results
whether they do.  Answering this question would require characterizing the
behavior of these critical values for small $\sigma_n$, and comparing them
to rates of convergence for the weighted KS statistic such as those
derived in the present paper.  Such an approach would likely build on the
ideas in this paper as well as \citet{andrews_inference_2009} and
\citet{kim_kyoo_il_set_2008}, using results on the asymptotic behavior of
the KS statistic with increasing weights that build on those derived in
this paper,
and comparing them to new results on the critical values proposed by
\citet{andrews_inference_2009} and \citet{kim_kyoo_il_set_2008} under
increasing weights, which would have to be derived and would likely
require stronger conditions than the ones in this paper.
In any case,
Theorem \ref{coverage_thm} can be used to form
estimates that contain the
identified set with probability one, and choosing
a critical value large enough to satisfy the assumptions of this theorem
will typically not affect the rate of convergence.  This is the approach I
take throughout the rest of the paper.

\section{Consistency and Rates of Convergence}%
\label{rate_sec}

To get consistency and rates of convergence, we need additional
assumptions that lead to $E_Pm(W_i,\theta)g(X_i)$ being large enough for
parameters far from the identified set.  Consistency and rate of
convergence results are stated for the Hausdorff metric on sets.  For a
metric $d$ on $\Theta$, define the Hausdorff distance between $d_H(A,B)$
any two sets $A$ and $B$ by
\begin{align*}
d_H(A,B)=\max \{\sup_{a\in A} \inf_{b\in B} d(a,b),
  \sup_{b\in B} \inf_{a\in A} d(a,b)\}.
\end{align*}
Here, I define $d_H$ to be the Hausdorff distance that arises when $d$ is
defined to be the metric associated with the Euclidean norm.
Note that under the conditions of Theorem \ref{coverage_thm},
$\Theta_0(P)\subseteq \mathcal{C}_n(\hat c_n)$ with probability
approaching one uniformly in $P\in\mathcal{P}$.  When this holds,
$\sup_{b\in \Theta_0(P)} \inf_{a\in \mathcal{C}_n(\hat c_n)} d(a,b)=0$
so that we just need to bound
$\sup_{a\in \mathcal{C}_n(\hat c_n)} \inf_{b\in \Theta_0(P)} d(a,b)$.

\subsection{Consistency}

The following assumption states that for $\theta$ bounded away from the
identified set, some moment $E_Pm_j(W_i,\theta)g_j(X_i)$ is negative and
is bounded away from zero.  This assumption is used to obtain consistency,
and is in general stronger than what would be needed for power against
fixed points in $\Theta\backslash\Theta_0(P)$, since consistency in the
sense of convergence under some metric on sets requires that the power
against fixed alternatives be uniform in alternatives bounded away from
the identified set in this metric.

\begin{assumption}\label{consistency_assump}
For every $\varepsilon>0$, there exists a $\delta>0$ such that, for all
$P\in\mathcal{P}$, $d_H(\theta,\Theta_0(P))>\varepsilon$ implies that
there exists a $g\in\mathcal{G}$ such that
$E_Pm_j(W_i,\theta)g_j(X_i)<-\delta$ for some $j$.
\end{assumption}

\begin{theorem}\label{consistency_thm}
Suppose that Assumption \ref{consistency_assump} and the assumptions of
Theorem \ref{coverage_thm} hold, and that $\sup_{P\in\mathcal{P}}P(\hat
c_n\sqrt{(\log n)/n}>\eta)\to 0$ for all $\eta>0$.  Then, for every
$\varepsilon>0$,
\begin{align*}
\sup_{P\in \mathcal{P}}
P(d_H(\Theta_0(P),\mathcal{C}_n(\hat c_n))>\varepsilon)
\stackrel{n\to\infty}{\to} 0.
\end{align*}
\end{theorem}

\subsection{Rates of Convergence under High Level Conditions}

While the focus of this paper is the interpretable conditions for rates of
convergence of the estimate of the identified set given in Section
\ref{interpretable_cond_sec}, I first present a result using a high level
condition.  The derivations of the rates of convergence in Section
\ref{interpretable_cond_sec} use this result along with additional
arguments relating the variance and expectation of the moments to the
conditions in this section.  The conditions in this section also encompass
the case where local alternatives violate the conditional moment
inequality on a non-shrinking set, leading to $\sqrt{n/\log n}$
convergence (such as Assumption \ref{ident_pos_assump} for the application
in Section \ref{selection_sec}), and it is instructive to compare the
verification of the conditions in this section under these two types of
set identification.

The next assumption is a high level assumption that
incorporates both the variance and expectation of the moments defined by
each $g\in\mathcal{G}$.  The assumption is similar to the polynomial
minorant condition in \citet{chernozhukov_estimation_2007}.

\begin{assumption}\label{rate_assump}
For some positive constants $C$, $\psi$, $\gamma$, and $\delta$ with
$\psi\le 1$, we have,
(i) for all $P\in\mathcal{P}$ and $\theta\in\Theta$ with
$d_H(\theta,\Theta_0(P))\le \delta$,
\begin{align*}
\inf_{g,j}
\frac{\mu_{P,j}(\theta,g)}
  {\sigma_{P,j}(\theta,g)\vee d_H(\theta,\Theta_0(P))^{\psi/\gamma}}
\le -C d_H(\theta,\Theta_0(P))^{1/\gamma}
\end{align*}
where the infemum is taken over $g\in\mathcal{G}$ and
$j\in\{1,\ldots,d_Y\}$ %
and (ii) $\sigma_n(n/\log n)^{\psi/2}$ is bounded uniformly in $P$.
\end{assumption}

Part (ii) of this assumption states that the cutoff $\sigma_n$ must go to
zero fast enough that the moments with the most identifying power relative
to their variance are scaled by their standard deviation.
How small $\sigma_{P,j}(\theta,g)$ can be in the assumption is determined
by how fast $\sigma_n$ goes to zero.  If the assumption holds with
$\psi$ small so that the infimum in the display is achieved when
$\sigma_{P,j}(\theta,g)$ is large relative to the distance from the
identified set, $\sigma_n$ can be chosen to go to zero more slowly.  If
part (i) holds for any $\psi$, it will hold for $\psi=1$, so that
choosing $\sigma_n$ so that part (ii) holds for $\psi=1$ will lead to
the assumption holding in a larger set of cases when the researcher is
unsure which $g$ functions have the most power.  In the cases considered
here, this will
not affect the rate of convergence, but will have a negative effect
on the tradeoff between power and size when considering power against
local alternatives at a particular rate.
In other words, part (i) of Assumption \ref{rate_assump} is weakest when
$\psi=1$, so, since $\sigma_n$ can always be chosen to go to zero at a
$[(\log n)/n]^{1/2}$ rate so that part (ii) holds with $\psi=1$, the
researcher can just choose $\sigma_n$ this way to have the rate of
convergence given in the next theorem hold under the weakest possible
conditions.

The following theorem gives rates of convergence to the identified set
under this assumption.

\begin{theorem}\label{rate_thm}
Suppose that Assumption \ref{consistency_assump} and \ref{rate_assump}
hold, and that Assumptions \ref{g_pos_assump}, \ref{covering_assump},
\ref{bdd_assump}, \ref{S_assump} and \ref{cutoff_assump} hold with
$\hat a_n$ and $\hat c_n$ chosen to satisfy the requirements of
Theorems \ref{coverage_thm} and \ref{consistency_thm}.  Then, for some
large $B$ that does not depend on $P$,
\begin{align*}
\sup_{P\in\mathcal{P}}
  P\left( \left(\frac{n}{\hat c_n^2\log n}\right)^{\gamma/2}
    d_H(\mathcal{C}_n(\hat c_n),\Theta_0(P)) > B\right)
\stackrel{n\to\infty}{\to} 0.
\end{align*}
\end{theorem}

The results in the next section use Theorem \ref{rate_thm} along with
additional arguments to formalize the intuition described in Section
\ref{intuition_sec}.  The balancing of the mean and variance described in
Section \ref{intuition_sec} plays out through the ratio of the mean
$\mu_{P,j}(\theta,g)$ and the standard deviation $\sigma_{P,j}(\theta,g)$
in Assumption \ref{rate_assump}.  This determines the best attainable
value of $\gamma$ in Assumption \ref{rate_assump}.  If a sequence of $g$
functions can be found such that, as the distance of $\theta$ to the
identified set decreases, the magnitude of $\mu_{P,j}(\theta,g)$ decreases
much more slowly than $\sigma_{P,j}(\theta,g)$, the left hand side of the
display in
Assumption \ref{rate_assump} will be large in magnitude, so that the
condition will hold with a larger value of $\gamma$.  It is useful to
contrast this with the case where local alternatives violate one of the
conditional moment inequalities on a non-shrinking set.  In this case, $g$
can be chosen to be some fixed function that is positive only on this
set.  This leads to $\sigma_{P,j}(\theta,g)$ being fixed while
$\mu_{P,j}(\theta,g)$ typically goes to zero at a rate proportional to
$d_H(\theta,\Theta_0(P))$, so that Assumption \ref{rate_assump} holds with
$\gamma=1$, and Theorem \ref{rate_thm} gives a $\sqrt{n/\log n}$ rate of
convergence for the set estimator (see the proof of the part of Theorem
\ref{rate_thm_selection} that applies under Assumption
\ref{ident_pos_assump} for more details).  In cases like those described
in Section \ref{intuition_sec}, the best attainable ratio of
$\mu_{P,j}(\theta,g)$ to $\sigma_{P,j}(\theta,g)$ depends on smoothness
properties of the data generating process and
leads to a smaller $\gamma$ and a slower rate of convergence.  The results
in the next section cover this case.

\subsection{Interpretable Conditions for Rates of Convergence}
  \label{interpretable_cond_sec}

Assumption \ref{rate_assump} is a high level condition that incorporates
both the expectation and variance of each $g$ function.  The next
assumptions place restrictions on the shape of the conditional mean
$\bar m(\theta,x,P)=E_P(m(X_i,\theta)|X_i=x)$ as a function of $x$ and
$\theta$ that can be used to verify Assumption \ref{rate_assump}.
These conditions shed light on how the shape of the data generating
process and $\bar m(\theta,x,P)$ as a function of $\theta$ and $x$
determine the rate of convergence, and are easier to verify in many
applications.  Once consistency is established, these assumptions only
need to hold for $d_H(\theta,\Theta_0(P))<\varepsilon$ for some
$\varepsilon>0$.

\begin{assumption}\label{diff_assump}
$\bar m(\theta,x,P)$ is differentiable in $\theta$ with derivative
$\bar m_\theta(\theta,x,P)$ that is continuous as a function of $\theta$
uniformly in $(\theta,x,P)$
\end{assumption}

\begin{assumption}\label{x0_assump}
For some $\eta>0$ and $C>0$, we have, for all
$\theta\in\Theta\backslash\Theta_0(P)$, there exists a
$j_0(\theta,P)$, $\theta_0(\theta,P)$ and $x_0(\theta,P)$ such that
\begin{align*}
\bar m_{\theta,j_0(\theta,P)}(\theta_0(\theta,P),x_0(\theta,P),P)
  (\theta-\theta_0(\theta,P))
\le -\eta \|\theta-\theta_0(\theta,P)\|,
\end{align*}
$\bar m_{j_0(\theta,P)}(\theta_0(\theta,P),x_0(\theta,P),P)=0$, and, for
$\|x-x_0\|<\eta$,
\begin{align*}
|\bar m_{j_0(\theta,P)}(\theta_0(\theta,P),x,P)
-\bar m_{j_0(\theta,P)}(\theta_0(\theta,P),x_0(\theta,P),P)|
\le C\|x-x_0(\theta,P)\|^\alpha.
\end{align*}
\end{assumption}

The first part of Assumption \ref{x0_assump} states that, for $\theta$
close to
the identified set, there is some element in the identified set such that
that moving from this element to $\theta$ corresponds to moving some index
of the conditional mean downward.  This assumption restricts the angle
between the path from $\theta$ to some point on the identified set and the
directional derivative of the conditional mean for $\theta$ along this
path.  To see that the first part of Assumption \ref{x0_assump} comes from
a condition on the magnitude of the derivative of the conditional mean
with respect to $\theta$ and the angle of between the derivative and the
difference between $\theta$ and some point on the identified set, note
that, letting $\phi$ be the angle between $\bar
m_{\theta,j_0(\theta,P)}(\theta_0(\theta,P),x_0(\theta,P),P)$ and
$\theta-\theta_0(\theta,P)$,
\begin{align*}
&\bar m_{\theta,j_0(\theta,P)}(\theta_0(\theta,P),x_0(\theta,P),P)
  (\theta-\theta_0(\theta,P))  \\
&=\|\bar m_{\theta,j_0(\theta,P)}(\theta_0(\theta,P),x_0(\theta,P),P)\|
  \|\theta-\theta_0(\theta,P)\| \cos \phi.
\end{align*}
Thus, the first part of Assumption \ref{x0_assump} will be satisfied if
$\|\bar m_{\theta,j_0(\theta,P)}(\theta_0(\theta,P),x_0(\theta,P),P)\|$ is
bounded away from zero and $\cos \phi$ is negative and bounded away from
zero.

The second part of Assumption \ref{x0_assump} is a restriction on the
shape of the conditional mean as a function of $x$ for $\theta$ on the
boundary of the identified set.  Combining this with the first part of the
assumption determines which functions in $\mathcal{G}$ have power under
local alternatives.
As verified for several models in Section
\ref{applications_sec}, this typically follows from Holder conditions or
conditions on the first two derivatives of conditional means or quantiles
of variables in the data, leading to some value of $\alpha$ between zero
and $2$, or from conditions on densities and conditional means near the
boundary of the support of the conditioning variable, which can lead to
larger values of $\alpha$ after a transformation of the data.

To better understand how Assumption \ref{x0_assump} factors into the rate
of convergence, it is helpful to relate it to the discussion in Section
\ref{intuition_sec} giving an informal overview of the results for the
interval regression model.  The
interested reader can consult the proofs of the results in Sections
\ref{oneside_sec} and \ref{int_reg_sec} for more details.  The second part
of Assumption \ref{x0_assump} is the condition described in
(\ref{alpha_cond_intuition}).  The first part of Assumption
\ref{x0_assump} relates to the choice of local alternative used in Section
\ref{intuition_sec}.  In that section, we fixed a parameter
$\theta_0=(\theta_{1,0},\theta_{-1,0})$ on the boundary of the identified
set, and considered local alternatives of the form
$\theta_n=(\theta_{1,0}+a_n,\theta_{-1,0})$ for some positive sequence
$a_n\to 0$.  This leads to the characterization of the drift term of the
KS objective function in (\ref{bias_var_intuition}).  The same argument
goes through for most types of local alternatives that also vary the
slope, but certain types of local alternatives have to be ruled out.  In
the interval regression example, these correspond to local alternatives
that rotate the regression line around a single tangency point.  For
example, in the example in Section \ref{intuition_sec}, suppose $d_X=1$,
and $x_0=0$.  %
If we instead took a
sequence of local alternatives of the form $\theta_n'=(0,a_n)$, the last
line in (\ref{bias_var_intuition}) would instead be
\begin{align*}
E_P(W_i^H-\theta_{1,0}-X_i \theta_{-1,0})I(\|X_i-x_0\|\le h)
-a_nE_P(X_i-x_0) I(\|X_i-x_0\|\le h).
\end{align*}
Going through the rest of the argument with $a_nE_PI(\|X_i-x_0\|\le h)$
replaced by $a_nE_P(X_i-x_0) I(\|X_i-x_0\|\le h)$ gives a slower rate of
convergence
because the latter term goes to zero more quickly as $h$
decreases
(see Section \ref{slope_examp_sec} for a more detailed
treatment of this counterexample).

The first part of Assumption \ref{x0_assump} ensures that these types of
sequences of local alternatives do not determine the rate of convergence.
To see how this works, note that, applying
the left hand side of the first display of Assumption \ref{x0_assump} to
the interval regression example %
gives
$\bar m_{\theta}(\theta_0,x_0,P)
  (\theta-\theta_0)
=-(1,x_0)(\theta-\theta_0)$.
Thus, in order for Assumption \ref{x0_assump} to hold for some $\theta$
and this value of $\theta_0$, $(1,x_0)(\theta-\theta_0)$ must be
positive and have the same order of magnitude as $\|\theta-\theta_0\|$.
For $\theta_n$ in the above example, this is
$(1,x_0)(\theta_n-\theta_0)=(1,x_0)(a_n,0)'=a_n=\|\theta_n-\theta_0\|$, so
the first display of Assumption \ref{x0_assump} holds.  For the example
with $\theta_n'=(0,a_n)$ (and $x_0=0$)  %
$(1,x_0)(\theta_n'-\theta_0)=(1,0)(\theta_n'-\theta_0)=(1,0)(0,a_n)=0$, so
the first display of Assumption \ref{x0_assump} does not hold.

The next assumption states that, for any $P\in\mathcal{P}$, all points
must either be outside of the support of $X_i$ under $P$, or have
sufficient probability mass nearby.
While this assumption rules out $X_i$
having infinite support or having a density that goes to zero near the
boundary of its support, these cases can typically be handled by
transforming the data to make this assumption hold.  I do this for one
application in Section \ref{selection_sec}.

\begin{assumption}\label{support_assump}
For some $\eta>0$, we have, for all
$P\in\mathcal{P}$ and all $\varepsilon>0$,
$P(\|X_i-x\|\le \varepsilon)/\varepsilon^{d_X}\ge \eta$
for all $x$ on the support of $X_i$.
\end{assumption}

The next assumption ensures that the set of functions $\mathcal{G}$ is
rich enough to contain functions that behave like indicators of small
sets.  This assumption
holds for any class that contains indicator
sets of the open balls for any norm on $\mathbb{R}^{d_X}$, or, for any
nonnegative bounded kernel function $k:\mathbb{R}^{d_X}\to\mathbb{R}_+$
with finite support and $k(x)$ bounded away from zero near $x=0$, the
class
$\{x\mapsto k((x-t)/h)|t\in\mathbb{R}, h\ge 0\}$
that contains all dilations and translations of the kernel function $k$.

\begin{assumption}\label{g_rate_assump}
The functions in $\mathcal{G}$ are uniformly bounded
and for some constants $0<C_{\mathcal{G},1}<1$ and
$0<C_{\mathcal{G},2}<1$, we have that, for all $s\in\mathbb{R}^{d_X}$ and
$t\ge 0$, $\mathcal{G}$ contains a function $g$ such that
$C_{\mathcal{G},1}I(\|X_i-s\|<C_{\mathcal{G},2}t)
\le g(X_i) \le I(\|X_i-s\|<t)$.
\end{assumption}

The next theorem gives rates of convergence under these assumptions.

\begin{theorem}\label{rate_thm_alpha}
Suppose that Assumptions %
\ref{diff_assump}, \ref{x0_assump}, \ref{support_assump} and
\ref{g_rate_assump} hold.
Then part (i) of Assumption
\ref{rate_assump} holds with $\gamma=2\alpha/(d_X+2\alpha)$ and
$\psi=d_X/(d_X+2\alpha)$.
\end{theorem}

Applying Theorem \ref{rate_thm}, this gives a
$(n/\log n)^{\alpha/(d_X+2\alpha)}$ rate of convergence as long as the
cutoff point $\sigma_n$ for the standard deviation weighting decreases at
least as quickly as $((\log n)/n)^{\psi/2}=(n/\log n)^{d_X/(2d_X+4\alpha)}$,
but slightly more slowly than $((\log n)/n)^{1/2}$, so that Assumption
\ref{cutoff_assump} will hold with $\hat a_n$ large enough.  One choice of
$\sigma_n$
that will work regardless of $\alpha$ is to take some data dependent value
like $\sup_{\theta\in\Theta,1\le j\le d_Y}
  \{E_n[m(W_i,\theta)-E_nm(W_i,\theta)]^2\}^{1/2}$
and multiply by $((\log n)/n)^{1/2}b_n$, where $b_n$ is a sequence that goes
to infinity more slowly than any power of $n$ (such as $b_n=\log n$).

\section{Applications}\label{applications_sec}

In this section, I verify the conditions for rates of convergence stated
above for some applications under primitive conditions.
I start with a one sided regression model.

\subsection{One Sided Regression}\label{oneside_sec}

We posit a linear regression model $E_P(W_i^*|X_i)=X_i'\beta$ for a latent
variable $W_i^*$, but we only observe $(X_i,W_i^H)$, where $W_i^H$ is
known to be greater than or equal to $W_i^*$.  This leads to the
conditional moment inequality $E_P(W_i^H|X_i)\ge X_i'\beta$, which fits into
the framework of this paper with $d_Y=1$, $W_i=(X_i,W_i^H)$ and
$m(W_i,\theta)=W_i^H-\theta_1-X_i'\theta_{-1}$.  Here,
$\bar m(\theta,x)=E_P(W_i^H|X_i=x)-\theta_1-x'\theta_{-1}$.  I verify the
conditions used above to derive rates of convergence (Assumptions
\ref{diff_assump} and \ref{x0_assump})
under the following assumptions.

\begin{assumption}\label{reg_holder_assump}
For some $C>0$ and $0<\alpha\le 1$,
$\|E_P(W_i^H|X_i=x)-E_P(W_i^H|X_i=x')\|\le C\|x-x'\|^\alpha$ for $x$ and
$x'$ on the support of $X_i$ for all $P\in\mathcal{P}$.
\end{assumption}

Assumption \ref{reg_holder_assump} places a Holder condition on the
conditional mean of the upper bound of the outcome given $X_i$.  This is a
smoothness condition on the data generating process.  For $\alpha=1$,
Assumption \ref{reg_holder_assump} states that this conditional mean must
be Lipschitz continuous.  For smaller $\alpha$, the conditional mean must
still be continuous, but can be less smooth.

For $\alpha>1$, a condition like Assumption
\ref{reg_holder_assump} would restrict $E_P(W_i^H|X_i=x)$ to be
constant, since its slope would have to converge to zero at every point.
However, as described in Section \ref{intuition_sec}, this condition
factors into the rate of convergence only in restricting
$E_P(W_i^H-\theta_1-X_i'\theta_{-1}|X_i=x)$ to increase
no faster than a multiple of $\|x-x_0\|^\alpha$ near some tangency point
$x_0$ for $\theta=(\theta_1,\theta_{-1})$ on the boundary of the
identified set.
The same argument will still go through as long as this restriction on the
difference between $E_P(W_i^H|X_i=x)$ and a tangent line holds
for some $\alpha$, even if $\alpha>1$.  While placing this condition
directly on $E_P(W_i^H-\theta_1-X_i'\theta_{-1}|X_i=x)$ near tangency
points is a bit awkward in general, this condition has a natural
interpretation when $\alpha=2$.  In this case, it requires that the
difference between the conditional mean $E_P(W_i^H|X_i=x)$ and any tangent
line behave quadratically near the tangent point, which is implied by a
bound on the second derivative.  This is the content of the next
assumption.

\begin{assumption}\label{reg_2diff_assump}
(i) $E_P(W_i^H|X_i=x)$ has a second derivative that is bounded uniformly in
$P$ and $x$ and
(ii) for any $P\in\mathcal{P}$, $\theta_0\in\Theta_0(P)$,
$E_P(W_i^H|X_i=x)$ is bounded away from $\theta_{0,1}+x'\theta_{0,-1}$ on
the boundary of the support of $X_i$
\end{assumption}

The next assumption ensures that the condition on the tangent angle in
Assumption \ref{x0_assump} holds.  Under this assumption, rates of
convergence to the identified set depend on sequences of parameters in
which only the intercept parameter varies.  This condition ensures that
varying the intercept parameter a small amount near the boundary of the
identified set gives an element that is still in the parameter space
$\Theta$.

\begin{assumption}\label{reg_param_assump}
The subvector $\theta_{-1}$ of $\theta$ is bounded over $\theta\in\Theta$
and, for any $\theta\in\Theta$, $(\theta_1',\theta_{-1})\in \Theta$ for
all $\theta_1'\in\mathbb{R}$.
\end{assumption}

\begin{theorem}\label{oneside_thm}
Suppose that Assumptions \ref{reg_param_assump} holds in the one sided
linear regression model and $X_i$ has compact support for all
$P\in\mathcal{P}$.  Then, if Assumption \ref{reg_holder_assump} holds,
Assumptions \ref{diff_assump} and \ref{x0_assump} will hold for $\alpha$
specified in Assumption \ref{reg_holder_assump}.  If Assumption
\ref{reg_2diff_assump} holds, Assumptions \ref{diff_assump} and
\ref{x0_assump} will hold for $\alpha=2$.
\end{theorem}

If the parameter space $\Theta$ is restricted so that all sequences of
local alternatives corresponded to rotating the regression line around a
tangent point, Assumption \ref{reg_param_assump} will fail and the rate
of convergence will be slower.  The verification of the
assumptions of Theorem \ref{rate_thm_alpha} will not go through in this
case because
the first part of Assumption \ref{x0_assump} will fail.  As an example,
suppose $E_P(W_i^H|X_i=x)=x^2$.  If the parameter space $\Theta$ does not
restrict the intercept parameter, the proof of Theorem \ref{oneside_thm}
will go through.  However, if
$\Theta=\{(0,\theta_1)|\theta_1\in\mathbb{R}\}$ (that is, we restrict the
intercept to be $0$), the rate of convergence will be determined by
local alternatives of the form $(0,a_n)$.  This corresponds to the sequence
of local alternatives $\theta_n'$ in the discussion in Section
\ref{interpretable_cond_sec}.  For the same reasons described in that
section, the first part of Assumption \ref{x0_assump} will not hold,
leading to a slower rate of convergence.  I show in Section
\ref{slope_examp_sec} of the appendix that the estimate of the identified
set converges no faster than at
a $((\log n)/n)^{1/5}$ rate, rather than the $((\log n)/n)^{2/5}$ rate for
the case where the parameter space is unrestricted.

These issues also make it more difficult
to state primitive conditions that lead to
Assumption \ref{x0_assump} in the %
case of two sided interval regression, in which we add the conditional
moment inequality $m_2(W_i,\theta)=\theta_{1}+X_i'\theta_{-1}-W_i^L$.
As with restricting the parameter space, adding the second conditional
moment inequality can
lead to the rate of convergence being deterimined by sequences of local
alternatives that correspond to rotating the regression line around a
tangent point.
One example that leads to this is when
$E_P(W_i^H|X_i=x)=x^2$ and $E_P(W_i^L|X_i=x)=-x^2$.
Adding the moment inequality on $W_i^L$ has the same effect as restricting
the intercept to be zero in the example above.  The rate of convergence to
the identified set is determined by local alternatives of the form
$(0,a_n)$, which leads to a slower rate of convergence.  The argument in
Section \ref{slope_examp_sec} applies here as well, leading to a slower
$((\log n)/n)^{1/5}$ rate of convergence.

For the case where $X_i$ is a scalar, these cases can be ruled out in the
interval regression model by
requiring that the conditional means of $W_i^H$ and $W_i^L$ be bounded
away from each other.  I go through this argument in the next section.
However, higher dimensions appear to require further conditions.

\subsection{Interval Regression with a Scalar Regressor}\label{int_reg_sec}

In the case of a single regressor, these types of slow convergence of a
slope parameter in the interval regression model can be ruled out by
relatively simple conditions.  In what follows, I consider an interval
regression model in which, in addition to $W_i^H$ defined as in Section
\ref{oneside_sec}, we observe $W_i^L$ that is known to satisfy $W_i^L\le
W_i^*$, so that $E_P(W_i^L|X_i)\le \theta_1+X_i'\theta_{-1}$.  This fits
into the framework of this paper with
$m(W_i,\theta)=(W_i^H-\theta_1-X_i'\theta_{-1}
  ,\theta_1+X_i'\theta_{-1}-W_i^L)$.
I restrict attention to the case where $d_X=1$, so that
$\theta_{-1}=\theta_2$ is a scalar.

In addition to the assumptions used in Section \ref{oneside_sec}, I impose
the following assumption, which rules out cases like the one described
above in which local alternatives correspond to rotating the regression
line around a tangent point.

\begin{assumption}\label{int_reg_assump}
(i) The support of $X_i$ is bounded uniformly in $P\in\mathcal{P}$.
(ii) The absolute value of the slope parameter $\theta_2$ is bounded
uniformly on the identified sets $\Theta_0(P)$ of $P\in\mathcal{P}$.
(iii) $E_P(W_i^H|X_i=x)-E_P(W_i^L|X_i=x)$ is bounded away from zero
uniformly in $x$ and $\mathcal{P}$.
\end{assumption}

\begin{theorem}\label{int_reg_thm}
In the interval regression model with $d_X=1$, suppose that Assumption
\ref{int_reg_assump} holds.
Then, if Assumption \ref{reg_holder_assump} holds as stated and with
$W_i^H$ replaced by $W_i^L$,
Assumptions \ref{diff_assump} and \ref{x0_assump} will hold for $\alpha$
specified in Assumption \ref{reg_holder_assump} (and $d_X=1$).  If
Assumption \ref{reg_2diff_assump} holds as stated and with $W_i^H$
replaced with $W_i^L$, Assumptions \ref{diff_assump} and \ref{x0_assump}
will hold for $\alpha=2$ (and $d_X=1$).
\end{theorem}

\subsection{One Sided Quantile Regression}\label{quantile_oneside_sec}

In this and the next section, I treat quantile versions of the regression
models considered above.  Here, we have a model for a conditional
quantile of the unobserved variable $W_i^*$ rather than the mean.  The
results are essentially the same, but, in addition to smoothness
conditions on the quantile itself, conditions are needed on the joint
density of the observed variables near the conditional quantile to
translate these into the conditions on $\bar m(\theta,x,P)$.

First, consider the one sided case in which
we observe $(X_i,W_i^H)$ with $W_i^H\ge W_i^*$.
For a random variable $Z_i$, define
$q_{\tau,P}(Z_i|X_i)$ to be the $\tau$th quantile of $Z_i$ conditional on
$X_i$ under $P$. Suppose that, for some known $\tau$, the conditional
$\tau$th quantile of $W_i^*$ satisfies
$q_{\tau,P}(W_i^*|X_i)=\theta_1+X_i'\theta_{-1}$ for some $\theta$.  Then
$E_P[\tau-I(W_i^*\le \theta_1+X_i'\theta_{-1})|X_i]=0$ so that
$E_P[\tau-I(W_i^H\le \theta_1+X_i'\theta_{-1})|X_i]\ge 0$.  Thus, this
fits into the framework of this paper with $W_i=(X_i,W_i^H)$ and
$m(W_i,\theta)=\tau-I(W_i^H\le \theta_1+X_i'\theta_{-1})$.

In many situations, models for quantiles of an outcome variable given
covariates can be more informative under interval data than models for the
conditional mean.  If $W_H$ can be infinite with positive probability
conditional on any value of $X_i$, the identified set for a conditional
mean model will be the entire parameter space.  If $W_H$ has a low
probability of being large or infinite, and is usually close to $W_i^*$,
a model for conditional quantiles of the unobserved variable will still
give informative bounds with interval data.

Smoothness conditions that lead to Assumptions \ref{diff_assump} and
\ref{x0_assump} for the quantile model are similar to those for the
conditional mean considered above, but with smoothness assumptions placed
on the conditional quantile $q_{\tau,P}(W_i^H|X_i)$ rather than the
conditional mean, and additional assumptions on the joint density of
$(X_i,W_i^H)$.  The first two assumptions are exactly the same as
Assumptions \ref{reg_holder_assump} and \ref{reg_2diff_assump}, but with
the conditional mean replaced by the conditional $\tau$th quantile.

\begin{assumption}\label{quantile_holder_assump}
For some $C>0$ and $\alpha\le 1$,
$\|q_{\tau,P}(W_i^H|X_i=x)-q_{\tau,P}(W_i^H|X_i=x')\|\le C\|x-x'\|^\alpha$
for $x$ and $x'$ on the support of $X_i$ for all $P\in\mathcal{P}$.
\end{assumption}

\begin{assumption}\label{quantile_2diff_assump}
(i) $q_{\tau,P}(W_i^H|X_i=x)$ has a second derivative that is bounded
uniformly in $P$ and $x$ and
(ii) for any $P\in\mathcal{P}$, $\theta_0\in\Theta_0(P)$,
$q_{\tau,P}(W_i^H|X_i=x)$ is bounded away from
$\theta_{0,1}+x'\theta_{0,-1}$ on the boundary of the support of $X_i$.
\end{assumption}

The next assumption states that $W_i^H$ has a density near its $\tau$th
quantile conditional on $X_i$.  One type of interval data that will
frequently lead to this assumption holding is if $(X_i,W_i^*)$ has a well
behaved joint density, and $W_i^H$ is equal to $W_i^*$ with high
probability and much larger than $W_i^*$ with some small probability.  For
example, suppose that $(X_i,W_i^*)$ has a joint density, and, $W_i^H$ is
either equal to $\infty$ or $W_i^*$, with $P(W_i^H=\infty|X_i=x,W_i^*=w)$
a smooth function of $(x,w)$ that is bounded from above by some constant
strictly less than $1-\tau$.  Then $(X_i,W_i^H)$ will have a joint density
near the $\tau$th conditional quantile of $W_i^H$.  This type of situation
arises naturally with missing data on an outcome variable.
However, other types of interval data will not lead to this assumption
holding.  If $W_i^H$ is the upper end of an interval from a survey in
which $W_i^*$ is always reported in the same interval, $W_i^H$ will not
have a density conditional on $X_i$.

\begin{assumption}\label{quantile_density_assump}
For some $\eta>0$,
$W_i^H|X_i$ has a conditional density $f_{W_i^H|X_i}(w|x)$ on
$\{(x,w)|q_{\tau,P}(W_i^H|X_i=x)-\eta\le w\le
q_{\tau,P}(W_i^H|X_i=x)+\eta\}$ that is continuous as a function of $w$
uniformly in $(w,x,P)$
and satisfies
$\underline f\le f_{W_i^H|X_i}(w|x)  \le \overline f$ for some
$0<\underline f<\overline f<\infty$.
\end{assumption}

Under these conditions, Assumptions \ref{diff_assump} and \ref{x0_assump}
will hold for the one sided quantile regression model.  The proof is
similar to the proof of Theorem \ref{oneside_thm} in the one sided
regression model.  The only difference is that some additional steps are
needed to translate smoothness conditions on the $\tau$th quantile into
smoothness conditions on the objective function using the assumptions on
the conditional density of $W_i^H$ given $X_i$.

\begin{theorem}\label{quantile_oneside_thm}
Suppose that the support of $X_i$ is bounded uniformly in
$P\in\mathcal{P}$, and that
Assumptions \ref{reg_param_assump} and
\ref{quantile_density_assump} hold in the one sided quantile regression
model.
Then, if Assumption \ref{quantile_holder_assump} holds,
Assumptions \ref{diff_assump} and \ref{x0_assump} will hold for $\alpha$
specified in Assumption \ref{quantile_holder_assump}.  If Assumption
\ref{quantile_2diff_assump} holds,
Assumptions \ref{diff_assump} and
\ref{x0_assump} will hold for $\alpha=2$.
\end{theorem}

\subsection{Interval Quantile Regression with a Scalar Regressor}
  \label{quantile_int_reg_sec}

Now consider a quantile regression model with two sided interval data in
which, in addition to $W_i^H$, we observe a variable $W_i^L$ that is known
to satsify $W_i^L\le W_i^*$.  This leads to
$E_P[I(W_i^L\le \theta_1+X_i'\theta_{-1})-\tau|X_i]
\ge E_P[I(W_i^*\le \theta_1+X_i'\theta_{-1})-\tau|X_i]
= 0$ so that the interval quantile regression fits into the conditional
moment inequality framework with $W_i=(X_i,W_i^L,W_i^H)$ and
$m(W_i,\theta)=(\tau-I(W_i^H\le \theta_1+X_i'\theta_{-1})
  ,I(W_i^L\le \theta_1+X_i'\theta_{-1})-\tau)$.

As with the case of mean regression, the condition on the angle of the
derivative and path in Assumption \ref{x0_assump} will not hold in general
in the quantile regression model with two sided interval data because of
cases where alternatives are closest to a point in the identified set
where the regression line is rotated around a contact point.  Sufficient
conditions to rule this out in the case of a scalar regressor are similar
as well.  Bounding the conditional quantiles of the upper and lower
endpoints of the interval away from each other rules out these cases when
the regressors include only a constant and a scalar.  The next assumption
is the same as Assumption \ref{int_reg_assump}, but with conditional
expectations replaced by conditional $\tau$th quantiles.

\begin{assumption}\label{quantile_int_reg_assump}
(i) The support of $X_i$ is bounded uniformly in $P\in\mathcal{P}$.
(ii) The absolute value of the slope parameter $\theta_2$ is bounded
uniformly on the identified sets $\Theta_0(P)$ of $P\in\mathcal{P}$.
(iii) $q_{\tau,P}(W_i^H|X_i=x)-q_{\tau,P}(W_i^L|X_i=x)$ is bounded away
from zero uniformly in $x$ and $\mathcal{P}$.
\end{assumption}

The next theorem states that KS statistic based set estimators will have
the same rate of convergence as in the one sided model with a scalar
regressor under these conditions, and the assumption stated earlier on the
density of the observed variables.  The proof is similar to the proof of
the analogous result for mean regression, Theorem \ref{int_reg_thm}, but
with additional steps to translate conditions on quantiles and densities
into conditions on the conditional mean of the objective function.

\begin{theorem}\label{quantile_int_reg_thm}
In the interval regression example with $d_X=1$, suppose that Assumptions
\ref{quantile_density_assump} and \ref{quantile_int_reg_assump} hold, and
that Assumption \ref{quantile_density_assump} also holds with $W_i^H$
replaced by $W_i^L$.
Then, if Assumption \ref{quantile_holder_assump} holds as stated and with
$W_i^H$ replaced by $W_i^L$,
Assumptions \ref{diff_assump} and \ref{x0_assump} will hold for $\alpha$
specified in Assumption \ref{quantile_holder_assump} (and $d_X=1$).  If
Assumption \ref{quantile_2diff_assump} holds as stated and with $W_i^H$
replaced with $W_i^L$, Assumptions \ref{diff_assump} and \ref{x0_assump}
will hold for $\alpha=2$ (and $d_X=1$).
\end{theorem}

\subsection{Selection Model and Identification at the Boundary}
  \label{selection_sec}

In this section, I treat a class of models in which the conditional moment
inequalities give the most identifying information when conditioning on a
set where $X_i$ may not have a density that is bounded away from zero and
infinity.  That is, as $\theta$ approaches the identified set, the moment
inequality $E_P(m(W_i,\theta)|X_i=x)\ge 0$ is violated on a region in
which the density of $X_i$ goes to zero or infinity, or in which $X_i$
does not have a density with respect to the Lebesgue measure.
This covers cases of conditional moment inequalities leading to
point or set identification at infinity or at a finite boundary.  While I
motivate the conditions in this section with a selection model, the
results apply more generally to other cases of set identification at the
boundary.

The selection model is
particularly interesting in that it leads naturally to different shapes of
the conditional mean of $m(W_i,\theta)$ and distribution of $X_i$, since
set identification at the boundary of the support of $X_i$ appears to be a
common case.  For cases where the conditioning variable has a density
function that goes to zero or infinity near a (possibly infinite) support
point, a transformation of the conditioning variable leads to a model for
which the smoothness assumptions for rates of convergence given in this
paper can be verified.  The resulting value of the Holder constant
$\alpha$ depends on the shape of both the density and the conditional
mean.

This is related to cases of point identification at infinity, such
as the estimator proposed by \citet{andrews_semiparametric_1998} for a
selection model similar to the one treated in this section, but under
conditions that lead to point identification.  As with the estimator
proposed in that paper, the estimators I consider based on KS statistics
for conditional moment inequalities and possible set identification have
rates of convergence that depend on the tail behavior of the random
variables in the model.  The behavior of distributions of random variables
at the tails determines which functions in $\mathcal{G}$ correspond
to the region of the tail of the conditioning variable with the most
identifying power.  The truncated variance weighting I propose allows the
KS statistic to automatically find these functions.

We are interested in the marginal distribution of a random variable
$Y_i^*$, but we do not always observe this variable.  Instead, we observe
$(Y_i,D_i)$ where $D_i$ is an indicator for being observed in the sample
and $Y_i\equiv Y_i^*\cdot D_i$.
For example, suppose we are interested in the distribution of wage offers
for a population of individuals, but we only observe wages of people who
decide to work.
In this case, $Y_i^*$ is the wage individual $i$ is offered, and $D_i$ is
an indicator for employment.
In what follows, $Y_i$ and $D_i$ are scalars, but the results described
below can be extended to multiple partially observed outcomes.  In the
treatment effects literature, potential outcomes under different treatment
programs are typically modeled as latent variables, with the observed
variable being the actual treatment.  In this case, we can consider each
possible treatment separately, each time defining $Y_i^*$ and $D_i$ to be
potential outcomes and indicators for the treatment group in question.
Bounds on the marginal distribution for each treatment will follow from
methods described in this section, and these bounds can be combined to
give bounds on treatment effects defined as differences between
statistics of the unobserved distribution of each outcome.

If $Y_i^*$ is not independent of $D_i$ and $D_i=0$ with positive
probability, the distribution of $Y_i$ will be different from the
distribution of $Y_i^*$ conditional on entry.  However, it is often
possible to obtain informative bounds.  Suppose that we observe a random
variable $X_i$ that shifts participation in the sample, but is exogenous
to outcomes in the sense that $Y_i^*$ is independent of $X_i$.  If $Y_i$
is known to lie in some interval $[\underline Y,\overline Y]$, we can
bound the distribution of $Y_i^*$ following
\citet{manski_nonparametric_1990}.  In this section, I consider estimation
of bounds for the mean of the distribution of $Y_i^*$, but bounds on
quantiles can be estimated using similar methods.  For the same reasons as
those described in Section \ref{quantile_oneside_sec}, bounds on quantiles
will often be tighter than bounds on the mean when the difference
between $\underline Y$ and $\overline Y$ is large or infinite.

To see how this model fits into the framework of this paper, note that
$Y_i\cdot D_i+\underline Y\cdot (1-D_i)
 \le Y_i^* \le Y_i\cdot D_i+\overline Y\cdot (1-D_i)$, so that,
letting $\gamma=E_P(Y_i^*)=E_P(Y_i^*|X)$, we have
$E_P(Y_i\cdot D_i+\underline Y\cdot (1-D_i)|X)\le \gamma
  \le E_P(Y_i\cdot D_i+\overline Y\cdot (1-D_i)|X)$.
Define $W_i^L=Y_i\cdot D_i+\underline Y\cdot (1-D_i)$ and
$W_i^H=Y_i\cdot D_i+\overline Y\cdot (1-D_i)$.  The problem of
estimating the identified set for $\gamma$ fits into the framework of this
paper with $W_i=(W_i^L,W_i^H,X_i)$ and
$m(W_i,\gamma)=(\gamma-W_i^L,W_i^H-\gamma)'$.

Typically, the best upper and lower bounds on $\gamma$ will come from
values of $X_i$ for which the probability of participation is high.  If
participation is monotonic, these points will be near the support of
$X_i$.  The support of $X_i$ could be infinite or finite, and there is
typically no reason to impose any conditions on how the distribution of
$X_i$ behaves near its support points (whether it has a density, whether
the density approaches zero, infinity, a positive constant, or oscillates
wildly) or how $E_P(W_i^H|X_i)$ and $E_P(W_i^L|X_i)$ behave near these
points.  In addition, while identification at the boundary of the support
seems likely, it is best not to impose this either.

The results in this section show that estimates of the identified set
using weighted KS statistics defined above are robust to all of these
types of set identification in the sense of controlling the probability
that the set estimate fails to contain the identified set uniformly in a
set of underlying distributions that contains these types of distributions
and many more.  In addition, for a wide variety of shapes of the density
and conditional mean, the weighted KS statistic based set estimate obtains
a better rate of convergence than estimates that do not weight the KS
statistic.

Uniform coverage of the identified set follows immediately from Theorem
\ref{coverage_thm}, and is stated in the next theorem.  Throughout this
section, $\Theta_0(P)$ denotes the
identified set for $\gamma$ in the selection model under $P$, and
$\mathcal{C}_n(\hat c_n)$ denotes an estimate of this set as described
above.

\begin{theorem}\label{coverage_thm_selection}
Let $\mathcal{P}$ be any class of probability measures on the random
variables in the selection model described above such that $W_i^H$ and
$W_i^L$ are bounded uniformly over $P\in\mathcal{P}$.  If the class of
functions $\mathcal{G}$, the function $S$, and the sequences $\hat a_n$
and $\hat c_n$ are chosen so that Assumptions \ref{g_pos_assump},
\ref{covering_assump}, \ref{S_assump} and \ref{cutoff_assump} hold with
$\hat a_n$ and $\hat c_n$ chosen so that the assumptions of Theorem
\ref{coverage_thm} hold, then
\begin{align*}
\inf_{P\in\mathcal{P}} P(\Theta_0(P)\subseteq \mathcal{C}_n(\hat c_n))
  \stackrel{n\to\infty}{\to} 1.
\end{align*}
\end{theorem}

Rates of convergence to the identified set will depend on the shape of the
conditional mean and the distribution of $X_i$.  Note, however, that the
set estimate based on the standard deviation weighted KS statistic can be
calculated in the same manner regardless of these aspects of the data, so
the researcher does not have to impose any restrictions on the shapes of
these objects when performing inference.  In this sense, inference based
on these statistics adapts to the shapes of the conditional means of
$W_i^H$ and $W_i^L$ and the distribution of $X_i$.  In what follows, I
consider several alternative assumptions.  These include different types
of set identification at the boundary, as well as set identification on a
positive probability set.

In the following assumptions, $[\underline \gamma,\overline \gamma]$ is the
identified set for $\gamma$, so that it is implicitly assumed that
$E_P(W_i^H|X_i)\ge \overline \gamma$ and $E_P(W_i^L|X_i)\le \underline \gamma$
with probability one.  Here, $\underline \gamma$ and $\overline \gamma$
could be equal, leading to point identification.  This will be the case
when the probability of selection into the sample conditional on $X_i=x$
converges to one as $x$ approaches some point on the support of $X_i$.
These assumptions are stated so that the same type
of identification holds for the upper and lower support of the identified
set, but the same results will hold (with possibly different rates of
convergence to the upper and lower support points) if different types of
identification hold for the upper and lower support.
When these assumptions are invoked for a class of
probability distributions $\mathcal{P}$, the constants $C$, $K_X$, and
$\eta_X$ are assumed not to depend on $P$.

\begin{assumption}[Set Identification at Infinity with Polynomial Tails]
  \label{ident_inf_assump}
$d_X=1$ and, for some positive constants $K_X$ and $C$,
we have, for all $x\ge K_X$,
(i) $E_P(W_i^H|X_i=x)-\overline \gamma
      \le C x^{-\phi_m}$ and
(ii) $X_i$ has a density $f_X(x)$ such that $f_X(x)\ge x^{-\phi_x}/C$
for some $\phi_m>0$ and $\phi_x>1$.
In addition, part (i) holds with $W_i^H-\overline \gamma$ replaced by
$\underline \gamma -W_i^L$.
\end{assumption}

\begin{assumption}[Set Identification at Finite Support with Polynomial
  Tails] \label{ident_fin_assump}
For some $x_0\in\mathbb{R}^{d_X}$ and $\eta_X>0$, we have, for
$x_0-\eta_X\iota\le x\le x_0$ (where $\iota$ is a vector of ones and $\le$
is elementwise if $d_X>1$)
(i) $E_P(W_i^H|X_i=x)-\overline \gamma
      \le C |x_0-x|^{\phi_m}$ and
(ii) $X_i$ has a density $f_X(x)$ such that
$f_X(x)\ge \prod_{k=1}^{d_X} |x_{0,k}-x_k|^{\phi_x}/C$
for some $\phi_m>0$ and some $\phi_x>-1$.
In addition, parts (i) and (ii) hold with $W_i^H-\overline \gamma$ replaced
by $\underline \gamma -W_i^L$ for some possibly different $x_0$.
\end{assumption}

\begin{assumption}[Set Identification on a Positive Probability Set]
  \label{ident_pos_assump}
For some interval $[\underline x,\overline x]$, $E_P(W_i^H|X_i)-\overline
\gamma=0$ $P$-a.s. for all $P\in\mathcal{P}$
and $P(\underline x\le X_i\le \overline x)$ is bounded away from zero
uniformly in $P\in\mathcal{P}$.
In addition, the same assumption holds with with $W_i^H-\overline \gamma$
replaced by $\underline \gamma -W_i^L$ for some possibly different interval
$[\underline x,\overline x]$.
\end{assumption}

All cases of Assumption \ref{ident_inf_assump} and \ref{ident_fin_assump}
can be transformed into Assumption \ref{ident_fin_assump} with $\phi_x=0$
and some $\phi_m$ by monotonic transformations of each element of $X_i$.
The case where Assumption \ref{ident_fin_assump} holds with $\phi_x=0$
fits into the framework of Theorem \ref{rate_thm_alpha}, so this can be
applied to the transformed model.

\begin{theorem}\label{rate_thm_selection}
Let $\mathcal{P}$ be any class of probability measures on the random
variables in the selection model described above such that $W_i^H$ and
$W_i^L$ are bounded uniformly over $P\in\mathcal{P}$.  Suppose that the
class of functions $\mathcal{G}$, the function $S$, and the sequences
$\hat a_n$ and $\hat c_n$ are chosen so that Assumptions
\ref{g_pos_assump}, \ref{covering_assump}, \ref{S_assump} and
\ref{cutoff_assump} hold with $\hat a_n$ and $\hat c_n$ chosen so that the
assumptions of Theorem \ref{coverage_thm} hold, and Assumption
\ref{g_rate_assump} holds.

If, in addition to these conditions, one of Assumptions
\ref{ident_inf_assump} or \ref{ident_fin_assump}
holds, then, for some $B$,
\begin{align*}
\sup_{P\in\mathcal{P}}
  P\left( \left(\frac{n}{\hat c_n^2\log n}\right)^{\alpha/(d_X+2\alpha)}
    d_H(\mathcal{C}_n(\hat c_n),\Theta_0(P)) > B\right)
\stackrel{n\to\infty}{\to} 0
\end{align*}
where $\alpha=\phi_m/(\phi_x+1)$ if Assumption \ref{ident_fin_assump}
holds and $\alpha=\phi_m/(\phi_x-1)$ (and $d_X=1$) if Assumption
\ref{ident_inf_assump} holds.  If Assumption \ref{ident_pos_assump}
holds, then, for some $B$,
\begin{align*}
\sup_{P\in\mathcal{P}}
  P\left( \left(\frac{n}{\hat c_n^2\log n}\right)^{1/2}
    d_H(\mathcal{C}_n(\hat c_n),\Theta_0(P)) > B\right)
\stackrel{n\to\infty}{\to} 0.
\end{align*}

\end{theorem}

The rate of convergence in Theorem \ref{rate_thm_selection} shows that,
for a given selection process conditional on $X_i$, the rate of
convergence will be faster when $X_i$ has more mass near the point $x_0$
or region $[\underline x,\overline x]$
where the conditional moment inequalities give the most identifiying
information.  The rate of convergence is fastest ($((\log n)/n)^{1/2}$)
under Assumption \ref{ident_pos_assump}, when this region has a positive
probability. Under identification at a finite point (Assumption
\ref{ident_fin_assump}), the rate of convergence depends on whether the
density of $X_i$ approaches infinity, zero, or a finite nonzero value.  If
$-1<\phi_x<0$, the density will approach infinity at a rate that is faster
when $\phi_x$ is closest to $-1$ ($\phi_x$ must be strictly greater than
$-1$ in order for the density to integrate to a finite number).  For
$\phi_x=0$, the density approaches a finite nonzero value, and, for
$\phi_x>0$ the density approaches zero at a rate that is faster for larger
values of $\phi_x$.  The rate of convergence under Assumption
\ref{ident_fin_assump} will always be slower than $((\log n)/n)^{1/2}$,
but it will be arbitrarily close to this rate when $\phi_x$ is close to
$-1$ (when the density approaches infinity at close to the fastest
possible rate).  Under identification at infinity (Assumption
\ref{ident_inf_assump}), the rate of convergence will be faster for
thicker tails (smaller $\phi_x$), and will be close to $((\log
n)/n)^{1/2}$ for $\phi_x$ close to $1$ (in this case, $\phi_x$ must be
greater than one in order for the density to integrate to a finite
number).

\section{Rates of Convergence for Other Estimators}\label{other_est_sec}

In order to compare the estimators based on KS statistics with increasing
variance weights proposed in this paper to estimation procedures based on
kernels or KS statistics with bounded weights, we need rates of
convergence for these estimators as well.  Since these results are not
available in the literature
(with the exception of the results of \citet{andrews_inference_2009} and
\citet{kim_kyoo_il_set_2008} for the local power of KS statistics with
bounded weights, which apply to the model in Section \ref{selection_sec}
under the positive probability set identification condition, Assumption
\ref{ident_pos_assump}, but not the other models or conditions in this
paper),
I derive these results in this section.

Under upper bounds on the
smoothness of the data generating process that correspond to the lower
bounds in Assumptions \ref{diff_assump} and \ref{x0_assump}, I show that
estimators
based on KS statistics with bounded weight functions converge at a
$n^{\alpha/(2d_X+2\alpha)}$ rate, slower than the
$(n/\log n)^{\alpha/(d_X+2\alpha)}$ rate of convergence derived in Section
\ref{rate_sec} for the estimator based on the truncated variance weighting
with the sequence of truncation points increasing quickly enough.  %
\citet{kim_kyoo_il_set_2008} shows that the rate of convergence of a
similar estimator will be $n^{1/2}$ under conditions similar to Assumption
\ref{ident_pos_assump} in which local alternatives violate the conditional
moment inequality on a positive probability set.  In these situations, the
increasing sequence of weights for the KS statistic proposed in this paper
will lead to a $(n/\log n)^{1/2}$ rate of convergence for the set
estimate.  For estimators of the identified set based on kernel estimates
of the conditional mean, if the sequence of bandwidth parameters is chosen
properly, I show that the set estimate will converge at the same
$(n/\log n)^{\alpha/(d_X+2\alpha)}$ rate as the variance weighted KS
statistic based estimates, but the rate of convergence can be much slower
if the bandwidth is chosen suboptimally.  However, with the optimal
sequence of bandwidths, power against local alternatives that approach the
identified set at this rate will likely be greater for kernel based
estimates.  Thus, the results in this section show that the weighted KS
statistic based estimates proposed in this paper do almost as well as
an infeasible procedure that uses prior knowledge of the data generating
process to choose the best from a set of other estimators.

While the results in this section show that the truncated variance
weighting allows KS statistic based estimates to adapt to a broad class of
smoothness conditions, these statistics will not achieve the optimal rate
of convergence when more than two derivatives are imposed on the
conditional mean (although the results in Section \ref{selection_sec} show
that KS statistics with the weighting in this paper also adapt to a broad
class of tail behavior in cases of set identification at the boundary).
The reason is that the KS statistics considered in this
paper integrate the conditional mean against nonnegative functions, which
prevents them from taking advantage of higher order smoothness
conditions.  Estimation methods based on higher order kernels or sieves
would likely perform better in some of these situations, although some of
these methods would fail to control the size of these tests when these
smoothness conditions fail.

\subsection{Bounded Weight Functions}

Consider a set estimate based on a KS statistic similar to the ones
considered so far, but with the weight function $1/(\hat
\sigma(\theta,g)\vee \sigma_n)$ replaced by some bounded weight function
$\omega_n(\theta,g)=(\omega_{n,1}(\theta,g),\ldots,\omega_{n,d_Y}(\theta,g))$.
Here, $\omega_n(\theta,g)$ is unrestricted, except for the requirement
that, for some $\overline \omega$ we have
$\|\omega_n(\theta,g)\|\le \overline \omega$ for all $n$, $\theta$, and
$g$.  Define
\begin{align*}
T_{n,\omega}(\theta)\equiv
\sup_{g\in\mathcal{G}}
S\left(\omega_{n,1}(\theta,g)\hat\mu_{n,1}(\theta,g)
  ,\ldots,\omega_{n,d_Y}(\theta,g)\hat\mu_{n,d_Y}(\theta,g)\right).
\end{align*}
Following \citet{andrews_inference_2009} (with additional conditions to
control the complexity of $m_j(W_i,\theta)g_j(X_i)$ over $\theta$ as well
as $g$), $T_{n,\omega}(\theta)$ will converge at a $\sqrt{n}$ rate, so
define the estimate of the identified set for critical value $\hat c_n$ to
be
\begin{align*}
\mathcal{C}_{n,\omega}(\hat c_n)
\equiv \left\{\theta\in\Theta
  \bigg |\sqrt{n}T_{n,\omega}(\theta)\le \hat c_n\right\}.
\end{align*}

Under upper bounds on the smoothness of the conditional mean that
correspond to the lower bounds given in Section \ref{rate_sec}, upper
bounds on the rate of convergence of set estimates based on KS statistics
with bounded weights can be derived.  These conditions are stated in the
following assumption.

\begin{assumption}\label{smoothness_upper_assump}
For some $\theta_0\in \delta\Theta_0(P)$ such that $\theta_0$ is in the
interior of $\Theta$, the following holds for %
some neighborhood $B(\theta_0)$ of $\theta_0$.
(i) $\bar m(\theta,x,P)$ is differentiable in $\theta$ with derivative
$\bar m_{\theta}(\theta,x,P)$ bounded over $\theta\in B(\theta_0)$.
(ii) For some $\eta>0$, we have, for all $\theta_0'\in
(\delta\Theta_0(P))\cap B(\theta_0)$, the set $\mathcal{X}_0(\theta_0')$
of points $x_0$ such that
$\min_k \bar m_k(\theta_0',x_0,P)=0$
satisfies
\begin{align*}
|\bar m_{j}(\theta_0',x,P)
-\bar m_{j}(\theta_0',x_0,P)|
\ge \eta\left(\|x-x_0\|^\alpha\wedge \eta\right),
\end{align*}
for all $j$,
and the number of elements in $\mathcal{X}_0(\theta_0')$ is bounded
uniformly over $\theta_0'$.
(iii) $X_i$ has finite support and a bounded density on its support.
(iv) There exists a path $t\mapsto \theta_t$ such that
$\theta_t\to\theta_0$ as $t\to 0$ and $t\to d(\theta_t,\theta_0)$ is
continuous for $t$ in a neighborhood of $0$.
\end{assumption}

Assumption \ref{smoothness_upper_assump} gives an upper bound on the
smoothness of the conditional mean similar to the lower bound of
Assumption \ref{x0_assump}.  It states that $\alpha$ is the best
(greatest) possible value of $\alpha$ for which Assumption \ref{x0_assump}
can hold.  Without this assumption, rates of convergence derived using
Assumption \ref{x0_assump} and some value of $\alpha$ could be
conservative, since the same assumption could also hold with a larger
value of $\alpha$.  The next theorem uses this condition to get an upper
bound on the rate of convergence of the set estimator
$\mathcal{C}_{n,\omega}(\hat c_n)$ when the sequence of weight functions
is uniformly bounded.

\begin{theorem}\label{upper_rate_bdd_thm}
Under Assumptions \ref{g_pos_assump}, \ref{covering_assump},
\ref{bdd_assump}, \ref{S_assump}
and \ref{smoothness_upper_assump}, if $\hat c_n$ is bounded away from zero
and $g(X_i)$ and $m(W_i,\theta)$ are uniformly bounded, then, for some
$\varepsilon>0$,
\begin{align*}
P\left(
n^{\alpha/(2d_X+2\alpha)}
d_H\left( \mathcal{C}_{n,\omega}(\hat c_n), \Theta_0(P) \right)
\ge \varepsilon
\right)
\stackrel{n\to\infty}{\to} 1.
\end{align*}
\end{theorem}

Under the smoothness conditions of Section \ref{rate_sec}, this slower
rate of convergence can be achieved (up to an arbitrarily slow rate of
growth of the critical value) using bounded weights with an estimated set
that contains $\Theta_0(P)$ with probability approaching one.

\begin{theorem}\label{rate_thm_bddweight}
Suppose that Assumptions 
\ref{g_pos_assump}, \ref{covering_assump}, \ref{bdd_assump},
\ref{S_assump}, \ref{consistency_assump},
\ref{diff_assump}, \ref{x0_assump}, \ref{support_assump} and
\ref{g_rate_assump} hold.  Let the weight function $\omega_n(\theta,g)$
satsify $\underline \omega\le \omega_n(\theta,g)\le \overline \omega$ for
some $0<\underline \omega \le \overline \omega <\infty$, and suppose that
$\hat c_n\to\infty$ with $\hat c_n/\sqrt{n}\to 0$.  Then
\begin{align*}
\inf_{P\in\mathcal{P}} P(\Theta_0(P)
  \subseteq \mathcal{C}_{n,\omega}(\hat c_n))
  \stackrel{n\to\infty}{\to} 1
\end{align*}
and, for $B$ large enough,
\begin{align*}
\sup_{P\in\mathcal{P}}
  P\left( \left(n/\hat c_n^2\right)^{\alpha/(2d_X+2\alpha)}
    d_H(\mathcal{C}_n(\hat c_n),\Theta_0(P)) > B\right)
\stackrel{n\to\infty}{\to} 0.
\end{align*}
\end{theorem}

The $n^{\alpha/(2d_X+2\alpha)}$ rate of convergence for the estimator
using bounded weights is slower than the
$(n/\log n)^{\alpha/(d_X+2\alpha)}$ rate of convergence derived in Section
\ref{rate_sec} for the estimator using the truncated variance weights.
The rate of convergence is slower because sequences of local alternatives
violate a shrinking set of moment inequalities.  This leads to sequences
of functions in $\mathcal{G}$ with the most power having a shrinking
sequence of variances, so that a bounded weighting function cannot give
them enough weight.  While the examples in Section \ref{applications_sec}
show that this case is likely to be common in practice, bounded weight
functions will have advantages in other cases.
Under conditions such as Assumption
\ref{ident_pos_assump} for the selection model in Section
\ref{selection_sec}, sequences of local alternatives lead to
a single function
in $\mathcal{G}$ with
positive variance
having
power.  In this case, using a bounded sequence of weight functions does
not cause such a problem, and the increasing sequence of truncation points
does worse by a power of $\log n$ because of the larger critical value
needed for the KS statistic.

\subsection{Kernel Methods}

Suppose that we estimate the conditional mean
$E_P(m_j(W_i,\theta)|X_i=x)=\bar m_j(\theta,x,P)$ using the kernel estimate
\begin{align*}
\hat{\bar m}_j(\theta,x)
\equiv \frac{E_n m_j(W_i,\theta)k((X_i-x)/h_n)}{E_nk((X_i-x)/h_n)}
\end{align*}
for some sequence $h_n\to 0$.  \citet{chernozhukov_intersection_2009} and
\citet{ponomareva_inference_2010} propose methods for inference on
conditional moment inequalities based on this estimate of the
conditional mean.
Following \citet{chernozhukov_intersection_2009} %
this estimate of the conditional mean will
converge at a $\sqrt{nh^{d_X}/\log n}$ rate uniformly over $x$.  Using the
results in this paper, this rate can be shown to be uniform over $\theta$
as well, so that the statistic
\begin{align*}
T^{\text{kern}}_{n,k,h_n}(\theta)
  \equiv \sup_{x\in\text{supp}_P(X_i)} S(\hat{\bar m}(\theta,x))
\end{align*}
can be used to form an estimate
\begin{align*}
\mathcal{C}_n^{\text{kern}}(\hat c_n)
  \equiv\left\{\theta\in \Theta
    \bigg|\frac{\sqrt{nh^{d_X}}}{\sqrt{\log n}}
      T^{\text{kern}}_{n,k,h_n}(\theta)\le \hat c_n\right\}
\end{align*}
that will contain the identified set with probability approaching one for
$\hat c_n$ large enough.

I place the following conditions on the choice of kernel function $k$.
All of these conditions are fairly mild regularity conditions, except for
the requirement that $k$ be positive, which rules out higher order
kernels.  Ruling out higher order kernels is important.  Since the class
of KS statistics used in this paper integrate the conditional moment
inequality against positive functions, these statistics cannot take
advantage of smoothness conditions of more than two derivatives, while
higher order kernels with a properly chosen bandwidth can.

\begin{assumption}\label{kernel_assump}
(i) $k$ is nonnegative
(ii) $k$ integrates to one, is bounded and square integrable over
$\mathbb{R}^{d_X}$ and $k(t)$ is bounded away from zero for $t$ in
some neighborhood of $0$
(iii) Assumption \ref{covering_assump} holds with $\mathcal{G}$
replaced by the class of functions $t\mapsto k((t-x)/h)$ where $x$ and $h$
vary.
\end{assumption}

As with set estimators based on KS statistics with bounded weights, the
upper bounds on the smoothness of the conditional mean in Assumption
\ref{smoothness_upper_assump} lead to upper bounds on the rate of
convergence of estimates of the identified set based on kernel estimates.
For the first order kernel estimates described above, estimates of the
identified set will converge no faster than estimates based on variance
weighted KS statistics, and will only achieve the same rate if the tuning
parameter $h_n$ is chosen to go to zero at the proper rate.  Although this
means that properly weighted KS statistics will generally do at least as
well as first order kernel estimates and sometimes better in terms of
rates of convergence, kernel estimates with a properly chosen sequence
$h_n$ may do better against alternatives that approach the identified set
at a given rate.

The upper bound on rates of convergence for kernel based estimators is
stated in the following theorem.  In this theorem, the requirements that
the critical value $\hat c_n$ be large and that the bandwidth $h_n$
not shrink too quickly ensure that the procedure controls the probability
of false rejection.  If these conditions do not hold, we may have
$\Theta_0(P)\not\subseteq \mathcal{C}_n^{\text{kern}}(\hat c_n)$ with
high probability asymptotically.

\begin{theorem}\label{kernel_rate_upper_thm}
Suppose that Assumptions \ref{support_assump},
\ref{smoothness_upper_assump} and \ref{kernel_assump} hold.  If $\hat
c_n$ is chosen large enough, and if $h_n^{d_X}n/\log n\ge a$ for $a$
large enough, then, for some $\varepsilon>0$,
\begin{align*}
P\left(
\left(\frac{\sqrt{n h_n^{d_X}}}{\sqrt{\log n}}
  \wedge h_n^{-\alpha} \right)
  d_H(\mathcal{C}_n^{\text{kern}}(\hat c_n),\Theta_0(P))\ge \varepsilon
\right)
\stackrel{n\to\infty}{\to} 1.
\end{align*}
\end{theorem}

The upper bound on the rate of convergence in Theorem
\ref{kernel_rate_upper_thm} is the slower of
$\frac{\sqrt{n h_n^{d_X}}}{\sqrt{\log n}}$, which comes from a variance
term, and $h_n^{-\alpha}$, which comes from a bias term.  The optimal rate
of convergence for estimates based on first order kernels will be achieved
only when these terms are of the same order of magnitude, which
corresponds to
$h_n^{-\alpha}
=\mathcal{O}\left(\frac{\sqrt{n h_n^{d_X}}}{\sqrt{\log n}}\right)$
or
$h_n
=\mathcal{O}\left(\frac{\log n}{n}\right)^{1/(d_X+2\alpha)}$.  Thus,
choosing the optimal $h_n$ requires knowing or estimating the Holder
constant $\alpha$.  While kernel based estimates may give more power when
$h_n$ is chosen optimally, variance weighted KS statistics give the same
rate of convergence as kernel based estimates with the optimally chosen
$h_n$ without knowing $\alpha$.
If $h_n$ is chosen to go to zero at a different rate from the optimal rate
for a given data generating process, kernel based estimates of the
identified set will converge more slowly than estimates based on variance
weighted KS statistics.  If the choice of $h_n$ is far enough off from the
optimal choice (i.e. if the researcher is wrong enough about the
smoothness of the data generating process), even the rate of convergence
for unweighted KS statistics in Theorem \ref{rate_thm_bddweight} will be
better than the rate of convergence of the kernel based estimate.

\section{Monte Carlo}\label{monte_carlo_sec}

To examine the finite sample properties of the set estimates proposed in
this paper, and to illustrate their implementation, I perform a monte
carlo study.  I apply the weighted KS statistic based set estimates to a
quantile regression model with missing data
on the outcome variable, where no additional assumptions are imposed on
the process generating the missing values.  Letting $W_i^*$ be the true
value of the outcome variable, I simulate from a model where the median of
$W_i^*$ given $X_i=x$ is given by $\theta_1+\theta_2x$, but $W_i^*$ is not
always observed.  This falls into the framework of the interval quantile
regression model described in Section \ref{quantile_int_reg_sec}, with
$W_i^H=W_i^L=W_i^*$ when the outcome variable is observed, and
$W_i^H=\infty$ and $W_i^L=-\infty$ when the outcome variable is
unobserved.  The identified set contains all values of
$(\theta_1,\theta_2)$ that are consistent with the median regression model
and some, possibly endogenous, censoring mechanism generating the missing
values.

I generate data as follows.  For $X_i$ and $U_i^*$ generated as
independent variables with
$X_i\sim \text{unif}(-3,3)$ and
$U_i^*\sim \text{unif}(-1,1)$ and
$(\theta_{1,*},\theta_{2,*})=(1/4,1/2)$,
I set $W_i^*=\theta_{1,*}+\theta_{2,*}X_i+U_i^*$.  Then, %
I set $W_i^*$ to be missing (that is,
$(W_i^L,W_i^H)=(-\infty,\infty)$) with probability
$1/5-X_i^2/20+X_i^4/200$, and
observed ($W_i^L=W_i^H=W_i^*$) with the remaining probability
$1-(1/5-X_i^2/20+X_i^4/200)$.  Note that, while the data are generated by
taking a particular point $(\theta_{1,*},\theta_{2,*})$ in the identified
set and using a censoring process that satisfies
the missing at random assumption (that the event of $W_i^*$ not being
observed is independent of $U_i^*$ conditional on $X_i^*$), the identified
set for this model is larger than a single point, and contains all values
of $(\theta_1,\theta_2)$ that are consistent with median regression and
any form of censoring, including those where the probability of not
observing $W_i^*$ depends on the outcome $W_i^*$ itself.

\begin{figure}[h!]
  \begin{center}
  \includegraphics[width=3.5in]{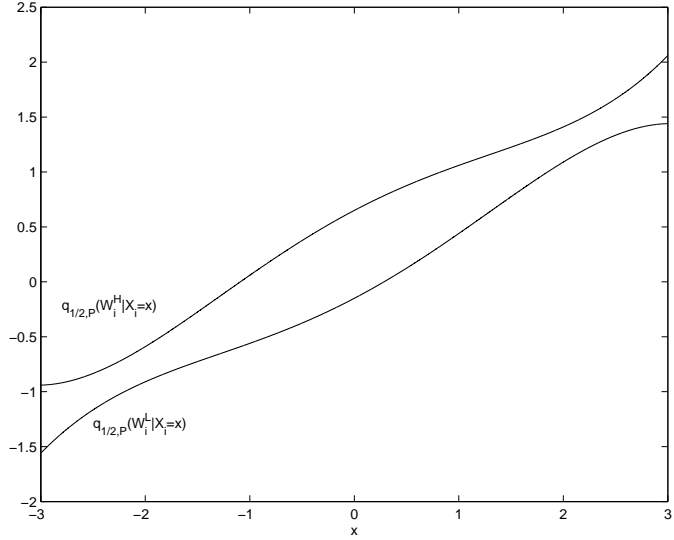}
  \end{center}
  \caption{Conditional Medians of $W_i^H$ and $W_i^L$ for Quantile
    Regression Model}
  \label{endpoint_cond_quants_fig}
\end{figure}

\begin{figure}[h!]
  \begin{center}
  \includegraphics[width=3.5in]{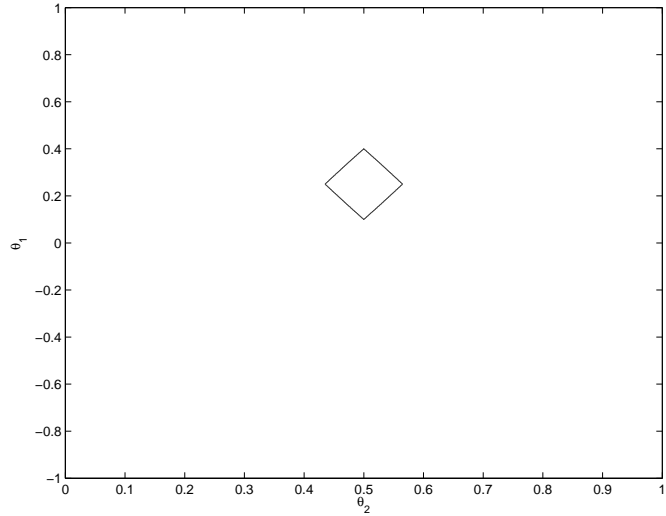}
  \end{center}
  \caption{Identified Set for Quantile Regression Model}
  \label{id_set_quant_fig}
\end{figure}

Figure \ref{endpoint_cond_quants_fig} shows the true conditional medians
$q_{1/2,P}(W_i^H|X_i=x)$ and $q_{1/2,P}(W_i^L|X_i=x)$ as a function of $x$
for this example.  The true identified set $\Theta_0(P)$ for this example
is the set of parameter values $(\theta_1,\theta_2)$ such that the line
$\theta_1+\theta_2x$ is between these two conditional medians for every
value of $x$ on the support of $X_i$.  Figure \ref{id_set_quant_fig} plots
the boundary of this identified set.  The identified set consists of all
points outlined by the shape in this figure.

To illustrate the implementation of the set estimates in this paper
applied to this model, I present a contour plot of the KS statistic
evaluated at different values of the parameters for a single data set
drawn from this data generating process.  For a given choice of the
critical value $\hat c_n$, the set estimate $\mathcal{C}_n(\hat c_n)$ is
then given by the set of points $(\theta_1,\theta_2)$ such that the KS
statistic $T_n(\theta)$ given in this plot is less than or equal to $\hat
c_n \sqrt{(\log n)/n}$.  In other words, each of the level sets in this
plot gives the boundary of $\mathcal{C}_n(\hat c_n)$ for some choice of
$\hat c_n$, with the level sets for larger values of the KS statistic
corresponding to larger (more conservative) choices of $\hat c_n$.

\begin{figure}[h!]
  \begin{center}
  \includegraphics[width=3.5in]{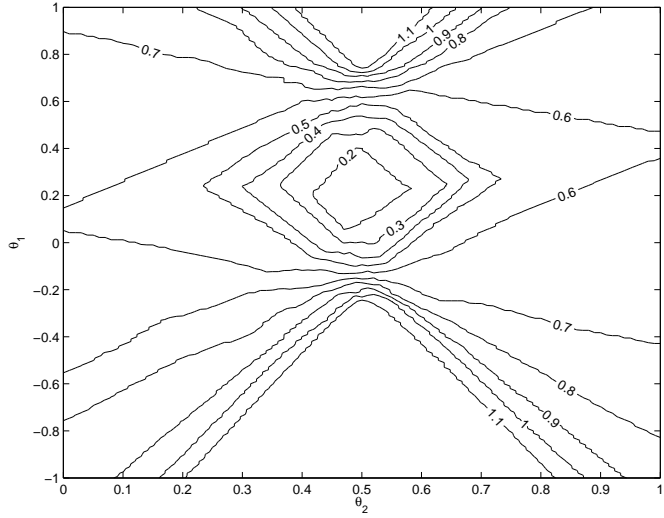}
  \end{center}
  \caption{Contours of KS Statistic for Quantile
    Regression Model}
  \label{ks_stat_contour_quant_fig}
\end{figure}

The contour plot is shown in Figure \ref{ks_stat_contour_quant_fig}.  This
plot was formed from a single data set of $n=500$ observations drawn from
the data generating process described above.  For the set of functions
$\mathcal{G}$, I used the set of indicator functions for intervals
$I(s<X<t)$.  For the truncation point $\sigma_n$ for the standard
deviation weights, I multiply $1/2$, the standard deviation of a single
$\text{Bernoulli}(1/2)$ variable, by $\sqrt{(\log n)(\log \log n)/n}$, a
sequence that converges to zero more slowly than
$\sqrt{\log n/n}$ as required.  I set $S(t_1,t_2)=\max(t_1,t_2,0)$.

For the monte carlo, I use the same choices of $\mathcal{G}$, $\sigma_n$,
and $S$.  For the critical value $\hat c_n$, I use the slowly increasing
sequence $2\sqrt{\log \log n}$.
Setting the critical value to $2$ regardless of $n$ leads to a
somewhat conservative critical value for the case where $\mathcal{G}$
contains a single function.  As $n$ increases, functions $I(s<X_i<t)$ with
$s$ close to $t$ are given increasing weight, so that the KS statistic
behaves like the maximum of an increasing number of standard normal
variables, and the critical value increases appropriately (or slightly
faster than needed).
I generate monte carlo data sets with the data generating process
described above and $n$ equal to $200$, $500$, and $1000$ observations.
I use 1000 replications for each monte carlo design.

\begin{table}[h!]
\centering
\begin{tabular}{c|ccccc|c}
\multicolumn{1}{c}{} & \multicolumn{5}{c}{quantiles 
  } &  \\
$n$&.25&.5&.75&.9&.95 & coverage \\\hline
200&0.45&0.5&0.54&0.59&0.62 & 100\%\\
500&0.34&0.36&0.39&0.41&0.43 & 100\%\\
1000&0.27&0.28&0.3&0.32&0.33 & 100\%\\\hline
\end{tabular}

\caption{Summary Statistics for Hausdorff Distances for Monte Carlo}
\label{mc_dists_table_quant}
\end{table}

\begin{table}[h!]
\centering
\begin{tabular}{c|c|ccccc}
 \multicolumn{2}{c}{} & \multicolumn{5}{c}{quantiles}  \\
&$n$&.25&.5&.75&.9&.95\\\hline
&200&0.45&0.49&0.53&0.57&0.59\\
$\theta_1$&500&0.34&0.36&0.38&0.41&0.43\\
&1000&0.26&0.28&0.3&0.32&0.33\\\hline
&200&0.33&0.36&0.41&0.48&0.52\\
$\theta_2$&500&0.24&0.25&0.27&0.28&0.29\\
&1000&0.19&0.2&0.21&0.22&0.22\\\hline
\end{tabular}

\caption{Summary Statistics for Hausdorff Distances for Individual
  Parameters for Monte Carlo}
\label{mc_proj_dists_table_quant}
\end{table}

\begin{table}[h!]
\centering
\begin{tabular}{c|c|ccccc|ccccc}
 \multicolumn{2}{c}{}  & \multicolumn{5}{c}{quantiles of $\hat \ell_{\theta_i}$}
    & \multicolumn{5}{c}{quantiles of $\hat u_{\theta_i}$}  \\
&$n$&.05&.1&.25&.5&.75&.25&.5&.75&.9&.95\\\hline
&200&-0.47&-0.45&-0.4&-0.35&-0.31&0.81&0.86&0.9&0.95&0.98\\
$\theta_1$&500&-0.31&-0.3&-0.26&-0.23&-0.2&0.71&0.74&0.77&0.79&0.81\\
&1000&-0.22&-0.2&-0.18&-0.16&-0.14&0.64&0.67&0.69&0.71&0.72\\\hline
&200&-0.05&0&0.05&0.1&0.14&0.87&0.9&0.95&1.01&1.06\\
$\theta_2$&500&0.16&0.17&0.18&0.2&0.22&0.78&0.8&0.82&0.83&0.84\\
&1000&0.22&0.23&0.24&0.25&0.27&0.73&0.75&0.76&0.77&0.78\\\hline
\end{tabular}

\caption{Summary Statistics for Set Estimates for Individual
  Parameters for Monte Carlo}
\label{mc_proj_table_quant}
\end{table}

Results are reported in Tables \ref{mc_dists_table_quant},
\ref{mc_proj_dists_table_quant} and \ref{mc_proj_table_quant}.
In Table \ref{mc_dists_table_quant}, I report, for each sample size, the
coverage probability (the proportion of the monte carlo replications for
which the estimate contains the identified set) and quantiles of
the Hausdorff distance between the estimate and the identified set.
Even for the smallest sample size of $n=200$, the set estimate contains
the identified set for every monte carlo replication.
For the data generating process used in these monte carlos, %
Theorem \ref{rate_thm_alpha} holds with $\alpha=2$, giving a rate of
convergence of
$((\hat c_n^2 \log n)/n)^{\alpha/(d_X+2\alpha)}
  =((\hat c_n^2 \log n)/n)^{2/5}$
(this follows from Theorem \ref{quantile_int_reg_thm}, since the
upper and lower conditional medians are bounded away from each other and
have smooth second derivatives, and regression lines corresponding to
parameters on the boundary of the identified set are tangent to one of the
conditional medians on the interior of the support of $X_i$).
According to this result, the distance of the set estimate to the
identified set should decrease by a factor of
$.77$ going from $n=200$ to $n=500$, and should decrease again by a factor
of
$.81$ going from $n=500$ to $n=1000$.
These asymptotic results give a decent approximation of the monte carlo
results reported in Table \ref{mc_dists_table_quant}, although the
Hausdorff distances in this table decrease slightly more quickly.

The Hausdorff distance to the identified set summarizes the accuracy of
the set estimator, but can be difficult to interpret, since it combines
the accuracy of the estimate across both coordinates.  Tables
\ref{mc_proj_dists_table_quant} and \ref{mc_proj_table_quant} report
monte carlo results for the projections of the set estimate
$\mathcal{C}_n(\hat c_n)$ onto each coordinate.
Define
$%
\Theta_{1,\text{proj}}(P)=\{\theta|(\theta,\theta_2)\in\Theta_0(P)
  \text{ some $\theta_2$}\}$
and
$\Theta_{2,\text{proj}}(P)=\{\theta|(\theta_1,\theta)\in\Theta_0(P)
  \text{ some $\theta_1$}\}$
to be the projections of the identified set onto each coordinate and
$\mathcal{C}_{n,1,\text{proj}}(\hat c_n)
  =\{\theta|(\theta,\theta_2)\in\mathcal{C}_n(\hat c_n)
    \text{ some $\theta_2$}\}$
and
$\mathcal{C}_{n,2,\text{proj}}(\hat c_n)
  =\{\theta|(\theta_1,\theta)\in\mathcal{C}_n(\hat c_n)
    \text{ some $\theta_1$}\}$
to be the corresponding projections of the set estimate.
Let $[\hat\ell_{\theta_i},\hat u_{\theta_i}]$ be the smallest interval
containing $\mathcal{C}_{n,i,\text{proj}}(\hat c_n)$ for $i=1,2$.
In Table \ref{mc_proj_dists_table_quant}, I report
quantiles of the realizations of
$d_H(\mathcal{C}_{n,1,\text{proj}}(\hat c_n),\Theta_{1,\text{proj}}(P))$
and
$d_H(\mathcal{C}_{n,2,\text{proj}}(\hat c_n),\Theta_{2,\text{proj}}(P))$
for the monte carlo replications.
In Table \ref{mc_proj_table_quant}, I report quantiles of
$\hat\ell_{\theta_i}$ and $\hat u_{\theta_i}$.

Table \ref{mc_proj_dists_table_quant} reveals that the set estimate for
the slope parameter $\theta_2$ has less sampling error in this case than
the intercept parameter $\theta_1$.  Indeed, the Hausdorff distances in
Table \ref{mc_dists_table_quant} appear to be driven mostly by the
intercept parameter.  While estimation error in $\theta_2$ is generally
smaller, than estimation error for the intercept parameter, the estimate
for $\theta_2$ appears to be shrinking towards the identified set for
$\theta_2$ at a similar rate.

Table \ref{mc_proj_table_quant} summarizes the finite sample behavior of
the confidence intervals $[\hat\ell_{\theta_i},\hat u_{\theta_i}]$ for
each $\theta_i$ generated from the set estimate $\mathcal{C}_n(\hat
c_n)$.  While these confidence intervals for individual coordinates
contain less information than the confidence region $\mathcal{C}_n(\hat
c_n)$ for the identified set, they are less cumbersome to report and
summarize.  For comparison, the projection of the true identified set onto
the intercept coordinate $\theta_2$ is
$\Theta_{1,\text{proj}}(P)=[.17,.33]$, and the projection of the true
identified set onto the slope coordinate $\theta_2$ is
$\Theta_{2,\text{proj}}(P)=[.47,.53]$.  Note that the confidence interval
for the slope parameter $\theta_2$ contains only positive values for 90\%
of the monte carlo replications even with the smallest sample size of 200
observations.  Thus, one would correctly conclude that $\theta_2$ is
positive for an overwhelming proportion of realizations of the data even
with a relatively small sample size, despite the conservative nature of
the estimate $\mathcal{C}_n(\hat c_n)$.

\section{Conclusion}\label{conclusion}

This paper proposes estimates of the identified set in conditional
moment inequality models based on variance weighted KS statistics.  I
derive rates of convergence of these and other set estimators to the
identified set under conditions that apply to many models of practical
interest.  In many settings, the rate of convergence of the set
estimator I propose is the fastest among those available, and, in settings
where other estimators are better, the improvement in rate of convergence
is no more than a factor of $\log n$.  While, in most cases, there is some
other estimator that does slightly better, choosing the correct one
requires knowledge of smoothness and shape conditions on the data
generating process, and guessing incorrectly about these conditions can
lead the researcher to use an estimator with a much slower rate of
convergence.  The advantage of the estimator proposed in this paper is
that it performs well under a variety of conditions without prior
knowledge of which of these conditions hold.

In settings where local
alternatives violate the conditional moment inequalities on a shrinking
set, the weights I propose for KS statistics give the statistics more power
against local alternatives than bounded weights.  The examples in Section
\ref{applications_sec} show that this
situation is common in practice.  When sequences of local alternatives
violate the conditional moment inequalities on a fixed, positive
probability set, the larger critical values required by the increasing
sequence of weight functions lead to a loss in power, but only by a factor
of $(\log n)^{1/2}$.  This provides a theoretical justification for
variance weighting in this context.  Under certain conditions, weighting
the KS statistic objective function by a truncated inverse of the
estimated variance increases the rate of convergence of the corresponding
estimator of the identified set.

\appendix
\section{Appendix}

This appendix collects several results not stated in the body of the
paper.
In Section \ref{unif_conv_sec}, I state and prove uniform convergence
results for classes of functions weighted by truncated standard
deviations.  These results are used later in the appendix in proving some
of the results stated in the body of the paper.
In Section \ref{exact_rate_sec}, I provide sufficient conditions for the
rate of convergence to be strictly faster than $\sqrt{n}$.
In Section \ref{slope_examp_sec}, I provide an example of a data
generating process for an interval regression where low power against
local alternatives when the slope parameter varies leads to a slower rate
of convergence to the identified set.
In Section \ref{covering_sec}, I state conditions under which Assumption
\ref{covering_assump} holds and verify them for the applications described
in Section \ref{applications_sec}.
Section \ref{proofs_sec} contains proofs of the theorems stated in the
body of the paper.

\subsection{Uniform Convergence Lemma}\label{unif_conv_sec}

The following lemma is useful in deriving some of these results.  Applied
to mean zero functions, the lemma says that any sequence of classes of
functions that is not too complex converges uniformly at a
$\sqrt{n/\log n}$ rate when scaled by the standard deviation if the
minimum standard deviation does not go to zero too fast.

\begin{lemma}\label{sd_rate_lemma}
Let $Z_1,\ldots,Z_n$ be iid observations and let $\mathcal{P}$ be a set of
probability distributions and $\mathcal{F}_{n,P}$ a set of classes of
functions indexed by $n\in \mathbb{N}$ and $P\in\mathcal{P}$
such that, for some $\overline f$,
$f(Z_i)\le \overline f$ with $P$-probability one for $P\in\mathcal{P}$ and
$f\in\mathcal{F}_{n,P}$ for each $n$.
Let $\mu_{2,P}(f)=(E_Pf(Z_i)^2)^{1/2}$ and let $\mu_{2,n}$ be a sequence
such that $\mu_{2,n}\sqrt{n/\log n}$ is bounded away from zero.
Let $\mathcal{G}_{n,P}=\{f\mu_{2,n}/(\mu_{2,P}(f)\vee
\mu_{2,n})|f\in\mathcal{F}_{n,P}\}$ and suppose that
\begin{align*}
\sup_{P\in \mathcal{P}}\, \sup_{n\in\mathbb{N}}\, \sup_{Q}
N(\varepsilon,\mathcal{G}_{n,P},L_1(Q))
\le A\varepsilon^{-W}
\end{align*}
for $0<\varepsilon<1$ where the supremum over $Q$ is over all probability
measures. Then
for some $B$ that does not depend on $N$,
\begin{align*}
\sup_{P\in\mathcal{P}}
P\left(
\frac{\sqrt{n}}{\sqrt{\log n}}
\sup_{f\in\mathcal{F}_{n,P}}
\left|(E_n-E_P)\frac{f(Z_i)}{\mu_{2,P}(f)\vee \mu_{2,n}}\right|
\ge B
\, \text{some $n\ge N$}\right)
\stackrel{N\to\infty}{\to} 0.
\end{align*}
\end{lemma}
\begin{proof}
The result follows by applying the following theorem to the classes of
functions $\mathcal{G}_{n,P}$.  For $g=f\mu_{2,n}/(\mu_{2,P}(f)\vee \mu_{2,n})\in
\mathcal{G}_{n,P}$,
$E_Pg(Z_i)^2=E_Pf(Z_i)^2\mu_{2,n}^2/(\mu_{2,P}(f)^2\vee \mu_{2,n}^2)
=\mu_{2,P}(f)^2\mu_{2,n}^2/(\mu_{2,P}(f)^2\vee \mu_{2,n}^2)\le \mu_{2,n}^2$, so the
theorem applies with the same $\mu_{2,n}$.
\end{proof}

Specialized to a class $\mathcal{P}$ of probability distributions with a
single element $P$, this says that the sequence in the probability
statement in the last display of the lemma is bounded by $B$ with
$P$-probability one.  The conclusion of the lemma implies that this scaled
sequence is $\mathcal{O}_P(1)$ uniformly in $P\in\mathcal{P}$, but is
slightly stronger.

The proof of the lemma uses the following theorem, which is a slightly
stronger version of Theorem 37 in \citet{pollard_convergence_1984}, with
the conditions stated in a slightly different way.
The following theorem basically follows the arguments of
the proof of Theorem 37 in \citet{pollard_convergence_1984}, but changes a
few things to get a slightly stronger result.  Note that the notation
$\mu_{2,P}^2$ is used for the raw second moment of functions rather than
their variance, although the distinction is often not important since
applications typically involve the raw second moment going to zero at the
same rate as the variance.

\begin{theorem}\label{sd_scale_thm}
Let $Z_1,\ldots,Z_n$ be iid observations and let
$\mathcal{P}$ be a set of probability measures and
$\mathcal{F}_{n,P}$ a
set of classes of functions indexed by $n\in \mathbb{N}$ and
$P\in\mathcal{P}$
such that, for some $\overline f$,
$f(Z_i)\le \overline f$ $P$-a.s. for $f\in\mathcal{F}_{n,P}$ for
$P\in\mathcal{P}$ for each $n$ and, for some
positive constants $A$ and $W$,
\begin{align*}
\sup_{P\in\mathcal{P}}\, \sup_{n\in \mathbb{N}}\, \sup_{Q}
N(\varepsilon,\mathcal{F}_{n,P},L_1(Q))
\le A\varepsilon^{-W}
\end{align*}
for $0< \varepsilon< 1$ where the supremum over $Q$ is over all
probability measures.  Suppose that, for some sequence $\mu_{2,n}$,
$E_Pf(Z_i)^2\le \mu_{2,n}^2$ for all $f\in \mathcal{F}_{n,P}$ for all
$P\in\mathcal{P}$ for all $n$.
Then, if $\mu_{2,n}\sqrt{n/\log n}$ is bounded away from zero
we will have, for some $B$ that does not depend on $N$,
\begin{align*}
\sup_{P\in \mathcal{P}} P\left(
\frac{\sqrt{n}}{\mu_{2,n}\sqrt{\log n}}
\sup_{f\in\mathcal{F}_{n,P}} \left|(E_n-E_P)f(Z_i)\right|
\ge B
\, \text{some $n\ge N$}\right)
\stackrel{N\to\infty}{\to} 0.
\end{align*}
\end{theorem}
\begin{proof}
The proof is a slight modification of the proof of Theorem 37 in
\citet{pollard_convergence_1984}.  The sequence $\mu_{2,n}$ corresponds to
$\delta_n$ in that theorem, and, in contrast to the theorem from
\citet{pollard_convergence_1984} which defines a sequence $\alpha_n$ that
must satisfy certain conditions, this theorem corresponds to using the
best $\alpha_n$ sequence possible, and noting that $\alpha_n$ need not be
nonincreasing as long as it is bounded.

Without loss of of generality, assume that $\overline f=1$.
Fix $B$ (conditions on how large $B$ has to be will be stated throughout
the theorem) and set
$\varepsilon_n=\frac{B \mu_{2,n}\sqrt{\log n}}{8\sqrt{n}}$.  Since
$var_P((E_n-E_P)f(Z_i))/(4\varepsilon_n^2)
  \le (\mu_{2,n}^2/n)/(4 B^2\mu_{2,n}^2(\log n)/(64n))
  = 16/(B^2\log n) \le 1/2$
for $n$ greater than some number that does not depend on $P$, the
inequality (30) in \citet{pollard_convergence_1984} will eventually imply
\begin{align*}
&P\left(\frac{\sqrt{n}}{\mu_{2,n}\sqrt{\log n}}
\sup_{f\in\mathcal{F}_{n,P}} \left|(E_n-E)f(Z_i)\right|
\ge B\right)
=P\left(\sup_{f\in\mathcal{F}_{n,P}} \left|(E_n-E)f(Z_i)\right|
\ge 8\varepsilon_n\right)  \\
&\le 4 (P\times\nu) \left(\sup_{f\in\mathcal{F}_{n,P}}
  \left|\mathbb{P}_n^\circ f(Z_i)\right|
\ge 2\varepsilon_n\right)
\end{align*}
for all $P\in\mathcal{P}$
where $\mathbb{P}_n^\circ f(Z_i) = \frac{1}{n}\sum_{i=1}^nf(Z_i)\cdot s_i$
and $s_1,\ldots,s_n$ are iid random variables that take on values $\pm 1$
each with probability one half drawn independent of $Z_1,\ldots,Z_n$ and
$\nu$ denotes the probability measure of $s_1,\ldots,s_n$.
Conditional on the data, this
is bounded by
\begin{align*}
(P\times \nu)
\left(\sup_{f\in\mathcal{F}_n} \left|\mathbb{P}_n^\circ f(Z_i)\right|
\ge 2\varepsilon_n\bigg|Z_1,\ldots,Z_n\right)
\le 2N(\varepsilon_n,\mathcal{F}_{n,P},L_1(\mathbb{P}_n))
  \exp\left[-\frac{1}{2}
    \frac{n\varepsilon_n^2}
      {\left(\sup_{f\in\mathcal{F}_{n,P}} E_nf(Z_i)^2\right)}\right].
\end{align*}
For any constant $a>0$,
on the event that
\begin{align}\label{sup_En_ineq}
\sup_{f\in\mathcal{F}_{n,P}} E_nf(Z_i)^2\le a^2\mu_{2,n}^2,
\end{align}
the previous display will be bounded by
\begin{align*}
&2N(\varepsilon_n,\mathcal{F}_{n,P},L_1(\mathbb{P}_n))
  \exp\left(-\frac{1}{2}
    \frac{n\varepsilon_n^2}{a^2\mu_{2,n}^2}\right)
\le 2A\varepsilon_n^{-W}
  \exp\left(-\frac{1}{2}
    \frac{n\varepsilon_n^2}{a^2\mu_{2,n}^2}\right)  \\
&= 2A\exp\left[-\frac{1}{2}
  \cdot n \cdot \frac{B^2 \mu_{2,n}^2\log n}{64 n}
  \cdot \frac{1}{a^2\mu_{2,n}^2}
  - W \log \frac{B \mu_{2,n}\sqrt{\log n}}{8\sqrt{n}}\right]  \\
&= 2A\exp\left[-\frac{B^2 \log n}{128 a^2}
  - W \log \frac{B}{8}
  - W \log \frac{\mu_{2,n}\sqrt{\log n}}{\sqrt{n}}\right]  %
\end{align*}
The condition that $\mu_{2,n}\sqrt{n/\log n}$ is bounded away from zero is
more than enough to guarantee that the term in the last logarithm is
bounded from below by a fixed power of $n$.
Thus, the expression in the last display can be
made to go to zero at any polynomial rate for any $a$ by choosing $B$ to
be large enough (in a way that depends on $a$ but not $n$ or $P$).

For any $P\in\mathcal{P}$, the $P$-probability of (\ref{sup_En_ineq}) failing to hold can be bounded using
Lemma 33 in \citet{pollard_convergence_1984} with $\delta_n=a\mu_{2,n}/8$
(the lemma holds for $a\ge 8$):
\begin{align*}
&P\left(\sup_{f\in\mathcal{F}_{n,P}} E_nf(Z_i)^2> a^2\mu_{2,n}^2\right)
=P\left(\sup_{f\in\mathcal{F}_{n,P}} E_nf(Z_i)^2> 64\delta_n^2\right)
\le 4 E_P[N(\delta_n,\mathcal{F}_{n,P},L_2(\mathbb{P}_n))]
  \exp(-n\delta_n^2)  \\
&\le 4 A(\delta_n/2)^{-W}\exp(-n\delta_n^2)
= 4\cdot 2^W A \exp(-n\delta_n^2-W \log \delta_n)  \\
&= 4\cdot 2^W A \exp\left[-n a^2\mu_{2,n}^2/64
  -W \log \frac{a}{8}
  -W \log \mu_{2,n}\right]  \\
&\le 4\cdot 2^W A \exp\left[-\frac{n a^2}{64}\frac{c\log n}{n}
  -W \log \frac{a}{8}
  -\frac{W}{2} \log \frac{c\log n}{n}\right]  %
\end{align*}
where $\sqrt{c}$ is a lower bound for $\mu_{2,n}\sqrt{n/\log n}$.  This can
be made to go to zero at any polynomial rate by choosing $a$ large.

Thus, if we choose $a$ and $B$ large enough,
$\sup_{P\in \mathcal{P}} P\left(\frac{\sqrt{n}}{\mu_{2,n}\sqrt{\log n}}
\sup_{f\in\mathcal{F}_{n,P}} \left|(E_n-E_P)f(Z_i)\right|\ge B\right)$ will be
summable over $n$, so that
\begin{align*}
&\sup_{P\in \mathcal{P}} P\left(\frac{\sqrt{n}}{\mu_{2,n}\sqrt{\log n}}
\sup_{f\in\mathcal{F}_{n,P}} \left|(E_n-E_P)f(Z_i)\right|\ge B
\, \text{some $n\ge N$}\right)  \\
&\le \sum_{n\ge N} \sup_{P\in \mathcal{P}}
P\left(\frac{\sqrt{n}}{\mu_{2,n}\sqrt{\log n}}
\sup_{f\in\mathcal{F}_{n,P}} \left|(E_n-E_P)f(Z_i)\right|\ge B
\right)
\stackrel{N\to\infty}{\to} 0.
\end{align*}

\end{proof}

With this lemma in hand, we can get rates of convergence for classes of
functions weighted by their standard deviation under additional conditions
that allow the standard deviation to be consistently estimated.  In order
to get results for functions weighted by the standard deviation rather
than the raw second moment, I apply the previous results to classes of
functions of the form $f-E_Pf(Z_i)$.
Letting
$\hat\sigma(f)^2=E_n(f(Z_i))^2-(E_nf(Z_i))^2$ and
$\sigma_P(f)^2=E_P(f(Z_i))^2-(E_Pf(Z_i))^2$, rates of convergence for
\begin{align*}
\sup_{f\in\mathcal{F}_n}
\left|(E_n-E_P)\frac{f(Z_i)}{\hat\sigma(f)\vee\sigma_n}\right|
\end{align*}
will follow by applying the above results to the classes of functions
$f-E_Pf(Z_i)$ once we can bound
$\frac{\sigma_P(f)\vee\sigma_n}{\hat\sigma(f)\vee\sigma_n}$,
and for this it is sufficient to show that $\hat\sigma(f)/\sigma_P(f)$
converges to one uniformly over $\sigma_P(f)\ge \sigma_n$.  The following
lemma gives sufficient conditions for this.

\begin{lemma}\label{sd_est_lemma}
Let $Z_1,\ldots,Z_n$ be iid observations and let $\mathcal{F}_n$ be a
sequence of classes of functions
and $\mathcal{P}$ a set of probability distributions
such that, for some $\overline f$,
$f(Z_i)\le \overline f$ with $P$-probability one for $P\in\mathcal{P}$ and
$f\in\mathcal{F}_n$ for each $n$.
Let $\sigma_P(f)=(E_Pf(Z_i)^2-(E_Pf(Z_i))^2)^{1/2}$ and let $\sigma_n$ be
a sequence such that $\sigma_n\sqrt{n/\log n}$ is bounded away from zero.
Define
$\mathcal{G}_{n,P}^1
  =\{(f-E_P f(Z_i))\sigma_n/(\sigma_P(f)\vee\sigma_n)\}$
and
$\mathcal{G}_{n,P}^2=\{(f-E_P f(Z_i))^2\sigma_n
  /(\mu_{2,P}([f-E_P f(Z_i)]^2)\vee \sigma_n)\}$,
and suppose that, for some positive constants $A$ and $W$,
\begin{align*}
\sup_{P\in\mathcal{P}}\, \sup_{n\in\mathbb{N}}\, \sup_Q
N(\varepsilon,\mathcal{G}_{n,P}^i,L_1(Q))\le A\varepsilon^{-W}
\end{align*}
for $0<\varepsilon<1$ and $i=1,2$, where the supremum over $Q$ is over all
probability measures.

Then, for every $\varepsilon>0$, there exists a $c$ such that, if
$\sigma_n\sqrt{n/\log n}\ge c$ for all $n$,
\begin{align*}
\sup_{P\in\mathcal{P}}
P\left(
\sup_{f\in\mathcal{F}_n, \sigma_P(f)\ge \sigma_n}
  \left|\frac{\hat\sigma(f)}{\sigma_P(f)} - 1\right|
\ge \varepsilon
\, \text{some $n\ge N$}\right)
\stackrel{N\to\infty}{\to}0.
\end{align*}

\end{lemma}
\begin{proof}
We have
\begin{align}\label{sd_tri_ineq}
&\sup_{f\in\mathcal{F}_n, \sigma_P(f)\ge \sigma_n}
\left|\frac{\hat\sigma^2(f)-\sigma_P^2(f)}{\sigma_P^2(f)}\right|  \notag  \\
&=\sup_{f\in\mathcal{F}_n, \sigma_P(f)\ge \sigma_n}
\left|\frac{(E_n-E_P)(f(Z_i)-E_Pf(Z_i))^2-(E_nf(Z_i)-E_Pf(Z_i))^2}
  {\sigma_P^2(f)}\right|  \notag  \\
&\le \sup_{f\in\mathcal{F}_n, \sigma_P(f)\ge \sigma_n}
\left|\frac{(E_n-E_P)(f(Z_i)-E_Pf(Z_i))^2}
  {\sigma_P^2(f)}\right|
+\left|\frac{[(E_n-E_P)f(Z_i)]^2}
  {\sigma_P^2(f)}\right|.
\end{align}
The first term is equal to
\begin{align*}
\left|\frac{(E_n-E_P)(f(Z_i)-E_Pf(Z_i))^2}
  {\mu_{2,P}([f-E_Pf(Z_i)]^2)\vee \sigma_n}\right|
\frac{\mu_{2,P}([f-E_Pf(Z_i)]^2)\vee \sigma_n}
  {\sigma_P^2(f)}.
\end{align*}
We have
$\mu_{2,P}([f-E_Pf(Z_i)]^2)^2=E_P[f(Z_i)-E_Pf(Z_i)]^4
\le 4\overline f^2E_P[f(Z_i)-E_Pf(Z_i)]^2=4\overline f^2\sigma_P(f)^2$ so
that
\begin{align*}
\frac{\mu_{2,P}([f-E_Pf(Z_i)]^2)\vee \sigma_n}
  {\sigma_P^2(f)}
\le \frac{[2\overline f\sigma_P(f)]\vee \sigma_n}{\sigma_P^2(f)}
\le \frac{2\overline f\vee 1}{\sigma_P(f)}
\le \frac{2\overline f\vee 1}{\sigma_n}
\end{align*}
where the last two inequalities hold for $\sigma_P(f)\ge \sigma_n$.  Thus,
for any $\varepsilon>0$,
\begin{align*}
&\sup_{P\in\mathcal{P}}
P\left(
\sup_{f\in\mathcal{F}_n, \sigma_P(f)\ge \sigma_n}
\left|\frac{(E_n-E_P)(f(Z_i)-E_Pf(Z_i))^2}
  {\sigma_P^2(f)}\right|\ge \varepsilon
\, \text{some $n\ge N$}\right)  \\
&\le \sup_{P\in\mathcal{P}}
P\left(
\sup_{f\in\mathcal{F}_n, \sigma_P(f)\ge \sigma_n}
\frac{2\overline f\vee 1}{\sigma_n}
\left|\frac{(E_n-E_P)(f(Z_i)-E_Pf(Z_i))^2}
  {\left\{E[(f(Z_i)-E_Pf(Z_i))^2]^2\right\}^{(1/2)}\vee \sigma_n}\right|
\ge \varepsilon
\, \text{some $n\ge N$}\right)  \\
&\le \sup_{P\in\mathcal{P}}
P\left(
\sup_{f\in\mathcal{F}_n, \sigma_P(f)\ge \sigma_n}
\frac{\sqrt{n}}{\sqrt{\log n}}
\left|\frac{(E_n-E_P)(f(Z_i)-E_Pf(Z_i))^2}
  {\left\{E[(f(Z_i)-E_Pf(Z_i))^2]^2\right\}^{(1/2)}\vee \sigma_n}\right|
\ge c\varepsilon/(2\overline f\vee 1)
\, \text{some $n\ge N$}\right)
\end{align*}
where the last inequality holds for $\sigma_n\sqrt{n/\log n}\ge c$.  By
Lemma \ref{sd_rate_lemma}, this
will go to zero if $c$ is large enough so that $c\varepsilon/(2\overline
f\vee 1)$ is greater than the $B$ for which the conclusion of Lemma
\ref{sd_rate_lemma} holds for the class $\mathcal{G}_{n,P}^2$.

The probability that the second term in the last line of Equation
\ref{sd_tri_ineq} is greater than $\varepsilon>0$ for some $n\ge N$ goes
to zero uniformly in $P\in\mathcal{P}$ by Lemma \ref{sd_rate_lemma}
with the class $\{f-E_P f(Z_i)|f\in\mathcal{F}_n\}$ taking the place of $\mathcal{F}_{n,P}$ in that lemma.
\end{proof}

Combining these lemmas gives a consistency result for classes of functions
weighted by their standard deviations.  The conditions are the same as
those for Lemma \ref{sd_est_lemma}.

\begin{lemma}\label{sd_hat_rate_lemma}
Let $Z_1,\ldots,Z_n$ be iid observations and let $\mathcal{F}_n$ be a
sequence of classes of functions
and $\mathcal{P}$ a set of probability distributions
such that, for some $\overline f$,
$f(Z_i)\le \overline f$ with $P$-probability one for $P\in\mathcal{P}$ and
$f\in\mathcal{F}_n$ for each $n$.
Let $\sigma_P(f)=(E_Pf(Z_i)^2-(E_Pf(Z_i))^2)^{1/2}$.
Define
$\mathcal{G}_{n,P}^1
  =\{(f-E_P f(Z_i))\sigma_n/(\sigma_P(f)\vee\sigma_n)\}$
and
$\mathcal{G}_{n,P}^2=\{(f-E_P f(Z_i))^2\sigma_n
  /(\mu_{2,P}([f-E_P f(Z_i)]^2)\vee \sigma_n)\}$,
and suppose that, for some positive constants $A$ and $W$,
\begin{align*}
\sup_{P\in\mathcal{P}}\, \sup_{n\in\mathbb{N}}\, \sup_Q
N(\varepsilon,\mathcal{G}_{n,P}^i,L_1(Q))\le A\varepsilon^{-W}
\end{align*}
for $0<\varepsilon<1$ and $i=1,2$, where the supremum over $Q$ is over all
probability measures.

Then, for some $B$ and $c$ that do not depend on $N$ or $P$, if
$\sigma_n\sqrt{n/\log n}\ge c$ for all $n$,
\begin{align*}
\sup_{P\in\mathcal{P}}
P\left(
\sup_{f\in\mathcal{F}_n} \frac{\sqrt{n}}{\sqrt{\log n}} \left|
  \frac{f(Z_i)-E_P(f(Z_i))}{\hat\sigma(f)\vee \sigma_n}\right|
  \ge B
\, \text{some $n\ge N$}
\right)
\stackrel{N\to\infty}{\to} 0.
\end{align*}

\end{lemma}
\begin{proof}
We have
\begin{align*}
&P\left(
\sup_{f\in\mathcal{F}_n} \frac{\sqrt{n}}{\sqrt{\log n}} \left|
  \frac{f(Z_i)-E_P(f(Z_i))}{\hat\sigma(f)\vee \sigma_n}\right|
  \ge B
\, \text{some $n\ge N$}
\right)  \\
&=P\left(
\sup_{f\in\mathcal{F}_n} \frac{\sqrt{n}}{\sqrt{\log n}} \left|
  \frac{f(Z_i)-E_P(f(Z_i))}{\sigma_P(f)\vee \sigma_n}\right|
  \frac{\sigma_P(f)\vee \sigma_n}{\hat\sigma(f)\vee \sigma_n}
  \ge B
\, \text{some $n\ge N$}
\right)  \\
&\le P\left(
\sup_{f\in\mathcal{F}_n} \frac{\sqrt{n}}{\sqrt{\log n}} \left|
  \frac{f(Z_i)-E_P(f(Z_i))}{\sigma_P(f)\vee \sigma_n}\right|
  \ge B/2
\, \text{some $n\ge N$}
\right)  \\
&+P\left(
\inf_{f\in\mathcal{F}_n}
\frac{\hat\sigma(f)\vee \sigma_n}{\sigma_P(f)\vee \sigma_n}
\le 1/2
\, \text{some $n\ge N$}
\right).
\end{align*}
The second to last line goes to zero uniformly in $P\in\mathcal{P}$ by
Lemma \ref{sd_rate_lemma} applied to the classes
$\{f-E_P(f)|f\in\mathcal{F}_n,P\in\mathcal{P}\}$ (here, $B$ must be chosen
large enough so that the conclusion of this lemma holds with $B$ replaced
by $B/2$).  Since $\frac{\hat\sigma(f)\vee \sigma_n}{\sigma_P(f)\vee \sigma_n}\ge 1>1/2$ when $\sigma_P(f)<\sigma_n$,
the last line is bounded by
\begin{align*}
P\left(
\inf_{f\in\mathcal{F}_n, \sigma_P(f)\ge \sigma_n}
\frac{\hat\sigma(f)}{\sigma_P(f)}
\le 1/2
\, \text{some $n\ge N$}
\right),
\end{align*}
which goes to zero uniformly in $P\in\mathcal{P}$ if
$\sigma_n\sqrt{n/\log n}\ge c$ for $c$ large enough by Lemma
\ref{sd_est_lemma}.

\end{proof}

\subsection{Conditions for Exact Rate of Convergence}
\label{exact_rate_sec}

If $\sigma_n$ is fixed, we will have a $\sqrt{n}$ rate of uniform
convergence for the KS statistic.  The $\sqrt{n/\log n}$ rate of
convergence results used in Theorem \ref{coverage_thm} do not rule this
out for the case where $\sigma_n$ goes to zero, but another argument shows
that the rate of convergence will be strictly slower than $\sqrt{n}$ in many situations.

\begin{assumption}\label{contact_set_assump}
For some $\theta\in\Theta_0(P)$, some $j$, and some open set
$\mathcal{X}$, the following hold.
(i) %
$E_P(m_j(W_i,\theta)|X_i)=0$ a.s. on $\mathcal{X}$ and
$X_i$ has a density $f_X(x)$ on $\mathcal{X}$ that is bounded from above
and from below away from zero.
(ii) $var(m(W_i,\theta)|X_i=x)$ is continuous as a function of $x$ and
bounded away from zero and infinity on $\mathcal{X}$.
(iii) $\mathcal{G}$ contains the function $t\mapsto k((t-x)/h)$ for all
$x$ and all $h$ less than some fixed positive constant where $k$ satisfies
Assumption \ref{kernel_assump} and is continuous at zero.
\end{assumption}

The assumption on the set of functions $\mathcal{G}$ covers many commonly used cases, including indicator sets for $d_X$ dimensional rectangles or boxes.

\begin{theorem}
If Assumption \ref{contact_set_assump} holds and $S$ satisfies Assumption
\ref{S_assump}, then, if $\sigma_n\to 0$, $\sqrt{n}T_n(\theta)$ will
diverge to $\infty$.
\end{theorem}
\begin{proof}
Fix any points $x_1,\ldots,x_\ell\in\mathcal{X}$.  For $k$ from $1$ to
$\ell$, let $g_{n,k}(t)=k((t-x_k)/h_n)$
\begin{align*}
Z_{n,k}=\frac{1}{\hat\sigma_{n,j}(\theta,g_{n,k})\vee \sigma_n}
E_nm_j(W_i,\theta)k((X_i-x_k)/h_n)
\end{align*}
where $h_n$ is a sequence going to zero
such that $h_n^{d_X/2}/\sigma_n$ goes to infinity and $h_n^{d_X/2}\ge
n^{-\alpha}$ for some $\alpha<1$.  By the assumption on $S$,
$\sqrt{n}T_n(\theta)$ will diverge to $\infty$ if
$\inf_{x,h,j} \frac{1}{\hat\sigma_{n,j}(\theta,x,h)\vee \sigma_n}
E_nm_j(W_i,\theta)k((X_i-x)/h)$
diverges to $-\infty$, and, for this, it is sufficient to show that
$\min_k Z_{n,k}$ can be made arbitrarily small asymptotically by making
$\ell$ large enough.
Using standard arguments, it can be shown that
$\hat\sigma_{n,j}(\theta,g_{n,k})
  /\sigma_{P,j}(\theta,g_{n,k})$
converges in probability to one, and, since
$\sigma_{P,j}(\theta,g_{n,k})/h_n^{d/2}$ converges to a constant
under these assumptions, we also have that
\begin{align*}
Z_{n,k}=\frac{1}{\hat\sigma_{n,j}(\theta,g_{n,k})}
E_nm_j(W_i,\theta)k((X_i-x_k)/h_n)
\end{align*}
with probability approaching one.  By the Lindeberg central limit theorem,
defining
\begin{align*}
\tilde Z_{n,k}\equiv \frac{1}{\sigma_{P,j}(\theta,g_{n,k})}
  E_nm_j(W_i,\theta)k((X_i-x_k)/h_n)
\end{align*}
$(\sqrt{n}\tilde Z_{n,1},\ldots, \sqrt{n}\tilde Z_{n,\ell})$ converges to
a vector of independent standard normal variables, so, since each
$Z_{n,k}$ is eventually equal to $\tilde Z_{n,k}$ times something that
converges to one, $(\sqrt{n}Z_{n,1},\ldots, \sqrt{n}Z_{n,\ell})$ also
converges to a vector of independent standard normal variables.  Thus,
$\min_k \sqrt{n}Z_{n,k}$ converges to the minimum of $\ell$ independent
standard normal variables, which can be made arbitrarily small by making
$\ell$ large.

\end{proof}

\subsection{Rates of Convergence for Slope Parameters}
\label{slope_examp_sec}

In this section of the appendix, I present a counterexample that shows
that a condition along the lines of part (iii) of Assumption
\ref{int_reg_assump} is necessary to obtain the rate of convergence in
Theorem \ref{int_reg_thm}.  As discussed below, a similar counterexample
shows that a condition on the parameter space $\Theta$ such as Assumption
\ref{reg_param_assump} is necessary in Theorem \ref{oneside_thm}.
These counterexamples also show that the first display in
Assumption \ref{x0_assump} cannot be replaced with an assumption that only
takes into account the magnitude of the derivative vector.
Consider an example where
$E_P(W_i^H|X_i=x)=x^2$, $E_P(W_i^L|X_i=x)=-x^2$,
$var(W_i^H|X_i)=var(W_i^L|X_i)=1$, and $X_i$ is has a uniform
distribution on $[-1/2,1/2]$.  Suppose that we use the set of functions
$\{I(s<X_i<s+t)|s\in\mathbb{R}, t\ge 0\}$.

In this case, the identified set is a single point $(0,0)$.  Consider the
sequence of local alternatives given by $\theta_n=(0,b_n)$.  We have, for
all $s,t$ with $-1/2\le s\le s+t\le 1/2$,
\begin{align*}
&E_P[(W_i^H-b_nX_i)I(s<X_i<s+t)]
=E_P[(X_i^2-b_nX_i)I(s<X_i<s+t)]  \\
&=\int_s^{s+t} (x^2-b_n x) \, dx
=\int_s^{s+t} [(x-b_n/2)^2-b_n^2/4] \, dx  \\
&\ge \int_{-t/2}^{t/2} [u^2-b_n^2/4] \, du
=2\left[\frac{1}{3}u^3-\frac{b_n^2}{4}u\right]_{u=0}^{t/2}
=2t\left[\frac{1}{24}t^2-\frac{b_n^2}{8}\right]
\end{align*}
and
\begin{align*}
&var_P[(W_i^H-b_nX_i)I(s<X_i<s+t)]
\ge E_P\{var_P[(W_i^H-b_nX_i)I(s<X_i<s+t)|X_i]\}  \\
&= E_P[I(s<X_i<s+t)]
= t.
\end{align*}
Thus, for $s,t$ such that $E_P[(W_i^H-b_nX_i)I(s<X_i<s+t)]$ is negative,
\begin{align*}
\left|\frac{E_P[(W_i^H-b_nX_i)I(s<X_i<s+t)]}
  {\{var_P[(W_i^H-b_nX_i)I(s<X_i<s+t)]\}^{1/2}}\right|
\le 2t^{1/2}\left|\frac{1}{24}t^2-\frac{b_n^2}{8}\right|_{-}
\le \frac{3^{1/4}}{4}b_n^{1/2}b_n^2.
\end{align*}
A symmetric argument applies to moments based on $W_i^L$.
For some constant $K$, this sequence of local alternatives will be in
$\mathcal{C}_n(\hat c_n)$ if $b_n^{5/2}\le K ((\log n)/n)^{1/2}$
iff. $b_n\le K ((\log n)/n)^{1/5}$.  In contrast, convergence to the
identified set for one sided regression will be at a $((\log n)/n)^{2/5}$
rate if the parameter space $\Theta$ is restricted so that the absolute
value of the slope parameter cannot be too large.

Now consider the one sided regression model of Section \ref{oneside_sec}
with $E_P(W_i^H|X_i=x)=x^2$ and the parameter space $\Theta$ given by
$[0,\infty)\times \mathbb{R}$.  That is, the parameter space $\Theta$
incorporates the prior knowledge that the intercept is nonnegative.
Again, the identified set is the point $(0,0)$, and the Hausdorff distance
between the set estimate $\mathcal{C}_n(\hat c_n)$ and the identified set
will be at least $b_n$ if $\mathcal{C}_n(\hat c_n)$ contains the point
$(0,b_n)$.  By the same argument used above, $(0,b_n)$ will be in
$\mathcal{C}_n(\hat c_n)$ for some sequence $b_n$ going to zero at a
$((\log n)/n)^{1/5}$ rate, so that the rate of convergence of
$\mathcal{C}_n(\hat c_n)$ to the identified set will be no faster than
$((\log n)/n)^{1/5}$, which is slower than the $((\log n)/n)^{2/5}$ rate
given by Theorem \ref{oneside_thm} when the intercept is not restricted.
Note that, in the case where the intercept parameter is not restricted a
priori, the sequence of local alternatives $(0,b_n)$ will still be in the
estimate $\mathcal{C}_n(\hat c_n)$, but the distance of these points to
the identified set will no longer be equal to $b_n$, since the identified
set will contain a point $(\theta'(b_n),b_n)$ for some $\theta'(b_n)$ that
is smaller in magnitued than $b_n$.

\subsection{Covering Number Conditions}\label{covering_sec}

In this section, I state some simple sufficient conditions under which
Assumption \ref{covering_assump} holds.  I first prove that Assumption
\ref{covering_assump} holds under individual bounds on the complexity of the
classes $\mathcal{G}$ and $\{w\mapsto m(w,\theta)|\theta\in\Theta\}$.  The
proof of this result uses Lemma \ref{prod_class_lemma}, stated and proved
at the end of the section.  I then provide examples of classes
$\mathcal{G}$ that satisfy these bounds, and show that the class
$\{w\mapsto m(w,\theta)|\theta\in\Theta\}$ satisfies these bounds in each
of the applications covered in Section \ref{applications_sec}.  Throughout
this section, I define $\mathcal{F}_m\equiv\{w\mapsto
m(w,\theta)|\theta\in\Theta\}$ to be the class of moment functions indexed
by $\theta$.

The following theorem translates bounds on the covering numbers of the
classes $\mathcal{G}$ and $\{w\mapsto m(w,\theta)|\theta\in\Theta\}$ to
the conditions of Assumption \ref{covering_assump}.

\begin{theorem}\label{covering_assump_thm}
Suppose that the classes
$\mathcal{F}_m\equiv\{w\mapsto m(w,\theta)|\theta\in\Theta\}$
and $\mathcal{G}$ are uniformly bounded
and satisfy
$\sup_Q N(\varepsilon,\mathcal{F}_m,L_1(Q))\le A\varepsilon^{-W}$ and
$\sup_Q N(\varepsilon,\mathcal{G},L_1(Q))\le A\varepsilon^{-W}$ for some
$A,W>0$ where the supremum is over all probability measures $Q$.  Then
Assumption \ref{covering_assump} holds.
\end{theorem}
\begin{proof}
The result follows immediately from Lemma \ref{prod_class_lemma}, since
the classes of functions in Assumption \ref{covering_assump} are sums and
products of these bounded classes and bounded classes of constant
functions, which also have polynomial uniform covering numbers.
\end{proof}

With this result in hand, we can verify Assumption \ref{covering_assump}
for a particular model and choice of $\mathcal{G}$
using results stated in \citet{pollard_convergence_1984},
\citet{vaart_weak_1996} and other sources.  For convenience, I do this
here for some choices of $\mathcal{G}$.

\begin{theorem}
Suppose that $\mathcal{F}_m\equiv\{w\mapsto m(w,\theta)|\theta\in\Theta\}$
and $\mathcal{G}$ are uniformly bounded
$\sup_Q N(\varepsilon,\mathcal{F}_m,L_1(Q))\le
A\varepsilon^{-W}$.  Then Assumption \ref{covering_assump} will hold for
the following classes of functions $\mathcal{G}$:
\begin{itemize}
\item[(i)] The class of indicator functions
  $\mathcal{G}=\{x\mapsto I(x\in V)|V\in\mathcal{V}\}$ for any VC class of
  sets $\mathcal{V}$.
\item[(ii)] The class of dilations of a kernel function $k$ given by
  $\mathcal{G}=\{x\mapsto k((x-t)/h)|
  x\in\mathbb{R}^{d_X},h\in\mathbb{R}_+\}$ for any kernel function $k$
  given by $k(x)=r(\|x\|)$ for a decreasing, bounded function $r$ on
  $\mathbb{R}_+$.
\end{itemize}
\end{theorem}
\begin{proof}
The covering number bound for $\mathcal{G}$ in Theorem
\ref{covering_assump_thm} holds by Lemma 25 in
\citet{pollard_convergence_1984} (since a VC class of sets has polynomial
discrimination) for part (i), and by problem 18 in Chapter 2 of
\citet{pollard_convergence_1984} for part (ii).
\end{proof}

See \citet{pollard_convergence_1984} for the definition of a VC class and
examples of VC classes of sets.
The class of all $d_X$ dimensional rectangles falls into this category.
The condition that the class of functions $\mathcal{F}_m=\{w\mapsto
m(w,\theta)|\theta\in\Theta\}$ satisfy the covering number bound $\sup_Q
N(\varepsilon,\mathcal{F}_m,L_1(Q))\le A\varepsilon^{-W}$ can be verified
on a case by case basis using general
results such as those in \citet{pollard_convergence_1984} and
\citet{vaart_weak_1996}.  I do this for the examples in this paper in the
next theorem.

\begin{theorem}
The class of moment functions
$\mathcal{F}_m=\{w\mapsto m(w,\theta)|\theta\in\Theta\}$ satisfies the
covering number bound $\sup_Q N(\varepsilon,\mathcal{F}_m,L_1(Q))\le
A\varepsilon^{-W}$ in all of the models of Section \ref{applications_sec}
as long as the data are bounded and $\Theta$ is compact in the conditional
mean models of Sections \ref{oneside_sec}, \ref{int_reg_sec} and
\ref{selection_sec}.
\end{theorem}
\begin{proof}
The class $\{w\mapsto m(w,\theta)|\theta\in\Theta\}$ has VC subgraph for
all of the models of Section \ref{applications_sec}, so the result follows
from Lemma 25 in \citet{pollard_convergence_1984}.
\end{proof}

The proof of Theorem \ref{covering_assump_thm} uses the following lemma,
which modifies an argument from \citet{vaart_weak_1996}.

\begin{lemma}\label{prod_class_lemma}
Let $\mathcal{F}$, $\mathcal{G}$ and $\mathcal{H}$ be classes of functions
bounded by a fixed constant $B$, and let
$\mathcal{F}\cdot\mathcal{G}+\mathcal{H}
  =\{f\cdot g+h|f\in\mathcal{F}, g\in\mathcal{G}, h\in\mathcal{H}\}$.
Suppose that, for some $A,W>0$,
$\sup_Q N(\varepsilon,\mathcal{F},L_1(Q))\le A\varepsilon^{-W}$, where the
supremum is taken over all probability measures, and that
the same statement holds with $\mathcal{F}$ replaced by $\mathcal{G}$ and
$\mathcal{H}$.  Then
$\sup_Q N(\varepsilon,\mathcal{F}\cdot\mathcal{G}+\mathcal{H},L_1(Q))
  \le A^3(2B+1)^{3W}\varepsilon^{-3W}$, where the
supremum is again taken over all probability measures.
\end{lemma}
\begin{proof}
The result follows from an argument similar to the proof of Theorem
2.10.20 in \citet{vaart_weak_1996}.
Given $\varepsilon>0$ and a probability measure $Q$, let
$k_{\mathcal{F},Q}= N(\varepsilon,\mathcal{F},L_1(Q))
 \le \sup_{Q'} N(\varepsilon,\mathcal{F},L_1(Q'))$ and let
$f_{1,Q},\ldots, f_{k_{\mathcal{F},Q},Q}$ be such that, for all
$f\in\mathcal{F}$, there exists a $f_{i,Q}$ such that
$E_Q|f_{i,Q}(Z_i)-f(Z_i)|\le \varepsilon$ (here, the notation $E_Q f(Z_i)$
refers to the expectation $\int f(z)\, dQ(z)$ of $f(Z_i)$ for $Z_i$
a random variable with distribution $Q$).
Define $k_{\mathcal{G},Q}$,
$k_{\mathcal{H},Q}$, $g_{1,Q},\ldots, g_{k_{\mathcal{G},Q},Q}$ and
$h_{1,Q},\ldots, h_{k_{\mathcal{H},Q},Q}$ similarly.  For any
$fg+h\in \mathcal{F}\cdot\mathcal{G}+\mathcal{H}$, there is some
$j_{\mathcal{F}}$, $j_{\mathcal{G}}$ and $j_{\mathcal{H}}$
such that $E_Q|f_{j_{\mathcal{F}},Q}(Z_i)-f(Z_i)|\le \varepsilon$,
$E_Q|g_{j_{\mathcal{F}},Q}(Z_i)-g(Z_i)|\le \varepsilon$ and
$E_Q|h_{j_{\mathcal{H}},Q}(Z_i)-h(Z_i)|\le \varepsilon$.  We have, for all
$z$,
\begin{align*}
&|f(z)g(z)+h(z)-(f_{j_{\mathcal{F}},Q}(z)g_{j_{\mathcal{G}},Q}(z)+h_{j_{\mathcal{H}},Q}(z))|
  \\
&=|(f(z)-f_{j_{\mathcal{F}},Q}(z))g(z)
+(g(z)-g_{j_{\mathcal{G}},Q}(z))f_{j_{\mathcal{F}},Q}(z)
+h(z)-h_{j_{\mathcal{H}},Q}(z))|  \\
&\le |f(z)-f_{j_{\mathcal{F}},Q}(z)|\cdot |g(z)|
+|g(z)-g_{j_{\mathcal{G}},Q}(z)|\cdot |f_{j_{\mathcal{F}},Q}(z)|
+|h(z)-h_{j_{\mathcal{H}},Q}(z))|  \\
&\le |f(z)-f_{j_{\mathcal{F}},Q}(z)|\cdot B
+|g(z)-g_{j_{\mathcal{G}},Q}(z)|\cdot B
+|h(z)-h_{j_{\mathcal{H}},Q}(z))|
\end{align*}
so that
\begin{align*}
&E_Q|f(Z_i)g(Z_i)+h(Z_i)
  -(f_{j_{\mathcal{F}},Q}(Z_i)g_{j_{\mathcal{G}},Q}(Z_i)+h_{j_{\mathcal{H}},Q}(Z_i))|  \\
&\le (E_Q|f(Z_i)-f_{j_{\mathcal{F}},Q}(Z_i)|
+E_Q|g(Z_i)-g_{j_{\mathcal{G}},Q}(Z_i)|)B
+E_Q|h(Z_i)-h_{j_{\mathcal{H}},Q}(Z_i))|
\le (2B+1)\varepsilon.
\end{align*}
Since $Q$ was arbitrary, it follows that
$\sup_Q
  N((2B+1)\varepsilon,\mathcal{F}\cdot\mathcal{G}+\mathcal{H},L_1(Q))
  \le (\sup_Q N(\varepsilon,\mathcal{F},L_1(Q)))
      \cdot(\sup_Q N(\varepsilon,\mathcal{G},L_1(Q)))
      \cdot (\sup_Q N(\varepsilon,\mathcal{G},L_1(Q)))
  \le A^3\varepsilon^{-3W}$.  Replacing $\varepsilon$ with
$\varepsilon/(2B+1)$ gives the result.
\end{proof}

\subsection{Proofs}\label{proofs_sec}

This section of the appendix contains proofs of the results stated in the
body of the paper.

\begin{proof}[proof of Theorem \ref{coverage_thm}]
If $\Theta_0(P)\nsubseteq \mathcal{C}_n(\hat c_n)$, then, for some
$\theta_0\in\Theta_0(P)$, $\sqrt{n/\log n}T_n(\theta_0)\ge \hat c_n$ so
that for some $g\in\mathcal{G}$,
\begin{align*}
S\left(\frac{\hat\mu_{n,1}(\theta,g)}
    {\hat\sigma_{n,1}(\theta,g)\vee \sigma_n}
  ,\ldots,\frac{\hat\mu_{n,d_Y}(\theta,g)}
    {\hat\sigma_{n,d_Y}(\theta,g)\vee \sigma_n}\right)
\ge \frac{\hat c_n \sqrt{\log n}}{\sqrt{n}}
\end{align*}
so that, for some $j$,
$\frac{\hat\mu_{n,j}(\theta,g)}
    {\hat\sigma_{n,j}(\theta,g)\vee \sigma_n}
\le -\frac{\hat c_n \sqrt{\log n}}{\sqrt{n}} K_{S,1}$.
Since $\theta_0\in\Theta_0(P)$, $E_Pm(W_i,\theta_0)g(X_i)\ge 0$, so this
implies that
\begin{align*}
\frac{\sqrt{n}}{\sqrt{\log n}}
\frac{(E_n-E_P)m(W_i,\theta_0)g(X_i)}
  {\hat\sigma_{n,j}(\theta,g)\vee \sigma_n}
\le -\hat c_n K_{S,1}.
\end{align*}
Thus, $\Theta_0(P)\nsubseteq \mathcal{C}_n(\hat c_n)$ implies that the
above display holds for some $\theta_0$, $g$, and $j$.
If $K$ is large enough so that the conclusion of Lemma
\ref{sd_hat_rate_lemma} holds for $B=K\cdot K_{S,1}.$ and $c$ from that
lemma equal to $K$, the probability that there exist some
$\theta_0\in\Theta_0(P)$ and $g\in \mathcal{G}$
such that this event holds and $\hat c_n$ and $\hat a_n$ is greater
than $K$ will be bounded by a sequence that goes to
zero uniformly in $P\in\mathcal{P}$.

\end{proof}

\begin{proof}[proof of Theorem \ref{consistency_thm}]
If $d_H(\Theta_0(P),\mathcal{C}_n(\hat c_n))>\varepsilon$ and
$\Theta_0(P)\subseteq \mathcal{C}_n(\hat c_n)$, then there exists some
$\theta\in\mathcal{C}_n(\hat c_n)$ such that
$d_H(\theta,\Theta_0(P))>\varepsilon$.
Letting $\delta$ be such that, for all $P\in\mathcal{P}$,
$E_Pm_j(W_i,\theta)g_j(X_i)<-\delta$ for some $j$ and $g\in\mathcal{G}$,
this implies that, once $\hat\sigma_{n,j}(\theta',g)$ is bounded uniformly
in $(\theta',g)$ by some $\bar \sigma$ (this happens with probability
approaching one uniformly in $P\in\mathcal{P}$ by Lemma
\ref{sd_est_lemma}),
\begin{align*}
-T_n(\theta)\le \frac{1}{K_{S,2}\bar \sigma}
 \left(E_nm_j(W_i,\theta)g_j(X_i)\vee 0\right)
\le -\frac{1}{K_{S,2}\bar \sigma}\left(\delta
  -\sup_{\theta',g,k}\left|(E_n-E_P)m_k(W_i,\theta')g_k(X_i)\right|\right).
\end{align*}
The probability that
$\sup_{\theta',g,k}\left|(E_n-E_P)m_k(W_i,\theta')g_k(X_i)\right|
\le \delta/2$ goes to one uniformly in $P\in\mathcal{P}$ by Lemma
\ref{sd_rate_lemma}, and once this holds, the above display will imply
$T_n(\theta)\ge \delta/(2K_{S,2}\bar \sigma)$.  This cannot hold for
$\theta\in\mathcal{C}_n(\hat c_n)$ for $\hat c_n\sqrt{(\log n)/n}\le
\delta/(2K_{S,2}\bar \sigma)$, and the probability of this holding goes to
zero uniformly in $P\in\mathcal{P}$.

\end{proof}

\begin{proof}[proof of Theorem \ref{rate_thm}]
If $d_H(\Theta_0(P),\mathcal{C}_n(\hat c_n))
  >B\left(\frac{\hat c_n^2\log n}{n}\right)^{\gamma/2}$, %
$\Theta_0(P)\subseteq \mathcal{C}_n(\hat c_n)$ and
$d_H(\mathcal{C}_n(\hat c_n,\Theta_0(P))\le \delta$
(the latter two events hold with
probability approaching one uniformly in $P\in\mathcal{P}$ by Theorems
\ref{coverage_thm} and \ref{consistency_thm}), then there exists some
$\theta\in\mathcal{C}_n(\hat c_n)$ such that
$d_H(\theta,\Theta_0(P))
  >B\left(\frac{\hat c_n^2\log n}{n}\right)^{\gamma/2}$.
For this $\theta$ (and $P$), there will be, by Assumption
\ref{rate_assump}, a $g^*\in\mathcal{G}$ and $j^*$ such that
\begin{align*}
\frac{\mu_{P,j^*}(\theta,g^*)}
  {\sigma_{P,j^*}(\theta,g^*)
  \vee \left[B^{1/\gamma}
    \left(\frac{\hat c_n^2\log n}{n}\right)^{\psi/2}\right]}
\le -(C/2)B^{1/\gamma}\left(\frac{\hat c_n^2\log n}{n}\right)^{1/2}
\end{align*}
(replacing $C$ with $C/2$ takes care of the possibility that the infimum
in the assumption is not achieved)
and, by part (ii), for some constant $\eta>0$ that does not depend on $P$,
this will eventually imply
\begin{align*}
\frac{\mu_{P,j^*}(\theta,g^*)}
  {\sigma_{P,j^*}(\theta,g^*)
  \vee (\eta \sigma_n)}
\le -(C/2)B^{1/\gamma}\left(\frac{\hat c_n^2\log n}{n}\right)^{1/2}
\end{align*}
so that, letting $C_1=(C/2) (\eta \wedge 1)$, we will have
$\frac{\mu_{P,j^*}(\theta,g^*)}
  {\sigma_{P,j^*}(\theta,g^*)
  \vee \sigma_n}
\le -C_1B^{1/\gamma}\left(\frac{\hat c_n^2\log n}{n}\right)^{1/2}$.
Since $\theta\in\mathcal{C}_n(\hat c_n)$, we will also have
$T_n(\theta)\le \hat c_n \left(\frac{\log n}{n}\right)^{1/2}$, so that,
for all $g\in\mathcal{G}$ and all $j$,
$\frac{\hat \mu_{n,j}(\theta,g)}
  {\hat \sigma_{n,j}(\theta,g)\vee \sigma_n}
  \ge -K_{S,2}\hat c_n \left(\frac{\log n}{n}\right)^{1/2}$.  By Lemma
\ref{sd_est_lemma}, this will also imply
$\frac{\hat \mu_{n,j}(\theta,g)}
  {\sigma_{P,j}(\theta,g) \vee \sigma_n}
  \ge -\frac{K_{S,2}\hat c_n}{2} \left(\frac{\log n}{n}\right)^{1/2}$
with probability approaching
one uniformly in $P\in\mathcal{P}$.  When these events all hold, we will
have
\begin{align*}
\frac{\hat \mu_{n,j^*}(\theta,g^*)}{\sigma_{P,j^*}(\theta,g^*)\vee \sigma_n}
-\frac{\mu_{P,j^*}(\theta,g^*)}{\sigma_{P,j^*}(\theta,g^*)\vee \sigma_n}
\ge -\frac{K_{S,2}\hat c_n}{2} \left(\frac{\log n}{n}\right)^{1/2}
+C_1 B^{1/\gamma}\left(\frac{\hat c_n^2\log n}{n}\right)^{1/2}
\end{align*}
so that
\begin{align*}
\sup_{\theta\in\Theta, g\in\mathcal{G},j\in\{1,\ldots,j\}}
\frac{\sqrt{n}}{\sqrt{\log n}}
\left|\frac{\hat \mu_{n,j}(\theta,g)}{\sigma_{P,j}(\theta,g)\vee \sigma_n}
-\frac{\mu_{P,j}(\theta,g)}{\sigma_{P,j}(\theta,g)\vee \sigma_n}\right|
\ge \hat c_n (B^{1/\gamma}C_1-K_{S,2}/2).
\end{align*}
Since $\hat c_n$ is bounded away from zero, we can choose $B$ large so
that $\hat c_n (B^{1/\gamma}C_1-K_{S,2}/2)$ is large enough so that the
conclusion of Lemma \ref{sd_rate_lemma} holds with $B$ from that lemma
replaced by $\hat c_n (B^{1/\gamma}C_1-K_{S,2}/2)$.  For this value of $B$,
the probability of the last display holding will go to zero uniformly in
$P\in\mathcal{P}$ so that the desired conclusion will hold.
\end{proof}

\begin{proof}[proof of Theorem \ref{rate_thm_alpha}]
It is sufficient to find a $C$ such that, given $\theta$ and $P$, there
exists a $\theta_0(\theta,P)$, $j_0(\theta,P)$, and a $g\in\mathcal{G}$
such that
\begin{align*}
\frac{\mu_{P,j}(\theta,g)}
  {\sigma_{P,j}(\theta,g)\vee d(\theta,\theta_0(P))^{\psi/\gamma}}
\le -C \|\theta-\theta_0(\theta,P)\|^{1/\gamma}.
\end{align*}
Given $\theta$ and $P$, let $\theta_0(\theta,P)$ and $j_0(\theta,P)$ be
chosen as in Assumption \ref{x0_assump}.  To avoid cumbersome
notation, I will use $\theta_0$ and $j_0$ to denote $\theta_0(\theta,P)$
and $j_0(\theta,P)$ when the dependence on $\theta$ and $P$ is clear.
For this $\theta_0$ and $j_0$, we will have, for $\|x-x_0\|<\eta$,
\begin{align*}
&\bar m_{j_0}(\theta,x,P)
=\bar m_{j_0}(\theta,x,P)-\bar m_{j_0}(\theta_0,x_0,P)  \\
&=[\bar m_{j_0}(\theta,x,P)-\bar m_{j_0}(\theta_0,x,P)]
+[\bar m_{j_0}(\theta_0,x,P)-\bar m_{j_0}(\theta_0,x_0,P)]  \\
&\le \bar m_{\theta,j_0}(\theta^*,x,P)(\theta-\theta_0)
  +C\|x-x_0\|^\alpha
\end{align*}
for some $\theta^*$ between $\theta$ and $\theta_0$.  By Assumptions
\ref{diff_assump} and \ref{x0_assump}, for
$\|\theta-\theta_0\|$ and $\|x-x_0\|$ smaller than some constant that does
not depend on $P$ or $\theta$, this will be less than or equal to
\begin{align*}
-(\eta/2)\left\|\theta-\theta_0\right\|+C\|x-x_0\|^\alpha.
\end{align*}
For $\|x-x_0\|\le [\eta/(4C)]^{1/\alpha}\|\theta-\theta_0\|^{1/\alpha}$,
this is less than or equal to $-(\eta/4)\|\theta-\theta_0\|$.  Thus,
letting $g\in\mathcal{G}$ be as in Assumption \ref{g_rate_assump}
with
$s=x_0$ and $t=[\eta/(4C)]^{1/\alpha}\|\theta-\theta_0\|^{1/\alpha}$ so
that $g(x)\le I(\|x-x_0\|\le
[\eta/(4C)]^{1/\alpha}\|\theta-\theta_0\|^{1/\alpha})$
and
$g(x)\ge C_{\mathcal{G},1}I(\|x-x_0\|\le
[\eta/(4C)]^{1/\alpha}\|\theta-\theta_0\|^{1/\alpha}C_{\mathcal{G},2})$,
we will have
\begin{align}\label{mu_bound_eq}
\mu_{P,j_0}(\theta,g)
=E_P\bar m_{j_0}(\theta,X_i,P)g(X_i)
\le -(\eta/4)\|\theta-\theta_0\| E_Pg(X_i)
\end{align}
and
\begin{align*}
&\sigma_{P,j_0}(\theta,g)
=\left\{var_P[m_{j_0}(W_i,\theta)g(X_i)]\right\}^{1/2}
\le \left\{E_P[m_{j_0}(W_i,\theta)g(X_i)]^2\right\}^{1/2}  \\
&\le \overline Y \overline g^{1/2} \left\{E_Pg(X_i)\right\}^{1/2}.
\end{align*}
The lower bound on $g$ implies that $\{E_Pg(X_i)\}^{1/2}$ is greater than
some constant that does not depend on $P$ times
$\|\theta-\theta_0\|^{d_X/(2\alpha)}
  \ge d(\theta,\theta_0(P))^{d_X/(2\alpha)}$.  Thus, for some constant $K$
that does not depend on $P$,
$\sigma_{P,j_0}(\theta,g)\vee d(\theta,\theta_0(P))^{d_X/(2\alpha)}
\le K \{E_Pg(X_i)\}^{1/2}$.
Thus,
\begin{align*}
&\frac{\mu_{P,j_0}(\theta,g)}
  {\sigma_{P,j_0}(\theta,g)\vee d(\theta,\theta_0(P))^{d_X/(2\alpha)}}
\le \frac{-(\eta/4)}{K}
  \|\theta-\theta_0\| [E_Pg(X_i)]^{1/2}  \\
&\le \frac{-(\eta/4)}{K}\|\theta-\theta_0\|
  C_{\mathcal{G},1}^{1/2}P\left\{\|x-x_0\|\le
    [\eta/(4C)]^{1/\alpha}\|\theta-\theta_0\|^{1/\alpha}
    C_{\mathcal{G},2}\right\}^{1/2}  \\
&\le \frac{-(\eta/4)}{K}\|\theta-\theta_0\|
  C_{\mathcal{G},1}^{1/2}\eta^{1/2}
    \left\{[\eta/(4C)]^{1/\alpha}\|\theta-\theta_0\|^{1/\alpha}
    C_{\mathcal{G},2}\right\}^{d_X/2}
\end{align*}
where the second inequality follows from the lower bound on $g$.
This is equal to a negative constant that does not depend on $P$ times
$\|\theta-\theta_0\|^{(d_X+2\alpha)/(2\alpha)}$, so that Assumption
\ref{rate_assump} holds with $\gamma=2\alpha/(d_X+2\alpha)$ and
$\psi=d_X/(d_X+2\alpha)$.

\end{proof}

\begin{proof}[proof of Theorem \ref{oneside_thm}]
Assumption \ref{diff_assump} holds because $\bar m(\theta,x)$ is linear,
so it remains to verify Assumption \ref{x0_assump}.  Given
$\theta\in\Theta$ and $P\in\mathcal{P}$, let $x_0(\theta,P)$ minimize
$E_P(W_i^H|X_i=x)-\theta_1-x'\theta_{-1}$ over the support of $X_i$, and let
$t(\theta,P)$ be the minimum (the minimum is taken since
$E(W_i^H|X_i=x)-\theta_1-x'\theta_{-1}$ is continuous).  Let
$\theta_0(\theta,P)=(\theta_1+t(\theta,P),\theta_{-1})$.
Then
$\bar m(\theta_0(\theta,P),x,P)
  =E(W_i^H|X_i=x)-\theta_1-t(\theta,P)-x'\theta_{-1}$ so
that $\theta_0(\theta,P)\in\Theta_0(P)$ and
$\bar m(\theta_0(\theta,P),x_0(\theta,P),P)=0$.  We have
\begin{align*}
&\bar m_\theta(\theta_0(\theta,P),x_0(\theta,P),P) (\theta-\theta_0(\theta,P))
=-(1,x_0(\theta,P)') (-t(\theta,P),0,\ldots,0)'  \\
&=t(\theta,P)=-\|(t(\theta,P),0,\ldots,0)'\|
=-\|\theta-\theta_0(\theta,P)\|
\end{align*}
where the second to last equality holds because $t(\theta,P)$ is negative
by definition of the identified set.
The Holder continuity part of Assumption \ref{x0_assump} is immediately
implied by Assumption \ref{reg_holder_assump}.  Under Assumption
\ref{reg_2diff_assump}, $x_0(\theta,P)$ must be on the interior of the
support of $X_i$ by part (ii) of this assumption.
Thus, $x_0(\theta,P)$ is an interior minimum of the twice differentiable
function $x\mapsto E_P(W_i^H|X_i=x)-\theta_1-x'\theta_{-1}$, so the first
derivative of this function at $x_0(\theta,P)$ is zero.
This and a second order mean value expansion of this function around
$x_0(\theta,P)$ imply the
Holder continuity part of Assumption \ref{x0_assump} with $C$ a bound on
the norm of the second derivative matrix.

\end{proof}

\begin{proof}[proof of Theorem \ref{int_reg_thm}]
Everything is the same as in the proof of Theorem \ref{oneside_thm} except
for the verification of the first part of Assumption \ref{x0_assump}.
For any $\theta$, either $(\theta_1',\theta_{2})$ is in $\Theta_0(P)$ for
some $\theta'$, in which case the same argument to verify
Assumption \ref{x0_assump} goes through, or  $\theta_{2}>\overline
\theta_{2}$ or $\theta_{2}<\underline \theta_{2}$, where
$\overline \theta_2\equiv \sup \{\theta_{2}
  |(\theta_1,\theta_{2})\in\Theta_0(P) \text{ some $\theta_1$}\}$ and
$\underline \theta_2\equiv \inf \{\theta_{2}
  |(\theta_1,\theta_{2})\in\Theta_0(P) \text{ some $\theta_1$}\}$.
Suppose that $\theta_2>\overline \theta_2$ (the case where
$\theta_2<\underline \theta_2$ is symmetric).  Then, for some $\theta_1'$,
we have $(\theta_1',\overline \theta_2)\in \Theta_0(P)$, and, for some
$x_{0,2}\le x_{0,1}$,
$E(W_i^H|X_i=x_{0,1})=\theta_1'+x_{0,1}\overline \theta_2$
and
$E(W_i^L|X_i=x_{0,2})=\theta_1'+x_{0,2}\overline \theta_2$.  We have
$\bar m_{\theta,1}(\theta,x,P)=-(1,x)$ and
$\bar m_{\theta,2}(\theta,x,P)=(1,x)$, so that
\begin{align*}
\bar m_{\theta,1}(\theta,x_{0,1},P)(\theta-(\theta_1',\overline \theta_2))
=-(1,x_{0,1})(\theta-(\theta_1',\overline \theta_2))
\end{align*}
and
\begin{align*}
\bar m_{\theta,2}(\theta,x_{0,2},P)(\theta-(\theta_1',\overline \theta_2))
=(1,x_{0,2})(\theta-(\theta_1',\overline \theta_2)).
\end{align*}
If the sum of the expressions in these two displays is less than
$-2\eta\|\theta-(\theta_1',\overline \theta_2)\|$,
at least one of them must be less than $-\eta\|\theta-(\theta_1',\overline
\theta)\|$, so it suffices to bound
\begin{align*}
[(1,x_{0,2})-(1,x_{0,1})](\theta-(\theta_1',\overline \theta_2))
/\|\theta-(\theta_1',\overline \theta_2)\|
=-\frac{(x_{0,1}-x_{0,2})(\theta_2-\overline \theta_2)}
  {\left[(\theta_1-\theta_1')^2
  +(\theta_2-\overline \theta_2)^2\right]^{1/2}}.
\end{align*}
For this, it suffices to bound $x_{0,1}-x_{0,2}$ away from zero and
$|\theta_1-\theta_1'|/|\theta_2-\overline \theta_2|$ away from infinity.

$x_{0,1}-x_{0,2}$ is bounded away from zero by parts (ii) and (iii) of
Assumption \ref{int_reg_assump}.  For parameter values where
$|\theta_1-\theta_1'|/|\theta_2-\overline \theta_2|$ is large, we can use
another argument.
Note that
\begin{align*}
\frac{-(1,x_{0,1})(\theta-(\theta_1',\overline \theta_2))}
  {\|\theta-(\theta_1',\overline \theta_2)\|}
=-\frac{(\theta_1-\theta_1')+x_{0,1}(\theta_2-\overline \theta_2)}
  {\|\theta-(\theta_1',\overline \theta_2)\|}
=-\frac{(\theta_1-\theta_1')/(\theta_2-\overline \theta_2)+x_{0,1}}
  {\left[(\theta_1-\theta_1')^2/(\theta_2-\overline \theta_2)^2
    + 1\right]^{1/2}}
\end{align*}
and, similarly,
\begin{align*}
\frac{(1,x_{0,2})(\theta-(\theta_1',\overline \theta_2))}
  {\|\theta-(\theta_1',\overline \theta_2)\|}
=\frac{(\theta_1-\theta_1')/(\theta_2-\overline \theta_2)+x_{0,2}}
  {\left[(\theta_1-\theta_1')^2/(\theta_2-\overline \theta_2)^2
    + 1\right]^{1/2}}.
\end{align*}
For $|\theta_1-\theta_1'|/|\theta_2-\overline \theta_2|
  >2\max\{|x_{0,1}|,|x_{0,2}|,1\}$, one of these displays will be less
than $-1/4$.

\end{proof}

\begin{proof}[proof of Theorem \ref{quantile_oneside_thm}]
For Assumption \ref{diff_assump}, note that
\begin{align*}
&\bar m_{\theta}(\theta,x,P)
=\frac{d}{d\theta}E_P[\tau-I(W_i^H\le \theta_1+X_i'\theta_{-1})|X_i=x]
=-\frac{d}{d\theta}P(W_i^H\le \theta_1+X_i'\theta_{-1}|X_i=x)  \\
&=-f_{W_i^H|X_i}(\theta_1+x'\theta_{-1}|x)(1,x').
\end{align*}
This is continuous as a function of $\theta$ uniformly in $(\theta,x,P)$
by Assumption \ref{quantile_density_assump} and the bound on the support
of $X_i$.

To verify the first part of Assumption \ref{x0_assump}, let
$x_0(\theta,P)$, $t(\theta,P)$ and $\theta_0(\theta,P)$ be defined as in
the proof of Theorem \ref{oneside_thm}, but with $E_P(W_i^H|X_i=x)$
replaced by $q_{\tau,P}(W_i^H|X_i=x)$.  Then
$\theta_0(\theta,P)\in\Theta_0(P)$ and
\begin{align*}
&\bar m(\theta_0(\theta,P),x_0(\theta,P),P)
=\tau-P(W_i^H\le \theta_1+t(\theta,P)+X_i'\theta_{-1}|X_i=x_0(\theta,P))  %
=0
\end{align*}
since
$q_{\tau,P}(W_i^H|X_i=x_0(\theta,P))
=\theta_1+t(\theta,P)+x_0(\theta,P)'\theta_{-1}$.  We also have
\begin{align*}
&m_{\theta}(\theta_0(\theta,P),x_0(\theta,P),P)
(\theta-\theta_0(\theta,P))
=-f_{W_i^H|X_i}(\theta_1+x'\theta_{-1}|x)(1,x')
(-t(\theta,P),0,\ldots,0)'  \\
&=f_{W_i^H|X_i}(\theta_1+x'\theta_{-1}|x)t(\theta,P)
=-f_{W_i^H|X_i}(\theta_1+x'\theta_{-1}|x)
\|\theta-\theta_0(\theta,P)\|
\le -\underline f \|\theta-\theta_0(\theta,P)\|.
\end{align*}

For the second part of Assumption \ref{x0_assump}, note that, since
$\theta_0=\theta_0(\theta,P)\in\Theta_0(P)$,
\begin{align*}
&\bar m(\theta_0,x,P)
=\tau-P(W_i^H\le \theta_{0,1}+X_i'\theta_{0,-1}|X_i=x)  \\
&=\tau-P(W_i^H\le q_{\tau,P}(W_i^H|X_i=x)|X_i=x)  \\
&+P(\theta_{0,1}+X_i'\theta_{0,-1}\le W_i^H
  \le q_{\tau,P}(W_i^H|X_i=x) |X_i=x)  \\
&=P(\theta_{0,1}+X_i'\theta_{0,-1}\le W_i^H
  \le q_{\tau,P}(W_i^H|X_i=x) |X_i=x).
\end{align*}
For $\|x-x_0\|$ small enough, the distance between
$\theta_{0,1}+x'\theta_{0,-1}$ and $q_{\tau,P}(W_i^H|X_i=x)$ will be less
than the $\eta$ in Assumption \ref{quantile_density_assump}.  For $x$ such
that this holds,
\begin{align*}
&|\bar m(\theta_0,x,P)-\bar m(\theta_0,x_0,P)|
=\bar m(\theta_0,x,P)  \\
&=P(\theta_{0,1}+X_i'\theta_{0,-1}\le W_i^H
  \le q_{\tau,P}(W_i^H|X_i=x) |X_i=x)  \\
&\le \overline f [q_{\tau,P}(W_i^H|X_i=x)-\theta_{0,1}-x'\theta_{0,-1}]
\\
&=\overline f\{ [q_{\tau,P}(W_i^H|X_i=x)-\theta_{0,1}-x'\theta_{0,-1}]
- [q_{\tau,P}(W_i^H|X_i=x_0)-\theta_{0,1}-x_0'\theta_{0,-1}] \}.
\end{align*}
Under Assumption \ref{quantile_holder_assump}, the second part of
Assumption \ref{x0_assump} then follows immediately since, for $\alpha\le
1$ and $\|x-x_0\|$ small enough,
$\|(x-x_0)'\theta_{0,-1}\|\le \|\theta_{0,-1}\|\|x-x_0\|
\le \|\theta_{0,-1}\|\|x-x_0\|^\alpha$
so that the expression in the above display is bounded by
$\overline f(C+\|\theta_{0,-1}\|)\|x-x_0\|^\alpha$.
Under Assumption
\ref{quantile_2diff_assump}, Assumption \ref{x0_assump} follows from a
second order mean value expansion of $q_{\tau,P}(W_i^H|X_i=x_0)$ since
$x_0$ is on the interior of the support of $X_i$.

\end{proof}

\begin{proof}[proof of Theorem \ref{quantile_int_reg_thm}]
Everything is the same as in the proof of Theorem
\ref{quantile_oneside_thm} except for the verification of the first part
of Assumption \ref{x0_assump}.  Verifying this condition uses a similar
argument to the one in Theorem \ref{int_reg_thm} for mean regression.
For
any $\theta$, either $(\theta_1',\theta_2)\in \Theta_0(P)$ for some
$\theta'$, in which case the same argument to verify Assumption
\ref{x0_assump} goes through, or $\theta_2>\overline \theta_2$ or
$\theta_2<\underline \theta_2$, where $\overline \theta_2$ and $\underline
\theta_2$ are defined as in the proof of Theorem \ref{int_reg_thm}
($\overline \theta_2\equiv \sup \{\theta_{2}
  |(\theta_1,\theta_{2})\in\Theta_0(P) \text{ some $\theta_1$}\}$ and
$\underline \theta_2\equiv \inf \{\theta_{2}
  |(\theta_1,\theta_{2})\in\Theta_0(P) \text{ some $\theta_1$}\}$).  If
$\theta_2>\overline \theta_2$ (a symmetric argument applies when
$\theta_2<\underline \theta_2$), then, for some $\theta_1'$,
$(\theta_1',\overline\theta_2)\in\Theta_0(P)$ and some
$x_{0,2}\le x_{0,1}$,
$q_{\tau,P}(W_i^H|X_i=x_{0,1})=\theta_1'+x_{0,1}\overline \theta_2$ and
$q_{\tau,P}(W_i^L|X_i=x_{0,2})=\theta_1'+x_{0,2}\overline \theta_2$.
We have
$\bar m_{\theta,1}(\theta,x_{0,1},P)
=-f_{W_i^H|X_i}(\theta_1+x_{0,1}'\theta_2|x_{0,1})(1,x_{0,1})$ and
$\bar m_{\theta,2}(\theta,x_{0,2},P)
=f_{W_i^L|X_i}(\theta_1+x_{0,2}'\theta_2|x_{0,1})(1,x_{0,1})$, so
\begin{align*}
\bar m_{\theta,1}(\theta,x_{0,1},P)(\theta-(\theta_1',\overline \theta_2))
=-f_{W_i^H|X_i}(\theta_1+x_{0,1}'\theta_2|x_{0,1})(1,x_{0,1})
(\theta-(\theta_1',\overline \theta_2))
\end{align*}
and
\begin{align*}
\bar m_{\theta,2}(\theta,x_{0,2},P)(\theta-(\theta_1',\overline \theta_2))
=f_{W_i^L|X_i}(\theta_1+x_{0,2}'\theta_2|x_{0,1})(1,x_{0,2})
(\theta-(\theta_1',\overline \theta_2)).
\end{align*}
Letting $a_1$ be the expression in the first display above, and $a_2$ the
expression in the second display above, note that, if
\begin{align*}
[f_{W_i^H|X_i}(\theta_1+x_{0,1}'\theta_2|x_{0,1})]^{-1} \cdot a_1
+[f_{W_i^L|X_i}(\theta_1+x_{0,2}'\theta_2|x_{0,1})]^{-1} \cdot a_2
\le -\frac{2\eta}{\underline f} \|\theta-(\theta_1',\overline \theta_2)\|,
\end{align*}
then either $a_1\le -\eta \|\theta-(\theta_1',\overline \theta_2)\|$ or
$a_2\le -\eta \|\theta-(\theta_1',\overline \theta_2)\|$.  Thus, it
suffices to bound the expression on the left hand side of the above
display divided by $\|\theta-(\theta_1',\overline \theta_2)\|$
away from zero from above.  The left hand side of the above display
divided by $\|\theta-(\theta_1',\overline \theta_2)\|$ is equal to
\begin{align*}
[(1,x_{0,2})-(1,x_{0,1})](\theta-(\theta_1',\overline \theta_2))
/\|\theta-(\theta_1',\overline \theta_2)\|
=-\frac{(x_{0,1}-x_{0,2})(\theta_2-\overline \theta_2)}
{[(\theta_1-\theta_1')^2+(\theta_2-\theta_2')^2]^{1/2}}.
\end{align*}
By the same argument as in the proof of Theorem \ref{int_reg_thm}, this is
bounded away from zero from above for
$|\theta_1-\theta_1'|/|\theta_2-\overline \theta_2|$ bounded away from
infinity since $x_{0,1}-x_{0,2}$ is bounded away from zero, and, for
$|\theta_1-\theta_1'|/|\theta_2-\overline \theta_2|$ large enough, either
$\bar m_{\theta,1}(\theta,x_{0,1},P)(\theta-(\theta_1',\overline
\theta_2))$ or
$\bar m_{\theta,2}(\theta,x_{0,2},P)(\theta-(\theta_1',\overline
\theta_2))$ will be less than the same negative constant for all
$P\in\mathcal{P}$.

\end{proof}

\begin{proof}[proof of Theorem \ref{coverage_thm_selection}]
The result follow immediately from Theorem \ref{coverage_thm}.
\end{proof}

\begin{proof}[proof of Theorem \ref{rate_thm_selection}]
For the case where Assumption \ref{ident_pos_assump} holds, the result
follows by verifying the conditions of Theorem \ref{rate_thm} with $g$ a
function that is positive only on $[\underline x,\overline x]$.  For the
other cases, the result will follow by verifying the conditions of Theorem
\ref{rate_thm_alpha} once we show that these models can be transformed so
that Assumption \ref{ident_fin_assump} holds with $\phi_x$ in the
transformed model equal to zero and, under Assumption
\ref{ident_fin_assump} on the original model, $\phi_m$ in the transformed
model equal to $\phi_m/(\phi_x+1)$ and, under Assumption
\ref{ident_inf_assump} (and $d_X=1$) on the original model, $\phi_m$ in
the transformed model equal to $\phi_m/(\phi_x-1)$.
(Assumption \ref{g_rate_assump} is invariant to taking the same invertible
monotonic transformation of each element of $X_i$, since we can replace
$\|\cdot\|$ in that assumption with the supremum norm, and then the sets
involved are $d_X$ dimensional boxes, and the set of all $d_X$-dimensional
boxes is invariant to such transformations.  This holds even for the
transformations used under Assumption \ref{ident_inf_assump} in which
infinity is taken to a finite support point by taking $t$ in Assumption
\ref{g_rate_assump} to be large enough so that the largest value of any
component of $X_i$ in the sample is contained in the $d_X$-dimensional
box.)

Suppose that Assumption \ref{ident_fin_assump} holds for some $\phi_m$ and
$\phi_x$.  Then, for any $t\in\mathbb{R}$ with each element less than $\eta_X$
\begin{align*}
&P(0<x_{0,k}-X_{i,k}<t_k \text{ all $k$})
=P(x_0-t<X_i<x_0)    \\
&\ge \frac{1}{C}\int_{x_{0,1}-t_1}^{x_0} \cdots
  \int_{x_{0,d_X}-t_{d_X}}^{x_0}
  \prod_{k=1}^d |x_{0,k}-x_k|^{\phi_x} \, dx_1\cdots dx_{d_X}
=\frac{1}{C}\prod_{k=1}^d \frac{t_k^{\phi_x+1}}{\phi_x+1}
\end{align*}
so that
\begin{align*}
&P(x_{0,k}-t_k<x_{0,k}-(x_{0,k}-X_{i,k})^{(\phi_x+1)}<x_{0,k}
  \text{ all $k$})
=P(0<x_{0,k}-X_{i,k}<t_k^{1/(\phi_x+1)}
  \text{ all $k$})  \\
&\ge \frac{1}{C}\prod_{k=1}^d \frac{t_k}{\phi_x+1}.
\end{align*}
Thus, the random variable $V_i$ defined to have $k$th element
$x_{0,k}-(x_{0,k}-X_{i,k})^{(\phi_x+1)}$ for $x_0-\eta_X<X_i<x_0$ and $X_{i,k}$
otherwise will satisfy part (ii) of Assumption \ref{ident_fin_assump} (for
a different value of $\eta_X$) with
$\phi_x$ equal to zero for the transformed variable.  To get the
conditional mean of the transformed model, note that, for
$x_0-\eta_X<X_i<x_0$,
\begin{align*}
&E_P(W_i^H|V_i=v)
=E_P(W_i^H|x_{0,k}-(x_{0,k}-X_{i,k})^{(\phi_x+1)}=v_k \text{ all $k$})  \\
&=E_P(W_i^H|x_{0,k}-X_{i,k}=(x_{0,k}-v_k)^{1/(\phi_x+1)}
  \text{ all $k$})
=E_P(W_i^H|X_{i,k}=x_{0,k}-(x_{0,k}-v_k)^{1/(\phi_x+1)}
  \text{ all $k$})  \\
&\le C\|((x_{0,1}-v_1)^{1/(\phi_x+1)},\ldots,
  (x_{0,d_X}-v_{d_X})^{1/(\phi_x+1)})\|^{\phi_m}
\le Cd_X^{\phi_m}\|x_0-v\|^{\phi_m/(\phi_x+1)}.
\end{align*}
Thus, Assumption \ref{ident_fin_assump} will hold for the transformed
model with $X_i$ replaced with $V_i$ and $\phi_m$ in the transformed model
equal to $\phi_m/(\phi_x+1)$ and $\phi_x$ in the transformed model equal
to zero.

If Assumption \ref{ident_inf_assump} holds for some $\phi_m$ and $\phi_x$,
then, for $t$ greater than $K_X$ (here $d_X=1$),
\begin{align*}
P(X_i\ge t)
\ge \frac{1}{C}\int_{t}^{\infty}
  x^{-\phi_x} \, dx
=\frac{1}{C(\phi_x-1)} t^{1-\phi_x}.
\end{align*}
Thus,
\begin{align*}
&P(K_X+1-1/(X_i-K_X+1)\ge K_X+1-t)
=P(-1/(X_i-K_X+1)\ge -t)  \\
&=P(1/(X_i-K_X+1)\le t)
=P(1/t\le X_i-K_X+1)  \\
&=P(K_X-1+1/t\le X_i)
\ge \frac{1}{C(\phi_x-1)} (K_X-1+1/t)^{1-\phi_x}
\ge \frac{2^{1-\phi_x}}{C(\phi_x-1)} t^{\phi_x-1}
\end{align*}
where the last inequality holds for $t$ small enough so that
$1/t\ge K_X-1$.
It follows that part (ii) of Assumption \ref{ident_fin_assump} holds with
$\phi_x$
in that assumption replaced by $\phi_x-2$ for the transformed random
variable $V_i$ given by $V_i=K_X+1-1/(X_i-K_X+1)$ for
$X_i>K_X$ and $V_i=X_i$ otherwise.
Here, $x_0$ from Assumption \ref{ident_fin_assump} is equal to $K_X+1$ in
the transformed model.
As for the
conditional mean of the transformed model, we have, for $v$ close enough
to $K_X+1$,
\begin{align*}
&E_P(W_i^H|V_i=v)
=E_P(W_i^H|K_X+1-1/(X_i-K_X+1)=v)  \\
&=E_P(W_i^H|-1/(X_i-K_X+1)=v-1-K_X)
=E_P(W_i^H|X_i-K_X+1=-1/(v-1-K_X))  \\
&=E_P(W_i^H|X_i=-1/(v-1-K_X)+K_X-1)
\le C(-1/(v-1-K_X)+K_X-1)^{-\phi_m}  \\
&\le 2C(1/(1+K_X-v))^{-\phi_m}
= 2C(1+K_X-v)^{\phi_m}
\end{align*}
so that part (i) of Assumption \ref{ident_fin_assump} holds with the same
$\phi_m$.

\end{proof}

\begin{proof}[proof of Theorem \ref{upper_rate_bdd_thm}]
Let $\theta_n$ be a sequence converging to $\theta_0$ such that, for some
$\varepsilon>0$,
$d_H(\theta_n,\Theta_0(P))=n^{-\alpha/(2d_X+2\alpha)}\varepsilon$,
for large enough $n$, (conditions on how small $\varepsilon$ is will be
stated below).
Such a sequence exists by part (iv) of Assumption
\ref{smoothness_upper_assump}.
For each $n$, let $\theta_0'(n)\in \delta\Theta_0(P)$ be such that
$d_H(\theta_n,\theta_0'(n))\le 2 n^{-\alpha/(2d_X+2\alpha)}\varepsilon$
(doubling the distance to the identified set covers the possibility that
the infimum is not achieved).  For each $j$, we have, for some
$x_0\in\mathcal{X}_0(\theta_0'(n))$ and
some $\theta_n^*$ between $\theta_n$ and $\theta_0'(n)$,
\begin{align}\label{obj_bound_eq}
&\bar m_j(\theta_n,x,P)
=\bar m_j(\theta_n,x,P)-\bar m_j(\theta_0'(n),x_0,P)  \notag  \\
&=[\bar m_j(\theta_n,x,P)-\bar m_j(\theta_0'(n),x,P)]
+[\bar m_j(\theta_0'(n),x,P)-\bar m_j(\theta_0'(n),x_0,P)]  \notag  \\
&=\bar m_{\theta,j} (\theta_n^*,x,P)(\theta_n-\theta_0'(n))
+[\bar m_j(\theta_0'(n),x,P)-\bar m_j(\theta_0'(n),x_0,P)]  \notag  \\
&\ge - 2K n^{-\alpha/(2d_X+2\alpha)}\varepsilon
+\eta \min_{x_0\in \mathcal{X}_0(\theta_0'(n))}
(\|x-x_0\|^{\alpha}\wedge \eta)
\end{align}
where $K$ is a bound on the derivative.  For $n$ large enough, the last
line of the above display is negative only for $x$ such that, for some
$x_0\in\mathcal{X}_0(\theta'_0(n))$,
$\|x-x_0\|
<\left(\frac{2K\varepsilon}{\eta}\right)^{1/\alpha}
  n^{-1/(2d_X+2\alpha)}$.
This will imply, letting $\overline g$ be an upper bound for functions in
$\mathcal{G}$ and $K_1$ an upper bound for the number of elements in
$\mathcal{X}_0(\theta'_0(n))$,
\begin{align*}
&\mu_{P,j}(\theta_n,g)=E_P\bar m_j(\theta_n,X_i,P)g(X_i)
\ge - 2K n^{-\alpha/(2d_X+2\alpha)}\varepsilon \overline g
  P(\bar m_j(\theta_n,X_i,P)<0)  \\
&\ge - 2K n^{-\alpha/(2d_X+2\alpha)}\varepsilon \overline g
  \sum_{x_0\in\mathcal{X}_0(\theta'_0(n))}
  P\left(\|X_i-x_0\|<\left(\frac{2K\varepsilon}{\eta}\right)^{1/\alpha}
  n^{-1/(2d_X+2\alpha)}\right)  \\
&\ge - 2K n^{-\alpha/(2d_X+2\alpha)}\varepsilon K_1 \overline g
  \overline f 2^{d_X}\left(\frac{2K\varepsilon}{\eta}\right)^{d_X/\alpha}
  n^{-d_X/(2d_X+2\alpha)}  \\
&= - 2K \varepsilon K_1 \overline g \overline f
  2^{d_X}\left(\frac{2K\varepsilon}{\eta}\right)^{d_X/\alpha} n^{-1/2}.
\end{align*}
Here, the first inequality follows for large enough $n$ since $\bar
m_j(\theta_n,x,P)\ge - 2K n^{-\alpha/(2d_X+2\alpha)}\varepsilon$
eventually by the argument above.

If $d_H(\mathcal{C}_{n,\omega}(\hat c_n),\Theta_0(P))<\varepsilon
n^{-\alpha/(2d_X+2\alpha)}$, then
$\theta_n\notin\mathcal{C}_{n,\omega}(\hat c_n)$, so that
$T_{n,\omega}(\theta_n)> \hat c_n n^{-1/2}\ge \underline c n^{-1/2}$ where
$\underline c$ is a lower bound for $\hat c_n$.  Then, for some $j$ and
$g$, we will have, letting $K_{S,1}$ be as in Assumption \ref{S_assump},
$\omega(\theta_n,g)\hat \mu_{n,j}(\theta_n,g)
\le -K_{S,1} \underline c n^{-1/2}$ so that, letting $\overline \omega$ be
an upper bound for $\omega(\theta,g)$,
$n^{1/2} \hat \mu_{n,j}(\theta_n,g)
  \le -K_{S,1} \underline c /\overline \omega$.  For large enough $n$, we
will also have
$n^{1/2}\mu_{P,j}(\theta_n,g)\ge - 2K \varepsilon K_1 \overline g \overline f
  2^{d_X}\left(\frac{2K\varepsilon}{\eta}\right)^{d_X/\alpha}$.  This will
imply
\begin{align*}
n^{1/2} \left\{\hat \mu_{n,j}(\theta_n,g)
-\left[\mu_{P,j}(\theta_n,g)\wedge 0\right]\right\}
\le -K_{S,1} \underline c /\overline \omega
+2K \varepsilon K_1 \overline g \overline f
  2^{d_X}\left(\frac{2K\varepsilon}{\eta}\right)^{d_X/\alpha}
\end{align*}
so that $n^{1/2} \left\{\hat \mu_{n,j}(\theta_n,g)
-\left[\mu_{P,j}(\theta_n,g)\wedge 0\right]\right\}$ is bounded away from
zero from above by a negative constant when this event holds for small
enough $\varepsilon$.  Thus, it suffices to show that, for any $\delta>0$,
$n^{1/2} \inf_{g\in\mathcal{G}} \left\{\hat \mu_{n,j}(\theta_n,g)
-\left[\mu_{P,j}(\theta_n,g)\wedge 0\right]\right\}>-\delta$ with
probability approaching one.

We have, for any $r>0$,
\begin{align*}
&n^{1/2} \inf_{g\in\mathcal{G}} \left\{\hat \mu_{n,j}(\theta_n,g)
-\left[\mu_{P,j}(\theta_n,g)\wedge 0\right]\right\}  \\
&\ge n^{1/2} \inf_{g\in\mathcal{G}} \hat \mu_{n,j}(\theta_n,g)
I(\mu_{P,j}(\theta_n,g)>r)
+n^{1/2} \inf_{g\in\mathcal{G}} \left\{\hat \mu_{n,j}(\theta_n,g)
-\mu_{P,j}(\theta_n,g)\right\}
I(\mu_{P,j}(\theta_n,g)\le r).
\end{align*}
The first term is greater than zero with probability approaching one since
$\hat \mu_{n,j}(\theta,g)$ converges to $\mu_{P,j}(\theta,g)$ at a
root-$n$ rate uniformly over $(\theta,g)$ by standard arguments
(e.g. Theorem 2.5.2 in \citet{vaart_weak_1996}).

As for the second term,
note that, for any $\delta_1,\delta_2>0$ with $\delta_1^\alpha\le \eta$,
$\bar m_j(\theta_n,x,P)$ will be greater than or equal to
$-\delta_2 I(d_H(x,\mathcal{X}_0(\theta_0'(n)))<\delta_1 )
+(\eta\delta_1^\alpha-\delta_2)
   I(d_H(x,\mathcal{X}_0(\theta_0'(n)))\ge \delta_1 )$
for large enough $n$ by (\ref{obj_bound_eq}).  To simplify notation,
define the sets
$A_{n,\delta_1}=\{x|d_H(x,\mathcal{X}_0(\theta_0'(n)))<\delta_1\}$.
Using this notation, the above observation implies that, for $n$ greater
than some constant that depends on $\delta_1$,
\begin{align*}
\mu_{P,j}(\theta_n,g)=E_P \bar m_j(\theta_n,X_i,P)g(X_i)
\ge -\delta_2 E_P g(X_i)I(X_i\in A_{n,\delta_1} )
+(\eta\delta_1^\alpha-\delta_2) E_P g(X_i)I(X_i\notin A_{n,\delta_1} ).
\end{align*}
If $\mu_{P,j}(\theta_n,g)\le r$, then this means that
\begin{align*}
(\eta\delta_1^\alpha-\delta_2) E_P g(X_i)I(X_i\notin A_{n,\delta_1} )
\le \delta_2 E_P g(X_i)I(X_i\in A_{n,\delta_1} ) + r
\end{align*}
where, as above, $K_1$ is an upper bound for the number of elements in
$\mathcal{X}_0(\theta_0'(n))$.
Thus, for $\mu_{P,j}(\theta_n,g)\le r$, and $n$ larger than some constant
that depends only on $\delta_1$, letting $\overline g$ be a bound for
$g(X_i)$ and $M$ a bound for $m_j(W_i,\theta)$,
\begin{align*}
&E_P [m_j(W_i,\theta_n)g(X_i)]^2\le \overline g M^2 E_P g(X_i)
=\overline g M^2 [E_P g(X_i)I(X_i\notin A_{n,\delta_1} )
+E_P g(X_i)I(X_i\in A_{n,\delta_1} )]  \\
&\le \overline g M^2
  \{[\delta_2E_P g(X_i)I(X_i\in A_{n,\delta_1} ) + r]
    /(\eta\delta_1^\alpha-\delta_2)
  +E_P g(X_i)I(X_i\in A_{n,\delta_1} )\}  \\
&=\overline g M^2
\left[
  \left(\frac{\delta_2}{\eta\delta_1^\alpha-\delta_2}+1\right)
    E_P g(X_i)I(X_i\in A_{n,\delta_1} )
  +\frac{r}{\eta\delta_1^\alpha-\delta_2}
\right]  \\
&\le \overline g M^2
\left[
  \left(\frac{\delta_2}{\eta\delta_1^\alpha-\delta_2}+1\right)
    \overline g K_1 (2\delta_1)^{d_X}
  +\frac{r}{\eta\delta_1^\alpha-\delta_2}
\right]
\end{align*}
By choosing $r$, $\delta_1$, and $\delta_2$ so that $\delta_1$,
$r/(\eta\delta_1^\alpha-\delta_2)$ and
$\delta_2/(\eta\delta_1^\alpha-\delta_2)$ are small,
we can make the last line of the display less than any $\delta_3>0$.
Then, for $n$ large enough, $\mu_{P,j}(\theta_n,g)\le r$ will imply
$var_P [m(W_i,\theta_n)g(X_i)]\le \delta_3$, so that
\begin{align*}
&n^{1/2} \inf_{g\in\mathcal{G}} \left\{\hat \mu_{n,j}(\theta_n,g)
-\mu_{P,j}(\theta_n,g)\right\}
I(\mu_{P,j}(\theta_n,g)\le r).  \\
&\ge n^{1/2} \inf_{g\in\mathcal{G}} \left\{\hat \mu_{n,j}(\theta_n,g)
-\mu_{P,j}(\theta_n,g)\right\}
I(var_P [m(W_i,\theta_n)g(X_i)]\le \delta_3).
\end{align*}
This can be made arbitrarily small in magnitude by the stochastic
asymptotic equicontinuity of $n^{1/2}(E_n-E_P)m(W_i,\theta)g(X_i)$ with
respect to the covariance semimetric
$\rho((\theta,g),(\theta',g'))
=var_P [m(W_i,\theta)g(X_i)-m(W_i,\theta')g'(X_i)]$
as a sequence of processes indexed by $(\theta,g)$.  Letting
$\tilde g(x)=0$ be the zero function and $\tilde \theta$ an arbitrary
value in $\Theta$, the last line of the above display is equal to
\begin{align*}
n^{1/2} \inf_{g\in\mathcal{G}} \left\{
(E_n-E)m_j(W_i,\theta_n)g_j(X_i)
-(E_n-E)m_j(W_i,\tilde \theta)\tilde g_j(X_i)\right\}
I(\rho((\theta_n,g),(\tilde \theta, \tilde g))\le \delta_3).
\end{align*}
By making $\delta_3$ small, the probability of this being less than any
negative constant can be made arbitrarily small by equicontinuity of
$n^{1/2} (E_n-E)m_j(W_i,\theta_n)g_j(X_i)$ in $\rho$.

\end{proof}

\begin{proof}[proof of Theorem \ref{rate_thm_bddweight}]
By the same argument that gives (\ref{mu_bound_eq}) in the proof of
Theorem \ref{rate_thm_alpha}, we will have, for $\theta$ with
$d_H(\theta,\Theta_0(P))$ smaller than some constant that does not
depend on $P$, there exists
a $\theta_0\in\Theta_0(P)$, $j_0$ and $g\in\mathcal{G}$ with $g(x)\ge
C_{\mathcal{G},1}I(\|x-x_0\|\le
[\eta/(4C)]^{1/\alpha}\|\theta-\theta_0\|^{1/\alpha}C_{\mathcal{G},2})$
such that
\begin{align*}
\mu_{P,j_0}(\theta,g)
=E_P\bar m_{j_0}(\theta,X_i,P)g(X_i)
\le -(\eta/4)\|\theta-\theta_0\| E_Pg(X_i).
\end{align*}
This, and the lower bound on $g$ gives
\begin{align*}
&\mu_{P,j_0}(\theta,g)
\le -(\eta/4)\|\theta-\theta_0\|
C_{\mathcal{G},1}%
\left\{[\eta/(4C)]^{1/\alpha}
\|\theta-\theta_0\|^{1/\alpha}C_{\mathcal{G},2}\right\}^{d_X} \eta  \\
&=-(\eta/4)\|\theta-\theta_0\|^{(\alpha+d_X)/\alpha}
C_{\mathcal{G},1}
\left\{[\eta/(4C)]^{1/\alpha}
C_{\mathcal{G},2}\right\}^{d_X} \eta  \\
&\le -(\eta/4)d_H(\theta,\Theta_0(P))^{(\alpha+d_X)/\alpha}
C_{\mathcal{G},1}
\left\{[\eta/(4C)]^{1/\alpha}
C_{\mathcal{G},2}\right\}^{d_X} \eta.
\end{align*}
Thus, the conditions of Lemma \ref{rate_lemma_bddweight} hold with
$\gamma=\alpha/(d_X+\alpha)$.

\end{proof}

The proof of Theorem \ref{rate_thm_bddweight} uses the following lemma,
which is analogous to Theorem \ref{rate_thm} for set estimates based on
variance weighted KS statistics.

\begin{lemma}\label{rate_lemma_bddweight}
Suppose that, for some positive constants $C$, $\gamma$, and $\delta$,
we have, for all $P\in\mathcal{P}$ and $\theta$ with
$d_H(\theta,\Theta_0(P))<\delta$,
\begin{align*}
\inf_{g,j}
\mu_{P,j}(\theta,g)
\le -C d_H(\theta,\Theta_0(P))^{1/\gamma}
\end{align*}
where the infemum is taken over $g\in\mathcal{G}$ and
$j\in\{1,\ldots,d_Y\}$.
Suppose that Assumptions
\ref{g_pos_assump}, \ref{covering_assump}, \ref{bdd_assump},
\ref{S_assump}, and \ref{consistency_assump}
hold, and that
the weight function $\omega_n(\theta,g)$
satsifies $\underline \omega\le \omega_n(\theta,g)\le \overline \omega$ for
some $0<\underline \omega \le \overline \omega <\infty$, and suppose that
$\hat c_n\to\infty$ with $\hat c_n/\sqrt{n}\to 0$.
Then, 
\begin{align*}
\inf_{P\in\mathcal{P}} P(\Theta_0(P)
  \subseteq \mathcal{C}_{n,\omega}(\hat c_n))
  \stackrel{n\to\infty}{\to} 1
\end{align*}
and, for some large $B$,
\begin{align*}
\sup_{P\in\mathcal{P}}
  P\left( \left(n/\hat c_n^2\right)^{\gamma/2}
    d_H(\mathcal{C}_n(\hat c_n),\Theta_0(P)) > B\right)
\stackrel{n\to\infty}{\to} 0.
\end{align*}

\end{lemma}
\begin{proof}%
First, note that, for all $j$, $\sup_{\theta,g}
\sqrt{n}|(E_n-E)m_j(W_i,\theta)g_j(X_i)|=\mathcal{O}_P(1)$ uniformly in
$P$ by Theorem 2.14.1 in \citet{vaart_weak_1996} (the constant function
equal to $\overline Y$ does not depend on $P$ and can be used as an
envelope function).  This, along with Assumption \ref{S_assump} and the
bound on the weight function, implies the first claim.

For the second claim, once $\Theta_0(P) \subseteq
\mathcal{C}_{n,\omega}(\hat c_n)$,
if $\left(n/\hat c_n^2\right)^{\gamma/2} 
  d_H(\mathcal{C}_n(\hat c_n),\Theta_0(P)) > B$,
there will be a
$\theta\in\mathcal{C}_{n,\omega}(\hat c_n)$ such that 
$d_H(\theta,\Theta_0(P)) > B \frac{\hat c_n^\gamma}{n^{\gamma/2}}$.
If $d_H(\mathcal{C}_{n,\omega}(\hat c_n),\Theta_0(P))<\delta$, which
happens with probability approaching one uniformly in
$P\in\mathcal{P}$ by arguments similar to the proof of Theorem
\ref{consistency_thm}, then,
for this $\theta$ and $P$, there will be a $g^*$ and $j^*$ such that,
for $n$ greater than some constant that does not depend on $P$,
$\mu_{P,j^*}(\theta,g^*)
\le -(C/2) \left(\hat c_n^2/n\right)^{1/2}B^{1/\gamma}$.
Since $\theta\in\mathcal{C}_{n,\omega}(\hat c_n)$, we will also have
$T_{n,\omega}(\theta)\le \hat c_n n^{-1/2}$, so that
$\hat \mu_{n,j^*}(\theta,g^*)
  \omega_{n,j^*}(\theta,g^*)\ge -\hat c_n n^{-1/2} K_{S,2}$.  By the
lower bound on the weight function, this implies
$\hat \mu_{n,j^*}(\theta,g^*)\ge -\hat c_n n^{-1/2}
  K_{S,2}/\underline\omega$.
Thus,
\begin{align*}
\sqrt{n}[\hat \mu_{n,j^*}(\theta,g^*)-\mu_{P,j^*}(\theta,g^*)]
\ge \hat c_n \left[-K_{S,2}/\underline \omega
+(C/2) B^{1/\gamma}\right].
\end{align*}
For $B$ large enough, the right hand side will go to infinity.  Since the
left hand side is $\mathcal{O}_P(1)$ uniformly in $P\in\mathcal{P}$, this
gives the desired result.

\end{proof}

\begin{proof}[proof of Theorem \ref{kernel_rate_upper_thm}]
Let $\theta_n$ and $\theta_0'(n)$ be as in the proof of Theorem
\ref{upper_rate_bdd_thm}, but with
$d_H(\theta_n,\Theta_0(P))
  =\varepsilon \left(\frac{\sqrt{\log n}}{\sqrt{n h_n^{d_X}}}
  \vee h_n^{\alpha} \right)$.

If $d_H(\mathcal{C}_n^{\text{kern}}(\hat c_n),\Theta_0(P))<
  \varepsilon \left(\frac{\sqrt{\log n}}{\sqrt{n h_n^{d_X}}}
  \vee h_n^{\alpha} \right)$,
then $\theta_n\notin \mathcal{C}_n^{\text{kern}}(\hat c_n)$ so that
$T^{\text{kern}}_{n,k,h_n}(\theta_n)\ge \hat c_n$.  Then, letting $K_{S,1}$
be as in Assumption \ref{S_assump}, we will have, for some $j$ and $x$,
$\frac{\sqrt{n h_n^{d_X}}}{\sqrt{\log n}}
\hat{\bar m}_j(x,\theta_n)\le -K_{S,1} \hat c_n$.  By Lemmas
\ref{kernel_scale_lemma} and \ref{kernel_var_lemma}, for large enough $a$
we will have, for some constant $K$,
\begin{align}\label{kern_consistency_eq}
\sup_{x\in\mathbb{R}^{d_X}, \theta\in\Theta}
\frac{\sqrt{n h_n^{d_X}}}{\sqrt{\log n}}
\left|\frac{(E_n-E_P)m_j(W_i,\theta)k((X_i-x)/h_n)}{E_nk((X_i-x)/h_n)}\right|
\le K
\end{align}
with probability approaching one (Lemma \ref{kernel_scale_lemma}
allows $E_Pk((X_i-x)/h_n)$ to be replaced by its sample analogue in
Lemma \ref{kernel_var_lemma}).  When
$T^{\text{kern}}_{n,k,h_n}(\theta_n)\ge \hat c_n$, we will have
\begin{align*}
&\frac{\sqrt{n h_n^{d_X}}}{\sqrt{\log n}}
\left[\frac{(E_n-E_P)m_j(W_i,\theta_n)k((X_i-x)/h_n)}{E_nk((X_i-x)/h_n)}
+\frac{E_P m_j(W_i,\theta_n)k((X_i-x)/h_n)}{E_nk((X_i-x)/h_n)} \right] \\
&=\frac{\sqrt{n h_n^{d_X}}}{\sqrt{\log n}}
  \hat{\bar m}_j(x,\theta_n)\le -K_{S,1} \hat c_n,
\end{align*}
so that, when (\ref{kern_consistency_eq}) holds, we will have
\begin{align*}
\frac{\sqrt{n h_n^{d_X}}}{\sqrt{\log n}}
\frac{E_P m_j(W_i,\theta_n)k((X_i-x)/h_n)}{E_nk((X_i-x)/h_n)}
\le -K_{S,1} \hat c_n + K.
\end{align*}
Appealing again to Lemma \ref{kernel_scale_lemma}, if $a$ is large enough,
this will imply
\begin{align*}
\frac{\sqrt{n h_n^{d_X}}}{\sqrt{\log n}}
\frac{E_P m_j(W_i,\theta_n)k((X_i-x)/h_n)}{E_Pk((X_i-x)/h_n)}
\le \frac{-K_{S,1} \hat c_n + K}{2}.
\end{align*}
Letting $\eta$ be as in Assumption \ref{support_assump}
letting $\varepsilon_1>0$ and $\varepsilon_2>0$ be such that $k(t)\ge
\varepsilon_1$ for $\|t\|\le \varepsilon_2$ and defining
$K_1=\eta \varepsilon_1 \varepsilon_2^{d_X}$, we have
$E_Pk((X_i-x)/h_n)\ge \varepsilon_1 P(\|X_i-x\|\le h_n\varepsilon_2)
\ge \eta\varepsilon_1 \varepsilon_2^{d_x}h_n^{d_X} = K_1h_n^{d_X}$ by
Assumption \ref{support_assump}, so that the above display implies
\begin{align*}
E_P m_j(W_i,\theta_n)k((X_i-x)/h_n)
\le K_1h_n^{d_X}\frac{-K_{S,1} \hat c_n + K}{2}
\frac{\sqrt{\log n}}{\sqrt{n h_n^{d_X}}}
= \frac{K_1\left(-K_{S,1} \hat c_n + K\right)}{2}
\frac{\sqrt{h_n^{d_X}}\sqrt{\log n}}{\sqrt{n}}.
\end{align*}
Let $\hat c_n$ be large enough so that
$K_1\left(-K_{S,1} \hat c_n + K\right)/2\le -\delta$ for some fixed
constant $\delta>0$.  Then the above display implies
\begin{align}\label{kernel_mean_bound}
E_P m_j(W_i,\theta_n)k((X_i-x)/h_n)
\le -\delta
\frac{\sqrt{h_n^{d_X}}\sqrt{\log n}}{\sqrt{n}}.
\end{align}
When this holds, the right hand side will be negative, so that, by Lemma
\ref{kernel_bias_lemma},
$h_n\le B [d_H(\theta_n,\Theta_0(P))]^{1/\alpha}$.
If $h_n^{\alpha}\ge \frac{\sqrt{\log n}}{\sqrt{n h_n^{d_X}}}$, this will
imply $h_n\le \varepsilon^{1/\alpha} B h_n$, which is a contradiction
for $\varepsilon$ small enough.

Now suppose $h_n^{\alpha}\le \frac{\sqrt{\log n}}{\sqrt{n h_n^{d_X}}}$.
By the same argument as in the proof of Theorem \ref{upper_rate_bdd_thm},
we have, for some constant $K_2$ that does not depend on $n$,
$\bar m_j(\theta_n,x)\ge -K_2 d_H(\theta_n,\Theta_0(P))$ so that, if
$h_n^{\alpha}\le \frac{\sqrt{\log n}}{\sqrt{n h_n^{d_X}}}$,
$\bar m_j(\theta_n,x)
  \ge -\varepsilon K_2 \frac{\sqrt{\log n}}{\sqrt{n h_n^{d_X}}}$
so that the left hand side of (\ref{kernel_mean_bound}) is greater
than or equal to
\begin{align*}
-\varepsilon K_2 \frac{\sqrt{\log n}}{\sqrt{n h_n^{d_X}}}
E_P k((X_i-x)/h_n)
\ge - \varepsilon K_2 \frac{\sqrt{\log n}}{\sqrt{n h_n^{d_X}}}
\overline f h_n^{d_X}
= - \varepsilon \overline fK_2
\frac{\sqrt{h_n^{d_X}}\sqrt{\log n}}{\sqrt{n}}
\end{align*}
so that (\ref{kernel_mean_bound}) implies
$\varepsilon \overline fK_2\ge \delta$, a contradiction for $\varepsilon$
small enough.

\end{proof}

The proof of Theorem \ref{kernel_rate_upper_thm} uses the lemmas stated
and proved below.

\begin{lemma}\label{kernel_scale_lemma}
Suppose that Assumption \ref{kernel_assump} holds, and that Assumption
\ref{support_assump} and part (iii) of Assumption
\ref{smoothness_upper_assump} hold, with the upper bound on the density in
the latter assumption uniform in $P\in\mathcal{P}$.
Then, for any $\varepsilon$, there exists an $a$ such that, if
$h_n^{d_X}n/\log n\ge a$ eventually,
\begin{align*}
\sup_{P\in\mathcal{P}}
P\left(
\sup_{x\in\text{supp}_P(X_i)}
\left|\frac{E_nk((X_i-x)/h)}{E_Pk((X_i-x)/h)}-1\right|>\varepsilon
\right)
\stackrel{n\to\infty}{\to} 0
\end{align*}
for all $\varepsilon>0$.
\end{lemma}
\begin{proof}%
We have
\begin{align*}
&\left|\frac{E_nk((X_i-x)/h_n)}{E_Pk((X_i-x)/h_n)}-1\right|
=\frac{\{E_P[k((X_i-x)/h_n)]^2\}^{1/2}}{E_Pk((X_i-x)/h_n)}
\left|\frac{(E_n-E_P)k((X_i-x)/h_n)}{\{E_P[k((X_i-x)/h_n)]^2\}^{1/2}}\right|
\\
&\le \overline k^{1/2}
\left|\frac{(E_n-E_P)k((X_i-x)/h_n)}{\{E_P[k((X_i-x)/h_n)]^2\}^{1/2}}\right|
\cdot \frac{1}{[E_Pk((X_i-x)/h_n)]^{1/2}}
\end{align*}
where $\overline k$ is an upper bound for the kernel function $k$.  By
Theorem \ref{sd_scale_thm}, 
\begin{align*}
\sup_{P\in\mathcal{P}}
P\left(
\sup_{x\in\text{supp}_P(X_i)}
\frac{\sqrt{n}}{\sqrt{\log n}}
\left|\frac{(E_n-E_P)k((X_i-x)/h_n)}{\{E_P[k((X_i-x)/h_n)]^2\}^{1/2}}\right|
> K
\right)
\to 0
\end{align*}
for large enough $K$ (the lower bound on the denominator follows from
Assumption \ref{support_assump}),
so the result will follow if we can show that
$[E_Pk((X_i-x)/h_n)]^{1/2}\sqrt{n}/\sqrt{\log n}$
can be made arbitrarily large by choosing $a$ large in the assumptions of
the lemma.  By Assumptions \ref{kernel_assump} and \ref{support_assump},
we have, for some $\delta>0$ and all $x$ on the support of $X_i$ under
$P$,
\begin{align*}
[n/(\log n)] E_Pk((X_i-x)/h_n)
\ge [n/(\log n)] \delta h_n^{d_X},
\end{align*}
and taking the square root of this expression gives something that can be
made arbitrarily large by choosing $a$ large.
\end{proof}

\begin{lemma}\label{kernel_var_lemma}
Suppose that Assumption \ref{kernel_assump} holds, and that Assumption
\ref{support_assump} and part (iii) of Assumption
\ref{smoothness_upper_assump} hold, with the upper bound on the density in
the latter assumption uniform in $P\in\mathcal{P}$.
Then, if $h_n^{d_X}n/\log n\ge a$ eventually for $a$ large enough, we will
have
\begin{align*}
\sup_{P\in\mathcal{P}}
P\left(
\sup_{x\in\text{supp}_P(X_i), \theta\in\Theta}
\frac{\sqrt{n h_n^{d_X}}}{\sqrt{\log n}}
\left|\frac{(E_n-E_P)m_j(W_i,\theta)k((X_i-x)/h_n)}{E_Pk((X_i-x)/h_n)}\right|
> B
\right)
\stackrel{n\to\infty}{\to} 0
\end{align*}
for some $B$.
\end{lemma}
\begin{proof}%
We have
\begin{align*}
&\frac{\sqrt{n h_n^{d_X}}}{\sqrt{\log n}}
\left|\frac{(E_n-E_P)m_j(W_i,\theta)k((X_i-x)/h_n)}{E_Pk((X_i-x)/h_n)}\right|
 \\
&=\frac{\sqrt{n}}{\sqrt{\log n}}
\left|\frac{(E_n-E_P)m_j(W_i,\theta)k((X_i-x)/h_n)}
  {\sqrt{var_P[m_j(W_i,\theta)k((X_i-x)/h_n)]}\vee
    \sqrt{h_n^{d_X}}}\right|  \\
&\cdot
\frac{\sqrt{h_n^{d_X}}\left\{\sqrt{var_P[m_j(W_i,\theta)k((X_i-x)/h_n)]}
  \vee \sqrt{h_n^{d_X}}\right\}}
     {E_Pk((X_i-x)/h_n)}  %
\end{align*}
Since
\begin{align*}
&var_P[m_j(W_i,\theta)k((X_i-x)/h_n)]
\le \overline Y E_P[k((X_i-x)/h_n)]^2
\le \overline Y \overline f 
\int_{t\in\mathbb{R}^{d_X}} [k((t-x)/h_n)]^2 \, dt  \\
&= h_n^{d_X}\int_{u\in\mathbb{R}^{d_X}} [k(u)]^2 \, du,
\end{align*}
the last line is bounded by a constant times
\begin{align*}
\frac{\sqrt{n}}{\sqrt{\log n}}
\left|\frac{(E_n-E_P)m_j(W_i,\theta)k((X_i-x)/h_n)}
  {\sqrt{var_P[m_j(W_i,\theta)k((X_i-x)/h_n)]}\vee
    \sqrt{h_n^{d_X}}}\right|
\cdot \frac{h_n^{d_X}}{E_Pk((X_i-x)/h_n)}.
\end{align*}
By Assumptions \ref{kernel_assump} and \ref{support_assump},
we have, for some $\delta>0$ and $x$ on the support of $X_i$ under $P$,
$E_Pk((X_i-x)/h_n) \ge \delta h_n^{d_X}$, so that this is bounded by
\begin{align*}
\frac{\sqrt{n}}{\sqrt{\log n}}
\left|\frac{(E_n-E_P)m_j(W_i,\theta)k((X_i-x)/h_n)}
  {\sqrt{var_P[m_j(W_i,\theta)k((X_i-x)/h_n)]}\vee
    \sqrt{h_n^{d_X}}}\right|
\cdot (1/\delta).
\end{align*}
The claim now follows from Theorem \ref{sd_scale_thm}, with
$\sqrt{h_n^{d_X}}$ playing the role of the cutoff point $\sigma_n$.

\end{proof}

\begin{lemma}\label{kernel_bias_lemma}
Suppose that Assumptions \ref{support_assump},
\ref{smoothness_upper_assump} and \ref{kernel_assump} hold.
Let $\theta_0$ be as in Assumption
\ref{smoothness_upper_assump} and let $\theta_n$ be a sequence in
$\Theta\backslash\Theta_0(P)$ converging to $\theta_0$.  Then, for some
constant $B$ that does not depend on $n$ and some $N\in\mathbb{N}$,
$E_Pm_j(W_i,\theta_n)k((X_i-x)/h)$ will be nonnegative for
$h_n\ge B [d_H(\theta_n,\Theta_0(P))]^{1/\alpha}$ and $n\ge N$ for $x$ on
the support of $X_i$.
\end{lemma}
\begin{proof}%
Let $b_n=d_H(\theta_n,\Theta_0(P))$.
By an argument similar to the one leading up to Equation
(\ref{obj_bound_eq}), we will have, for each $j$,
\begin{align*}
\bar m_j(\theta_n,x,P)
\ge - C b_n  %
+\eta \min_{x_0\in \mathcal{X}_0(\theta_0'(n))}
(\|x-x_0\|^{\alpha}\wedge \eta)
\end{align*}
for some $C$ that depends only on the bound on the derivative $\bar
m_{\theta,j}(\theta,x,P)$ in Assumption \ref{smoothness_upper_assump} and
some $\theta_0'(n)\in \Theta_0(P)$.  Thus, %
for $x$ such
that $\bar m_j(\theta_n,x,P)\le C b_n$, we will have, for some
$x_0\in \mathcal{X}_0(\theta_0'(n))$, %
$C b_n \ge - C b_n
+\eta (\|x-x_0\|^{\alpha}\wedge \eta)$
so that 
$2C b_n
\ge \eta (\|x-x_0\|^{\alpha}\wedge \eta)$.
For $b_n$ small enough, this implies that $\|x-x_0\|\le
(2Cb_n/\eta)^{1/\alpha}$.  This means that, letting $K$ be a bound for the
number of elements in $\mathcal{X}_0(\theta_0'(n))$ and $\overline f$
an upper bound for the density of $X_i$,
\begin{align}\label{small_mean_bound}
P\left(\bar m_j(\theta_n,X_i,P)\le C b_n\right)
\le K\overline f (2Cb_n/\eta)^{d_X/\alpha}.
\end{align}
This, and the lower bound on $\bar m_j(\theta_n,x,P)$ imply, letting
$\overline k$ be an upper bound on the kernel $k$,
\begin{align*}
&E_P m_j(W_i,\theta_n)k((X_i-x)/h_n)
  I(\bar m_j(\theta_n,X_i,P)\le C b_n)
\ge - \overline k C b_n
  P\left(\bar m_{\theta,j}(\theta,X_i,P)\le C b_n\right)  \\
&\ge - \overline k C b_n \cdot K\overline f (2Cb_n/\eta)^{d_X/\alpha}.
\end{align*}

We also have, for $x$ on the support of $X_i$, letting $\varepsilon$ and
$K_1$ be such that $k(t)\ge K_1$ for $\|t\|\le\varepsilon$,
\begin{align*}
&E_P m_j(W_i,\theta_n)k((X_i-x)/h_n)
  I(\bar m_j(\theta_n,X_i,P)> C b_n)  \\
&\ge C b_n E_P k((X_i-x)/h_n) I(\bar m_j(\theta_n,X_i,P)> C b_n)  \\
&\ge K_1 C b_n E_P I(\|(X_i-x)/h_n\|\le \varepsilon)
  I(\bar m_j(\theta_n,X_i,P)> C b_n)  \\
&\ge K_1 C b_n[P(\|(X_i-x)/h_n\|\le \varepsilon)
- P(\bar m_j(\theta_n,X_i,P)\le C b_n)]  \\
&\ge K_1 C b_n[\eta \varepsilon^{d_X} h_n^{d_X}
- K\overline f (2Cb_n/\eta)^{d_X/\alpha}].
\end{align*}
The last inequality follows from Assumption \ref{support_assump} and 
from the inequality (\ref{small_mean_bound}) above (here the two $\eta$s
come from different conditions, but they can be chosen to be the same by
decreasing one).  Combining this with the bound in the previous display
gives
\begin{align*}
&E_P m_j(W_i,\theta_n)k((X_i-x)/h_n)
\ge K_1 C b_n[\eta \varepsilon^{d_X} h_n^{d_X}
- K\overline f (2Cb_n/\eta)^{d_X/\alpha}]
- \overline k C b_n \cdot K\overline f (2Cb_n/\eta)^{d_X/\alpha}  \\
&= b_n (K_2 h_n^{d_X} - K_3 b_n^{d_X/\alpha})
\end{align*}
where
$K_2=K_1C\eta \varepsilon^{d_X}$
and
$K_3=K_1 C K\overline f (2C/\eta)^{d_X/\alpha}
  + \overline k C K \overline f (2C/\eta)^{d_X/\alpha}$
are both positive constants that do not depend on $n$.  For
$h_n\ge (K_3/K_2)^{1/d_X} b_n^{1/\alpha}$, this will be nonnegative.

\end{proof}

\bibliography{library}

\end{document}